\documentclass[DIV12]{scrartcl}

\usepackage[utf8]{inputenc}

\usepackage{microtype}
\usepackage{xspace}
\usepackage{amsfonts}
\usepackage{amssymb}
\usepackage{mathtools}
\usepackage{booktabs}
\usepackage{bm}
\usepackage{tikz}
\usetikzlibrary{automata}
\usepackage{csquotes}
\usepackage{stmaryrd}
\usepackage{colonequals}
\usepackage{tabularx}
\usepackage[breaklinks=true]{hyperref}
\usepackage{amsthm}

\newtheorem{theorem}{Theorem}[section]
\newtheorem{definition}[theorem]{Definition}
\newtheorem{example}[theorem]{Example}

\newtheorem{proposition}[theorem]{Proposition}
\newtheorem{lemma}[theorem]{Lemma}

\newcommand{\OR}{\Gamma} %

\newcommand{\PosRAbold}{\textrm{PosRA}\xspace}

\newcommand{\PosRA}{{\textrm{PosRA}}\xspace}
\newcommand{\Plex}{{\textrm{PosRA}}$_{\lex}$\xspace}
\newcommand{\Pnoprod}{{\textrm{PosRA}}$_{\mathrm{no}\times}$\xspace}
\newcommand{\Pgen}{{\textrm{PosRA}}$_{\gen}$\xspace}
\newcommand{\PosRAagg}{{\textrm{PosRA}}$^{\mathrm{acc}}$\xspace}
\newcommand{\PosRAagggby}{{\textrm{PosRA}}$^{\mathrm{accGBy}}$\xspace}
\newcommand{\Plexacc}{{\textrm{PosRA}}$_{\lex}^{\mathrm{acc}}$\xspace}

\newcommand{\Pnoprodacc}{{\textrm{PosRA}}$^{\mathrm{acc}}_{\mathrm{no}\times}$\xspace}
\newcommand{\posRAagg}{\PosRAagg}
\newcommand{\PosRAacc}{\PosRAagg}
\newcommand{\PosRAaccgby}{\PosRAagggby}

\newcommand{\Nat}{\mbox{$\mathbb{N}$}}

\newcommand{\ptw}{\dir}
\newcommand{\dir}{{\texttt{DIR}}\xspace}
\newcommand{\gen}{\ptw}

\newcommand{\lex}{{\texttt{LEX}}\xspace}
\newcommand{\cat}{{\texttt{CAT}}\xspace}

\newcommand{\possbold}{\texttt{POSS}\xspace}
\newcommand{\poss}{{\texttt{POSS}}\xspace}
\newcommand{\cert}{{\texttt{CERT}}\xspace}

\newcommand{\cupgen}{\cup}
\newcommand{\cupcat}{\cup_\cat}

\newcommand{\deft}[1]{\emph{#1}}

\makeatletter
\newcommand*{\defeq}{\mathrel{\rlap{%
  \raisebox{0.3ex}{$\m@th\cdot$}}%
  \raisebox{-0.3ex}{$\m@th\cdot$}}%
  =}
\makeatother

\renewcommand{\leq}{\leqslant}
\renewcommand{\geq}{\geqslant}
\renewcommand{\phi}{\varphi}
\renewcommand{\epsilon}{\varepsilon}

\newcommand\restr[2]{{%
  \kern-\nulldelimiterspace %
  #1 %
  _{|#2} %
  }}

\newcommand{\singleton}[1]{[#1]}
\newcommand{\ordern}[1]{[{\leq}#1]}

\newcommand{\card}[1]{\left|#1\right|}

\newcommand{\ID}{\mathit{ID}}
\newcommand{\id}{\mathit{id}}

\newcommand{\calD}{\mathcal{D}}

\newcommand{\calF}{\mathcal{F}}

\newcommand{\calI}{\mathcal{I}}

\newcommand{\calM}{\mathcal{M}}

\newcommand{\calS}{\mathcal{S}}

\newcommand{\arity}[1]{\mathop{\mathrm{a}}(#1)}

\newcommand{\axiom}[1]{\textsf{#1}}

\newcommand{\dupelim}{\mathrm{dupElim}}
\newcommand{\pw}{\mathit{pw}}

\newcommand{\langlem}{}
\newcommand{\ranglem}{}

\newcommand{\accum}{\mathrm{accum}}
\newcommand{\accumgby}{\mathrm{accumGroupBy}}
\newcommand{\concat}{\mathrm{concat}}
\newcommand{\pr}{\mathrm{pr}}

\renewcommand{\a}{\mathsf{a}}
\renewcommand{\b}{\mathsf{b}}
\newcommand{\n}{\mathsf{n}}
\newcommand{\e}{\mathsf{e}}
\newcommand{\f}{\mathsf{f}}
\newcommand{\ii}{\mathsf{i}}
\renewcommand{\l}{\mathsf{l}}
\renewcommand{\r}{\mathsf{r}}
\newcommand{\s}{\mathsf{s}}

\newcommand{\NN}{\mathbb{N}}
\newcommand{\ZZ}{\mathbb{Z}}

\newcommand{\tikzm}[1]{\tikz[overlay,remember picture]{\node (#1) {};}}
\newcommand{\tikzmd}[2]{\tikz[overlay,remember picture,thick]{\draw[->] (#1) -- (#2);}}

\newcommand{\bag}{\mathrm{bag}}

\newcommand{\BALG}{\text{BALG}}
\newcommand{\BQL}{\mathcal{BQL}}

\newcommand{\NRLaggr}{\mathcal{NRL}^{\mathrm{aggr}}}
\newcommand{\PomAlg}{\mathcal{P}\text{om-}\mathcal{A}\text{lg}}
\newcommand{\PomAlgEps}{\PomAlg_{\epsilon_n}}

\newcommand{\tmem}[2]{#1.#2}

\title{Computing Possible and Certain Answers\\over Order-Incomplete Data}
\date{}

\hypersetup{
    colorlinks,
    linkcolor={red!50!black},
    citecolor={blue!50!black},
    urlcolor={blue!30!black}
}

\author{
\begin{tabular}[t]{c}
Antoine Amarilli \\
  {\normalfont LTCI, Télécom ParisTech, Université Paris-Saclay;}\\
  {\normalfont Paris, France} \\
  {\normalfont \texttt{antoine.amarilli@telecom-paristech.fr}} \\[0.5em]
Mouhamadou Lamine Ba \\
  {\normalfont Universit\'e Alioune Diop de Bambey;}\\
  {\normalfont Bambey, Senegal} \\
  {\normalfont \texttt{mouhamadoulamine.ba@uadb.edu.sn}} \\[0.5em]
Daniel Deutch \\
  {\normalfont Blavatnik School of Computer Science, Tel Aviv University;}\\
  {\normalfont Tel Aviv, Israel} \\
  {\normalfont \texttt{danielde@post.tau.ac.il}} \\[0.5em]
Pierre Senellart \\
{\normalfont DI ENS, ENS, CNRS, PSL University; Paris, France} \\
{\normalfont \& Inria; Paris, France} \\
  {\normalfont \texttt{pierre@senellart.com}} \\[0.5em]
\end{tabular}
}

\begin{document}

\maketitle

\begin{abstract}
  This paper studies the complexity of query evaluation for databases
whose relations are partially ordered; the problem commonly arises when combining or transforming ordered data from multiple sources.
We focus on queries in a useful fragment of~SQL, namely positive relational algebra with aggregates,
whose bag semantics we extend to the partially ordered setting.
Our semantics leads to the study of two main computational problems:
the possibility and certainty of query answers. We show that these
problems are respectively NP-complete and coNP-complete, but identify tractable cases depending on the query operators or input partial
orders. We further introduce a duplicate elimination operator and study its effect on the complexity results.

\end{abstract}

\section{Introduction}
\label{sec:introduction}
Many applications need to combine and transform ordered data from
multiple sources. Examples include sequences of readings from multiple sensors, or
log entries from different applications or machines, that need to be
combined to form a complete picture of events; rankings of
restaurants and hotels based on various criteria (relevance,
preference,
or customer ratings); and concurrent edits of shared documents,
where the order of contributions made by different users needs to be
merged. Even if the order of items from each individual source is
usually known, the order of items across sources is often {\em
	uncertain}.
For instance, even
when sensor readings or log entries are provided with time\-stamps,
these may be ill-synchronized across sensors or machines; rankings of hotels
and restaurants may be biased by different preferences of different
users; concurrent contributions to documents may be ordered in
multiple reasonable ways. We say that the resulting information is
\emph{order-incomplete}.

This paper studies query evaluation over order-incomplete data in a
relational setting \cite{AHV-1995}. We focus on the running example of restaurants
and hotels from a travel website, ranked according to a proprietary
function.
An example query would ask for the ordered list of restaurant--hotel
pairs
such that the restaurant and hotel are in the same district,
or such that the restaurant features a particular cuisine, and 
may further apply order-dependent operators to the result,
e.g., limiting the output to the top-$k$ such pairs, or aggregating a relevance
score.
To evaluate such queries, the initial order on the hotels and
restaurants must be \emph{preserved} through transformations.
Furthermore, as we do not know how the proprietary order is defined,
the result of transformations may become \emph{uncertain}; hence, we
need to represent all \emph{possible} results that can be obtained
depending on the underlying order.

Our approach is to handle this uncertainty through the classical notions
of {\em possible and certain answers}.
We say that there is a 
\emph{certain answer} to the query when there is only one possible order on
query results, or only one accumulation result,
which is obtained no matter the order on the input and in intermediate
results. In this case, it is useful to compute the certain answer,
so that the user can then browse through
the ordered query results (as is typically done when there is no
uncertainty, using constructs such as SQL's \texttt{ORDER BY}). 
Certain answers can arise even in non-trivial cases 
where the combination of input data admits 
many possible
orders: consider user queries that select only a small interesting subset of the
data (for which the ordering happens to be certain), or a short summary
obtained through accumulation over large data. In many other cases, the
different orders on input data or the uncertainty caused by the query may
lead to several \emph{possible answers}. In this case, it is still of
interest (and non-trivial) to verify whether an answer is
possible, e.g., to check whether a given ranking of hotel--restaurant pairs is consistent with a combination of other rankings (the latter done through a query). Thus, we study the problems of deciding whether a given answer is
\emph{certain}, and whether it is \emph{possible}.

Our main contributions may be summarized as follows.

\paragraph*{Model and Problem Definition (Sections \ref{sec:model},
\ref{sec:posscertdef})}

Our work focuses on {\em bag semantics}, where a tuple may appear multiple
times. Note that in the context of (partially) ordered relations, this means
that multiple copies of the same tuple may appear in different ``positions" in
the order. For example, if we integrate multiple rankings of restaurants, then
the same restaurant appears multiple times in different positions. We capture
this model by a notion of {\em po-relations} (partially ordered) relations. A
po-relation is essentially a relation accompanied with a partial order over its
tuples; a technical subtlety is that each tuple is associated with an identifier
(presumably internal and automatically generated), so that we have a way
of referring to each tuple occurrence in the partial order (see further discussion in
the problem definition below).

We then introduce a query language for partially ordered
data. Our language design is guided by the goal of supporting SQL evaluation in
presence of such data, and as such we focus on defining a semantics for an
important fragment of SQL~-- namely positive relational algebra with aggregates.
The semantics is ``faithful" to SQL in the sense that, if we ignore
order, then we get
the standard SQL semantics. The notion corresponding to aggregation in our context is LISP-like accumulation,
whose semantics we extend to account for {\em partial} orders.

We view partially ordered relations as a concise representation of a set of
possible worlds, namely, the linear extensions of the partial orders over
the {\em underlying tuples of the relation}. For example, a linear
extension is a ranked list of restaurant cuisines, where a cuisine may appear
multiple times in the list. Note that in each such linear extension, the tuples
appear \emph{without} their respective internal identifiers which, as mentioned
earlier, were only present in the po-relation as a technical tool.

Our definitions lead to a possible worlds semantics for query evaluation, and to two
natural problems:
whether a candidate answer -- i.e., a ranked list of tuples (restaurants, cuisines, etc.)~-- is \emph{possible}, i.e., is obtained for some possible world,
and whether it is \emph{certain}, i.e., is obtained for every possible world. Here again, note that in a candidate answer a tuple may appear multiple times, and naturally it appears without identifiers (which, as mentioned above, are internal and are unknown to the user). We formally define these two problems for our settings, and then embark on a study of their complexity.

\paragraph*{Complexity in the General Case (Section \ref{sec:posscert})} We first study the possibility and certainty problems without any restrictions on the input database.
As usual in data management, given that queries are typically much smaller than databases,
we study the {\em data complexity} of the problems, i.e., the complexity when the query is fixed.
For our general definition of po-relations, we show that deciding whether an answer is possible is NP-complete,
even without accumulation, and even for some very simple queries and
input relations. In a particular case where we assume no
duplicates -- i.e., where tuples are uniquely identified, which means we
are essentially back to the set semantics -- possibility of an answer is
in PTIME without accumulation, but is again NP-complete with accumulation.   
As for certainty, the problem can be decided in polynomial time in the case with no
accumulation, but it is coNP-complete for queries with accumulation
(even if we assume no duplicates in the input). Faced by the general intractability of
the possibility and certainty problems, in the rest of the paper we search for
restricted cases for which tractability holds.

\paragraph*{Tractable Cases for Possibility Without Accumulation (Section~\ref{sec:fpt2})} Even though possibility is NP-hard even without accumulation, we identify realistic cases where it is in fact tractable.
In particular, we show that if the input relations are totally ordered
then possibility is in PTIME for queries using a subset of our
operators (all except the \emph{direct product}).
Assuming more severe restrictions on the query language,
we further show tractability when some of the relations are (almost)
ordered and the rest are (almost) unordered, as 
formalized via a newly introduced notion of \emph{ia-width}.

\paragraph*{Tractable Cases with Accumulation (Section~\ref{sec:fpt})} With accumulation, the certainty problem becomes intractable as well.
Yet we show that if accumulation is captured by a finite cancellative monoid (in particular, if it is performed in a finite group), then certainty can again be decided in polynomial time. Further, we revisit the tractability results for possibility from Section~\ref{sec:fpt2} and show that they extend to queries with accumulation under certain restrictions on the accumulation function.  

\paragraph*{Language Extensions (Section \ref{sec:extensions})} We then study two extensions to our language, which are the counterparts of common SQL operators.
The first is \emph{group-by}, which allows us to group tuples for accumulation (as is done for aggregation in SQL with \texttt{GROUP BY}); we revisit our complexity results in its presence.     
The second is \emph{duplicate elimination}: keeping a single representative of identical tuples, as in SQL with \texttt{SELECT DISTINCT}.
In presence of order, it is challenging to design a semantics for this operator,
and we discuss both semantic and complexity issues that arise from different possible definitions. 

\bigskip

We compare our model and results to related work in Section~\ref{sec:compare}, and conclude in Section~\ref{sec:conclusion}. 

This article is an extended version of the conference
paper~\cite{amarilli2017possible}. In contrast
with the conference paper~\cite{amarilli2017possible}, all proofs are included here. We also
discovered a bug in the proof of Theorem~22
of that paper~\cite{amarilli2017possible}, that also impacts Theorems~19 and~30
of~\cite{amarilli2017possible}. Consequently, these results are omitted
in the present paper.

\section{Data Model and Query Language}\label{sec:model}
We denote by $\NN$ the set of nonnegative natural numbers and by $\NN_{>0}$ the
set of positive natural numbers, i.e., $\NN_{>0} \colonequals \NN \setminus
\{0\}$.
We fix a countable set of values $\calD$
that includes $\Nat$ and infinitely many values not in $\Nat$.
A \emph{tuple~$t$ over~$\calD$} of \emph{arity}~$\arity{t}$ is an
element of~$\calD^{\arity{t}}$, denoted $\langle v_1, \dots,
v_{\smash{\arity{t}}}\rangle$: for $1 \leq i \leq \arity{t}$, we write
$\tmem{t}{i}$
to refer to~$v_i$. The simplest notion of ordered relations are then
\emph{list relations}~\cite{colby1994query,colby1994concepts}: a list relation
of arity~$n \in \Nat$
is an ordered list of tuples over~$\calD$ of arity~$n$ (where
the same tuple may appear multiple times).
List relations impose a single order over tuples, but when one
combines (e.g., unions) them, there may be multiple plausible ways
to order the results. 

We thus introduce \emph{partially
ordered relations} (\emph{po-relations}). A po-relation  $\OR =
(\ID, T, <)$ of arity~$n \in \Nat$ consists of a finite set of
\emph{identifiers} $\ID$ (chosen from some infinite set
closed under the Cartesian product, e.g., we can use tuples of natural numbers), a \emph{strict partial order} $<$ on~$\ID$, and a (generally
non-injective) mapping $T$ from $\ID$ to~$\calD^n$.
The \emph{domain} of~$\OR$ is the subset of values of~$\calD$ that occur in the
image of~$T$.
The actual identifiers in~$\ID$ do not matter, but
we need them to
refer to occurrences of the
same tuple value. Hence, we always
consider po-relations \emph{up to isomorphism}, where $(\ID, T, <)$ and $(\ID',
T', {<'})$ are \emph{isomorphic} iff there is a bijection $\phi: \ID \to \ID'$
such that $T'(\phi(\id)) = T(\id)$ for all $\id \in \ID$, and $\phi(\id_1) {<'}
\phi(\id_2)$ iff $\id_1 < \id_2$ for all $\id_1, \id_2 \in \ID$.

A special case of po-relations
are \emph{unordered po-relations} (or \emph{bag relations}), where $<$ is empty:
we denote them $(\ID, T)$.
Another special case is that of \emph{totally ordered po-relations}, where~$<$ is a
total order.

The point of po-relations is to represent \emph{sets} of list relations.
Formally, a \emph{linear extension} $<'$ of~$<$ is a total order on~$\ID$ such
that ${<} \subseteq {<'}$, i.e., for each $x<y$ we have $x <' y$. The 
{\em possible worlds} $\pw(\OR)$ of $\OR$ are then defined as follows:
for each linear extension ${<'}$ of~$<$, writing $\ID$ as $\id_1 <' \cdots <'
\id_{\card{\ID}}$,
the list relation $(T(\id_1), \ldots, T(\id_{\card{ID}}))$ is
in~$\pw(\OR)$. As $T$ is generally not injective, two different linear extensions may yield the
same list relation. Note that each such linear extension ``strips away" the
identifiers and includes only the tuples.
For instance, if $\OR$ is unordered, then $\pw(\OR)$ consists of all
permutations of the tuples of~$\OR$; and if~$\OR$ is totally ordered then
$\pw(\OR)$ contains exactly one possible world.

Po-relations can thus model uncertainty over the \emph{order} of
tuples.
However, note that they cannot model uncertainty on tuple \emph{values}. Specifically, let
us define the \emph{underlying bag relation} of a po-relation $\OR = (\ID, T,
<)$ as $(\ID, T)$. Unlike order, this underlying bag relation is always
certain.

We extend some classical notions from partial
    order theory to po-relations.

Letting $\OR=(\ID,T,<)$  be a po-relation, an
  \deft{order ideal} of~$\OR$ is a subset $S \subseteq \ID$ such that, for all
  $x, y \in \ID$, if $x < y$ and $y \in S$ then $x \in S$.
An \deft{antichain}~\cite{schroder2003ordered} of~$\OR$ is a set
$A\subseteq\ID$ of pairwise
incomparable tuple identifiers.
The
\deft{width} of $\OR$ is the size of its largest
antichain, and the
\deft{width} of a po-database is the maximal width of its po-relations.
In particular, totally ordered po-relations have width~$1$, and
unordered po-relations have a width equal to their number of
tuples; the width of a po-relation can be computed in
polynomial time~\cite{fulkerson1955note}.

A \deft{chain partition}
of~$\OR$ is a partition 
$\ID = \Lambda_1 \sqcup \cdots \sqcup \Lambda_n$
such that the restriction of~$<$ to each $\Lambda_i$ is a total
order: we call each $\Lambda_i$ a \emph{chain}.
Note that $<$ may include comparability relations across chains,
i.e., relating elements in $\Lambda_i$ to elements in $\Lambda_j$ for $i \neq j$.
The \deft{width} of the chain partition is~$n$. By Dilworth's
theorem~\cite{dilworth1950decomposition, fulkerson1955note},
the width~$w$ of~$\OR$ is the smallest possible width of a chain partition
of~$\OR$; furthermore, given $\OR$, we can compute in polynomial time both its
width~$w$ and a
chain partition of~$\OR$ of width~$w$.

\subsection{\PosRA: Queries Without Accumulation}
\label{sec:posra}

We now define a bag semantics for
{\em positive relational algebra} operators, to manipulate po-relations with queries. The positive relational
algebra,
written \PosRA, is a standard query language for relational
data~\cite{AHV-1995}.
We will extend \PosRA with
{\em accumulation} in Section~\ref{sec:posraacc},
and add further operations in Section~\ref{sec:extensions}.
Each \PosRA operator applies to po-relations and computes a new po-relation; we
present them in turn.

\smallskip

The \emph{selection} operator restricts
the relation to a subset of its tuples, and the order is
the restriction of the input order.
The \emph{tuple predicates}
allowed in selections are Boolean combinations of
equalities and inequalities, which involve constant values in~$\calD$ and tuple
attributes written as~$.i$ for $i \in \NN_{>0}$. For instance, the selection
$\sigma_{.1 \neq \text{``a''} \wedge .2 \neq .3}$ selects tuples whose first
attribute is equal to the constant~``a'' and whose second attribute is different
from their third attribute.

\begin{description}
  \item[\axiom{selection}:] For any po-relation
  $\OR = (\ID, T,<)$
  and tuple predicate $\psi$,
  we define the selection $\sigma_{\psi}(\OR) \defeq (\ID', T_{|\ID'},
    <_{|\ID'})$,
  where $\ID' \defeq \{\id \in \ID \mid \psi(T(\id))\text{~holds}\}$.
\end{description}
\smallskip
The \emph{projection} operator changes tuple values in the usual way, but
keeps the original tuple ordering in the result, and retains all copies of duplicate
tuples (following our \emph{bag semantics}).
\begin{description}
  \item[\axiom{projection}:] For a po-relation
    $\OR = (\ID,T,<)$ and attributes $A_1, \ldots, A_n$,
    we define the projection $\Pi_{A_1, \ldots, A_n}(\OR) \defeq
    (\ID,T', {<})$, where $T'$ maps each $\id \in \ID$ to
    $\Pi_{A_1, \ldots, A_n}(T(\id)) \colonequals \langle T(\id).A_1, \ldots,
    T(\id).A_n\rangle$.
\end{description}
\smallskip
As for \emph{union}, we impose the minimal order constraints that are
compatible with those of the inputs.
We use the
{\em parallel composition}~\cite{posets} of two partial orders $<$
and $<'$ on disjoint sets $\ID$ and $\ID'$, i.e., the partial order $<''
\defeq ({<} \cup 
{<'})$
on $\ID \cup \ID'$. Note that $<''$ is the same order as $<$ on~$\ID$ and as
$<'$ on~$\ID'$, and that all elements from~$\ID$ are incomparable to all
elements from~$\ID'$.
\begin{description}
  \item[\axiom{union}:] Let $\OR = (\ID,T,<)$ and $\OR' = (\ID',T',<')$ be two po-relations of
    the same arity. We assume that the identifiers of~$\OR'$ have been renamed
    if necessary to ensure that $\ID$ and~$\ID'$ are disjoint.
    We then define $\OR \cup \OR'
    \defeq (\ID \cup \ID', T'', ({<} \cup {<'}))$, where $T''$
    maps
    $\id \in \ID$ to $T(\id)$ and
    $\id' \in \ID'$ to $T'(\id')$.
\end{description}
\smallskip
The union result $\OR \cup \OR'$ does not depend on how we renamed~$\OR'$,
i.e., it is unique up to isomorphism.
Our definition also
implies that $\OR \cup \OR$ is different from~$\OR$, as per bag semantics.
In particular, when $\OR$ and $\OR'$ have only one possible world,
$\OR \cup \OR'$ usually does not.

We next introduce two possible product operators. First,
as in~\cite{stanley1986enumerative},
the
\emph{direct product} ${<_{\ptw}} \defeq
({<} \times_\ptw {<'})$ of two partial orders $<$ and $<'$ on
sets $\ID$ and $\ID'$ is defined by $(\id_{1}, \id_{1}') <_{\ptw} (\id_{2},
\id_{2}')$ 
iff $\id_{1} < \id_{2}$ and $\id_{1}' <' \id_{2}'$
for each $(\id_{1},  \id_{1}'), (\id_{2}, \id_{2}') \in \ID \times \ID'$.
We define the \emph{direct product} operator over po-relations
accordingly: two identifiers in the product are comparable only
if \emph{both components} of both identifiers compare in the same
way.
\begin{description}
  \item[\axiom{direct product}:] For any po-relations $\OR = (\ID, T,<)$ and
    $\OR' = (\ID', T',{<'})$,
    remembering that the set of possible identifiers is closed under product,
    we let $\OR \times_\gen \OR' \defeq
    (\ID \times \ID', T'', \allowbreak {< \times_{\ptw} <'})$,
    where $T''$ maps each $(\id, \id') \in \ID \times \ID'$ to the
    \emph{concatenation} $\langle T(\id), T'(\id') \rangle$.
\end{description}
\smallskip
Again, the direct product result often has multiple possible worlds even when
inputs do not.

The second product operator uses the \emph{lexicographic product} (or \emph{ordinal
product}~\cite{stanley1986enumerative})
${<_\lex} \defeq ({<} \times_\lex {<'})$
of two partial orders $<$ and $<'$,
defined by $(\id_{1}, \id_{1}') <_{\lex} (\id_{2}, \id_{2}')$
iff
    either $\id_{1} < \id_{2}$, or $\id_{1} = \id_{2}$ and $\id_{1}' <'
    \id_{2}'$,
for all $(\id_{1},
\id_{1}'), (\id_{2}, \id_{2}') \in \ID \times \ID'$.
\begin{description}
  \item[\axiom{lexicographic product}:]
    For any po-relations $\OR = (\ID, T,<)$ and $\OR' = (\ID', \allowbreak T', {<'})$,
    we define
    $\OR \times_\lex \OR'$ as $(\ID
    \times \ID', T'', < \times_\lex <')$ with $T''$ defined like for the direct
    product.
\end{description}
\smallskip
Last, we define the \emph{constant expressions} that we allow.
\begin{description}
  \item[\axiom{constant expressions}:] 
    \begin{minipage}[t]{.757\linewidth}\begin{itemize}
  	\item for any tuple~$t$, the singleton po-relation $\singleton{t}$
  	has only one tuple with value~$t$;
  	\item for any $n \in \Nat$, the po-relation $\ordern{n}$ has arity~$1$
  	and has
  	$\pw(\ordern{n}) = \{(1, \ldots, n)\}$.
  \end{itemize}
\end{minipage}
\end{description}

We have now defined a semantics on po-relations for each \PosRA
operator. We define a \emph{\PosRA query} in the expected way, as a query built from these operators and from relation names. Calling
\emph{schema} a set $\calS$ of relation names and arities, with an
attribute name for each position of each relation, we define a
\emph{po-database} $D$ as having a po-relation
of the correct arity
for each relation name $R$ in~$\calS$. For a po-database $D$ and a \PosRA
query $Q$, we denote by $\card{Q}$ the number of symbols of~$Q$, and we 
denote by $Q(D)$ the po-relation obtained by evaluating
$Q$ over $D$.

\begin{figure}
\noindent\begin{minipage}[b]{.73\linewidth}
{\renewcommand{\tabcolsep}{1.5pt}
{\small\noindent\begin{tabular}{l@{~~}l@{~~}l}
  $\!\!$\begin{tabular}[b]{l@{~}c}
\toprule
\textit{restname} & \textit{district} \\
\midrule
Gagnaire & 8\quad\raisebox{0.9em}{\tikzm{froma}} \\
TourArgent & 5\quad\raisebox{-0.3em}{\tikzm{toa}}\tikzmd{froma}{toa} \\
\bottomrule\\[-.8em]
\multicolumn{2}{c}{{(a) $\mathit{Restaurant}$ table}}\\
\end{tabular}
&
\begin{tabular}[b]{l@{~}c}
\toprule
\textit{hotelname} & \textit{district}   \\
\midrule
Mercure & 5\phantom{2}\quad\raisebox{0.9em}{\tikzm{fromb}} \\
Balzac & 8\phantom{2}\quad\raisebox{0.9em}{\tikzm{dummy}} \\
Mercure & 12\quad\raisebox{-0.3em}{\tikzm{tob}}\tikzmd{fromb}{tob} \\
\bottomrule\\[-.8em]
\multicolumn{2}{c}{{(b) $\mathit{Hotel}$ table}}
\end{tabular}
&
\begin{tabular}[b]{l@{~}c}
\toprule
\textit{hotelname} & \textit{district}   \\
\midrule
Balzac & 8\phantom{2}\quad\raisebox{0.9em}{\tikzm{fromc}} \\
Mercure & 5\phantom{2}\quad\raisebox{0.9em}{\tikzm{dummyd}} \\
Mercure & 12\quad\raisebox{-0.3em}{\tikzm{toc}}\tikzmd{fromc}{toc} \\
\bottomrule\\[-.8em]
\multicolumn{2}{c}{{(c) $\mathit{Hotel}_2$ table}}
\end{tabular}
\end{tabular}}
  } \vspace{-.6em}\caption{Running example: Paris restaurants and hotels} \label{fig:examplerels}
\end{minipage}
\hfill\begin{minipage}[b]{.26\linewidth}
  \footnotesize
\noindent\begin{tikzpicture}[xscale=0.8,yscale=.94]
  \node (GM) at (0,-2) {$\langle\textup{G},8,\textup{M},5\rangle$};
  \node (TAM) at (1,-1) {$\langle\textup{TA},5,\textup{M},5\rangle$};
  \node (GB) at (-1,-1) {$\langle\textup{G},8,\textup{B},8\rangle$};
  \node (TAB) at (0,0) {$\langle\textup{TA},5,\textup{B},8\rangle$};
  \draw[->] (GM) -- (GB);
  \draw[->] (GB) -- (TAB);
  \draw[->] (GM) -- (TAM);
  \draw[->] (TAM) -- (TAB);
\end{tikzpicture}\null
\caption{Example~\ref{exa:simplegen}}
  \label{fig:example}
\end{minipage}
\end{figure}

\begin{example}
\label{exa:simplegen}
  The po-database $D$ in
Figure~\ref{fig:examplerels} contains information
  about restaurants and hotels in Paris: each po-relation has a total order
  (from top to bottom) according to customer
ratings from a given travel website. For brevity, we do not represent
  identifiers in po-relations, and we also deviate slightly from our formalism
  by adopting the named perspective in examples, i.e., giving names to
  attributes.

  Let $Q \defeq \mathit{Restaurant} \times_\gen
(\sigma_{\mathit{district}\neq
\text{``12''}}(\textit{Hotel}))$.
  Its result $Q(D)$ has two possible worlds, where we abbreviate hotel and
  restaurant names:
  \begin{itemize}
    \item $(\langle \textup{G}, 8, \textup{M}, 5\rangle , \langle \textup{G}, 8,
\textup{B}, 8\rangle , \langle \textup{TA}, 5, \textup{M}, 5\rangle
, \langle
\textup{TA}, 5, \textup{B},8\rangle )$;
\item $(\langle \textup{G}, 8, \textup{M}, 5\rangle , \langle \textup{TA},
5, \textup{M}, 5\rangle ,\langle \textup{G}, 8, \textup{B}, 8\rangle
, \langle \textup{TA},5, \textup{B},8\rangle )$.
  \end{itemize}
  In a sense, these \emph{list relations} of hotel--restaurant pairs are
  \emph{consistent} with the order in~$D$: we
do not know how to order two pairs, except when both the hotel
and restaurant compare in the same way.
  The \emph{po-relation} $Q(D)$ is represented in
  Figure~\ref{fig:example} as a Hasse diagram, again writing tuple values instead of tuple
  identifiers for brevity: note that, following the usual convention for Hasse
  diagrams in partial order theory, the order in Figure~\ref{fig:example} is
  drawn in the  
  reverse direction of that of Figure~\ref{fig:examplerels}, i.e., from bottom to top.

Consider now the query $Q'\defeq\Pi
  (\sigma_{\mathit{Restaurant}.\mathit{district} =
\mathit{Hotel}.\mathit{district}} (Q))$, where $\Pi$
  projects out $\mathit{Hotel}.\mathit{district}$. The possible worlds of~$Q'(D)$
are $(\langle \textup{G},\textup{B},8\rangle,\langle
\textup{TA},\textup{M},5\rangle)$ and $(\langle
\textup{TA},\textup{M},5\rangle, \langle
\textup{G},\textup{B},8\rangle)$, intuitively reflecting two
different opinions on the order of restaurant--hotel pairs
in the same district.
  Defining $Q''$ similarly to $Q'$ but replacing $\times_\gen$ by
  $\times_\lex$ in~$Q$, we have $\pw(Q''(D)) =
  (\langle\textup{G},\textup{B},8\rangle,
  \langle\textup{TA},\textup{M},5\rangle)$.
\end{example}

It is easy to show that we can efficiently evaluate \PosRA queries on
po-relations, which we will use throughout the sequel.

\begin{proposition}\label{prp:repsys}
  For any fixed \PosRA\ query $Q$, given a po-database $D$, we
  can construct the
  \mbox{po-relation}
  $Q(D)$ in time $O\big(\card{D}^{\card{Q}}\big)$, i.e., in polynomial time in
  the size of~$D$.
\end{proposition}

\begin{proof}
  We show the claim by a simple induction on the query~$Q$, noting that
  $\card{Q}$ is at least $k+1$, where $k$ is the number of operators in~$Q$.

\begin{itemize}
  \item If $Q$ is a relation name~$R$, then $Q(D)$ is obtained 
in linear time.

\item If $Q$ is a constant expression, then $Q(D)$ is obtained in constant
  time.

\item If $Q=\sigma_\psi(Q')$ or $Q=\Pi_{k_1\dots k_p}(Q')$, then $Q(D)$ is
  obtained in time linear in $|Q'(D)|$, and we conclude by the induction
    hypothesis.

  \item If $Q=Q_1\cup Q_2$ or $Q=Q_1 \times_\lex Q_2$ or
    $Q=Q_1\times_\gen Q_2$, then $Q(D)$ is obtained in time linear in
    $|Q_1(D)|\times|Q_2(D)|$, and we conclude again by the induction hypothesis.
    \qedhere
\end{itemize}
\end{proof}

Note that Proposition~\ref{prp:repsys} computes the result of a query as a
po-relation~$\OR$. However, we cannot efficiently compute the complete set
$\pw(\OR)$ of possible worlds of~$\OR$, even if all relations of the input
po-database are totally ordered. For instance, consider the query $Q
\colonequals R \cup S$, and a po-database $D$ interpreting $R$ and $S$ as totally
ordered relations with disjoint domains and with $n$ tuples each. It is easy to
see that the query result $Q(D)$ has $2n \choose n$ possible worlds, which is
exponential in~$D$. This intractability is the reason why will we study the
possibility and certainty problems in the sequel.

\subsection{Incomparability of \PosRA\ Operators}

Before extending our query language with accumulation, we address the natural
question of whether any of our operators is subsumed by the others. We show that
this is not the case.

\begin{theorem}\label{thm:incomparable}
  No \PosRA\ operator can be expressed through a combination of the others.
\end{theorem}

We prove Theorem~\ref{thm:incomparable} in the rest of this subsection.
We consider each operator in turn,
and show that it cannot be expressed through a combination of the others. We first
consider constant expressions and show differences in expressiveness even
when setting the input po-database to be empty.
\begin{itemize}
  \item For $\singleton{t}$, consider the query $\singleton{\langle0\rangle}$.
    The value $0$ is not in the database,
    and cannot be produced by the $\ordern{n}$ constant
    expression, and so
    this query has no equivalent that does not use the $[t]$ constant
    expression.
  \item For $\ordern{n}$, observe that $\ordern{2}$ is a po-relation with
    a non-empty order, while any query involving the other operators
    will have empty order (none of
    our unary and binary operators turns unordered po-relations into an ordered
    one, and the $\singleton{t}$ constant expression produces an
    unordered po-relation).
\end{itemize}

Moving on to unary and binary operators, all operators but products are easily shown to be
non-expressible.

\begin{description}
  \item[selection.]
    For any constant $a$ not
    in~$\mathbb{N}$, consider the po-database $D_a$ consisting of a
    single unordered po-relation with name~$R$ formed of
    two unary tuples $\langle 0\rangle$ and
    $\langle a\rangle$. Let
    $Q=\sigma_{.1\neq\text{``0''}}(R)$. Then, $Q(D_a)$ is the po-relation
    consisting only of the tuple $\langle a\rangle$. No \PosRA{} query
    without selection has the same semantics, as no other operator than
    selection can create a po-relation containing the constant~$a$ for
    any input~$D_a$, unless it also contains the constant~$0$.
    
  \item[projection.] $\Pi$ is the only operator that can decrease the arity of an
    input po-relation.

  \item[union.] $[\langle 0\rangle]\cup[\langle 1\rangle]$ (over the empty
    po-database) cannot be simulated by any combination of operators, as
    can be simply shown by induction: no other operator will produce a
    po-relation which has the two elements $0$ and
    $1$ in the same attribute.
\end{description}

Observe that product operators are the only ones that can
increase arity, so taken together they are non-redundant with the other
operators. 
Hence, it only remains to show that each of $\times_\gen$ and
$\times_\lex$ is not redundant.
To do this, let us call \Plex the fragment of \PosRA that disallows the
$\times_{\dir}$ operator, but allows all other operators (including
$\times_{\lex}$).
We also define \Pgen that disallows $\times_{\lex}$ but
not $\times_{\dir}$.

We will first show that the $\times_\gen$ product is not redundant, which we will do using the notion of width.
Specifically, consider the query $Q_\gen=R\times_\gen R$ and an input
po-database $D_n$ where $R$ is mapped to $\ordern{n}$ (an input relation
of width $1$) for an arbitrary $R_n$.
It is then clear that the po-relation $Q(D_n)$ has width $n$.
We will show that this query cannot be captured in \Plex, because 
\Plex 
queries can only make width increase in a way that depends
on the \emph{width} of the input po-relations, but not on their \emph{size}.

\begin{lemma}\label{lem:lexwidth}
  Let $k \geq 2$ and $Q$ be a \Plex query.
  For any po-database~$D$ of width~$\leq k$, the po-relation $Q(D)$ has width $\leq
  k^{\card{Q}+1}$.
\end{lemma}

\begin{proof}
  We first show by induction on the \Plex query $Q$ that the width of the query
  output can be bounded as a function of~$k$. For the base case, 
  the input po-relations have width~$\leq k$, and all constant po-relations (singletons and constant chains) have width~$1$.
  Let us show the induction step.

  \begin{itemize}
    \item Given two po-relations $\OR_1$ and $\OR_2$ with width respectively
      $k_1$ and $k_2$, their
      union $\OR \colonequals \OR_1 \cupgen \OR_2$ clearly has width at most $k_1 + k_2$.
      Indeed, any antichain in~$\OR$ must be
      the union of an antichain of $\OR_1$ and of an antichain of $\OR_2$.
    \item Given a po-relation $\OR_1$ with width $k_1$, applying a projection or
      selection to~$\OR_1$ cannot increase the width.
    \item Given two po-relations 
      $\OR_1 = (\ID_1, T_1, <_1)$
      and
      $\OR_2 = (\ID_2, T_2, <_2)$
      with width respectively $k_1$ and $k_2$,
      their
      product $\OR \defeq \OR_1 \times_{\lex} \OR_2$ has width at most $k_1 \cdot k_2$.
      To show this,
      write 
      $\OR = (\ID, T, <)$, 
      consider any set $A \subseteq \ID$ of cardinality $> k_1
      \cdot k_2$, and let us argue that $A$ is not an antichain. By the definition
      of~$\times_\lex$, we can see each identifier of~$A$ as an element of~$\ID_1
      \times \ID_2$.
      Now, one of the following must hold.

      \begin{enumerate}
        \item Letting $S_1$ be the set of identifiers $u \in \ID_1$ for which we
          have \mbox{$(u, v) \in A$} for some 
          $v \in \ID_2$, it is the
          case that 
          $\card{S_1} > k_1$.
        \item There exists $u$ such that, letting $S_2(u) \defeq \{v \mid (u, v) \in
          A\}$, we have $\card{S_2(u)} > k_2$.
      \end{enumerate}

      Informally, when putting $> k_1 \cdot k_2$ values in buckets (the value of
      their first component), either $> k_1$ different buckets are used, or
      there is a bucket containing $> k_2$ elements.

      In the first case, as $S_1 \subseteq \ID_1$, as $\card{S_1} > k_1$,
      and as $\OR_1$ has width~$k_1$, we know that $S_1$ cannot be an antichain,
      so it must
      contain two comparable elements $u <_1 u'$. Hence, considering any
      $v, v' \in \ID_2$ 
      such that $w = (u, v)$ and $w' = (u', v')$ are in~$A$,
      we have by the definition of $\times_\lex$ that $w < w'$, so that
      $A$ is not an antichain.
      In the second case, as $S_2(u) \subseteq \ID_2$, as $\card{S_2(u)} > k_2$,
      and as $\OR_2$ has width~$k_2$, we know that $S_2(u)$ cannot be an
      antichain,
      so it must contain two comparable elements $v <_2 v'$. 
      Hence, considering
      $w = (u, v)$ and $w' = (u, v')$ which are in~$A$, we have
      $w < w'$, and again $A$ is not an antichain.
      Hence, no set of cardinality $> k_1 \cdot k_2$
      of~$\OR$ is an
      antichain, so $\OR$ has width $\leq k_1 \cdot k_2$ as claimed.
    \end{itemize}

  Second, we explain why the bound on the width of the query output can be
  chosen as in the lemma statement. Specifically,
  letting $o$ be the number of product operators in~$Q$ plus the number of
  union operators, we show that we can bound the width of~$Q(D)$ by~$k^{o + 1}$.
  Indeed, the output of queries
  without product or union operators have width at most $k$
  (because $k \geq 1$). Further, as projections and selections do not change the
  width, the only operators to consider are product and union. For the union
  operator, if $Q_1$ has
  $o_1$ such operators and $Q_2$ has $o_2$ such operators,
  bounding inductively the width
  of $Q_1(D)$ by $k^{o_1 + 1}$ and $Q_2(D)$ by~$k^{o_2 + 1}$, for $Q
  \colonequals Q_1 \cup
  Q_2$, the number of union and product operators is $o_1 + o_2 + 1$,
  and the new bound is $k^{o_1
  + 1} + k^{o_2 + 1}$, which is $\leq k^{o_1 + 1 + o_2 + 1}$ because $k \geq 2$,
  i.e., it is $\leq k^{(o_1 + o_2 + 1) + 1}$. For the~$\times_\lex$ operator, we proceed in the
  same way and directly obtain the $k^{(o_1 + o_2 + 1) + 1}$ bound. Hence, we
  can indeed bound the width of~$Q(D)$ by $k^{\card{Q}+1}$ as given in the statement,
  which concludes the proof.
\end{proof}

We have shown Lemma~\ref{lem:lexwidth}:
\Plex queries can only make the width increase as a
function of the query and of the width of the input po-relations. Hence, the
query $Q_\gen$ cannot be captured in \Plex, and the $\times_\gen$ product is not
redundant.

Conversely, let us show that the $\times_\lex$ product is not redundant. To do
this, we introduce the
\emph{concatenation} of po-relations.

\begin{definition}
  \label{def:concat}
  The \emph{concatenation}
  $\OR \cupcat \OR'$ of two po-relations $\OR$ and $\OR'$
is the series composition of their two partial
orders. Note that $\pw(\OR \cupcat \OR') = \{L \cupcat L' \mid L \in \pw(\OR),
L' \in \pw(\OR')\}$, where $L \cupcat L'$ is the concatenation of two list
  relations in
  the usual sense.
\end{definition}

  We show that concatenation can be captured in \Plex.
  \begin{lemma}
    \label{lem:lexconcat}
    For any arity $n \in \mathbb{N}$ and
    distinguished relation names $R$ and $R'$,
    there is a query $Q_n$ without $\times_\gen$ operator
    such that, for any two
    po-relations $\OR$ and $\OR'$ of arity~$n$, letting $D$ be the
    database mapping $R$ to $\OR$ and $R'$ to $\OR'$, the query result $Q_n(D)$ is
    $\OR \cupcat \OR'$.
  \end{lemma}

  \begin{proof}
    For any $n \in \mathbb{N}$ and names $R$ and $R'$,
    consider the following query:
    \[Q_n(R, R') \defeq
\Pi_{3\dots n+2} \left(\sigma_{.1= .2} \left(\ordern{2}
\times_\lex ((\singleton{1}\times_\lex R) \cupgen
   (\singleton{2}\times_\lex R'))\right)\right)\]
   It is easily verified that $Q_n$ satisfies the claimed property.
  \end{proof}

  By contrast, we show that concatenation cannot be captured in \Pgen.

\begin{lemma}\label{lem:noconcat}
  For any arity $n \in \NN_{>0}$ and distinguished relation names $R$ and~$R'$,
  there is no \Pgen query $Q_n$ such that, for any
  po-relations $\OR$ and $\OR'$ of arity~$n$, letting $D$
  be the
  po-database that maps $R$ to $\OR$ and $R'$ to $\OR'$, the query result $Q_n(D)$
  is $\OR \cupcat \OR'$.
\end{lemma}

  To prove Lemma~\ref{lem:noconcat}, we first introduce the following concept.

  \begin{definition}
    Let $v \in \calD$. We call a
    po-relation $\OR = (\ID, T, <)$ \emph{$v$-impartial} if, for any
    two identifiers $\id_1$ and $\id_2$ and $1 \leq i \leq \arity{\OR}$ such that exactly
    one of $T(\id_1).i$, $T(\id_2).i$ is $v$, the following holds: $\id_1$ and
    $\id_2$ are \emph{incomparable}, namely, neither $\id_1 < \id_2$ nor $\id_2
    < \id_1$ hold.
  \end{definition}

\begin{lemma}
  \label{lem:incomp}
  Let $v \in \calD \backslash \mathbb{N}$ be a value. For any \Pgen query
  $Q$,
  for any
  po-database $D$ of $v$-impartial
  po-relations, the po-relation $Q(D)$ is
  $v$-impartial.
\end{lemma}

\begin{proof}
  Let $D$ be a po-database of
  $v$-impartial po-relations.
  We show by induction on the query $Q$ that $v$-impartiality is preserved.
  The base cases are the following.

  \begin{itemize}
    \item For the base relations, the claim is vacuous by our hypothesis on~$D$.
    \item For the singleton constant expressions, the claim is trivial as
      they contain less than two tuples.
    \item For the $\ordern{i}$ constant expressions, the claim is
      immediate as $v \notin \mathbb{N}$.
  \end{itemize}
  We now prove the induction step.

  \begin{itemize}
    \item For selection, the claim is shown by noticing that, for any
      $v$-impartial
      po-relation $\OR$,
      letting $\OR'$ be the image of $\OR$ by any selection, $\OR'$
      is itself $v$-impartial. Indeed, considering two identifiers $\id_1$ and
      $\id_2$ in
      $\OR'$ and $1 \leq i \leq \arity{\OR}$ satisfying the condition, as $\OR$ is
      $v$-impartial, $\id_1$ and $\id_2$ are incomparable in $\OR$, so they are also
      incomparable in $\OR'$.
    \item For projection, the claim is also immediate as the property to prove
      is maintained when reordering, copying or deleting attributes. Indeed,
      considering again two identifiers $\id_1'$ and $\id_2'$ of $\OR'$ and $1
      \leq i' \leq \arity{\OR'}$, the respective preimages $\id_1$ and $\id_2$ in~$\OR$ of $\id_1'$ and $\id_2'$
      satisfy the same condition for some different $1 \leq i \leq \arity{\OR}$ which is the
      attribute in~$\OR$ that was projected to give attribute~$i'$ in~$\OR'$,
      so we again use the impartiality of the original
      po-relation to conclude.
    \item For union, letting $\OR'' \defeq \OR \cupgen
      \OR'$, and writing $\OR'' = (\ID'', T'', {<''})$, assume by contradiction the existence of two identifiers $\id_1,
      \id_2 \in
      \OR''$ and $1 \leq i \leq \arity{\OR''}$ such that exactly one of
      $T''(\id_1).i$ and
      $T''(\id_2).i$ is $v$ but (without loss of generality) $\id_1 < \id_2$ in
      $\OR''$.
      It is easily seen that, as $\id_1$ and
      $\id_2$
      are not incomparable, they must come from the same relation; but then, as
      that relation was $v$-impartial, we have a contradiction.
    \item For $\times_\gen$, consider
      $\OR'' \defeq \OR\times_\gen \OR'$ where $\OR$ and $\OR'$ are
      $v$-impartial, and write $\OR'' = (\ID'', T'', <'')$ as above.
      Assume that there are two identifiers $\id_1''$ and $\id_2''$ of~$\ID''$
      and $1 \leq i \leq \arity{\OR''}$ that violate the $v$-impartiality of~$\OR''$.
      Let $(\id_1, \id_1'), (\id_2, \id_2') \in \ID \times \ID'$ be the
      pairs of identifiers used to create $\id_1''$ and $\id_2''$.
  We distinguish on whether $1 \leq i \leq \arity{\OR}$ or $\arity{\OR} < i \leq
  \arity{\OR} + \arity{\OR'}$. In the first case, we deduce that exactly one of
      $T(\id_1).i$ and $T(\id_2).i$ is $v$, so that in particular $\id_1 \neq
      \id_2$.
      Thus, by the definition of the order in $\times_\gen$, it is easily seen that,
  because $\id''_1$ and $\id''_2$ are comparable in~$\OR''$,
  $\id_1$ and $\id_2$ must compare in the same way in~$\OR$, contradicting the
  $v$-impartiality of~$\OR$.
  The second case is symmetric.\qedhere
  \end{itemize}

\end{proof}

We now conclude with the proof of Lemma~\ref{lem:noconcat}.

\begin{proof}
  Let us assume by way of contradiction that there is $n \in \NN_{>0}$ and a
  \Pgen query~$Q_n$ that captures $\cupcat$.
  Let $v \neq v'$ be two distinct
  values in $\calD \backslash
  \mathbb{N}$, and consider the singleton
  po-relation $\OR$ containing one identifier of value~$t$ and $\OR'$ containing
  one identifier of value $t'$, where
  $t$ (resp.\ $t'$) are tuples of arity $n$ containing
  $n$ times the value $v$ (resp.\ $v'$). Consider the
  po-database $D$ mapping
  $R$ to $\OR$ and $R'$ to $\OR'$.
  Write $\OR'' \defeq Q_n(D)$.
  By our assumption, as $\OR'' = (\ID'', T'', {<''})$ is $\OR \cupcat \OR'$,
  it must contain an
  identifier $\id \in \ID''$ such that $T''(\id) = t$ and an identifier $\id'
  \in \ID''$ such that $T''(\id') = t'$.
  Now, as $\OR$ and $\OR'$ are (vacuously) $v$-impartial,
  Lemma~\ref{lem:incomp} implies that $\OR''$ is $v$-impartial.
  Hence, as $n > 0$, taking $i = 1$, as $t \neq t'$ and exactly one of $t.1$ and
  $t'.1$ is $v$, the identifiers $\id$ and $\id'$ are incomparable in~$<''$,
  so there is a possible world of $\OR''$ where $\id'$ precedes $\id$.
  This contradicts the fact that, as we should have $\OR'' = \OR \cupcat \OR'$,
  the po-relation $\OR''$ should have exactly one possible world, namely, 
  $(t, t')$.
\end{proof}

This establishes that the $\times_\lex$ operator cannot be expressed using the
others, and shows that none of our operators is redundant, which concludes the
proof of Theorem~\ref{thm:incomparable}.

\subsection{\PosRAacc: Queries With Accumulation}
\label{sec:posraacc}

We now enrich \PosRA with order-aware {\em
accumulation} as the outermost operation, inspired by
\emph{right accumulation} and \emph{iteration} in list programming,
and \emph{aggregation} in
relational databases.
Recall that a \emph{monoid} $(\calM, \oplus, \epsilon)$
consists of a set $\calM$ (not necessarily finite), an associative operation
$\oplus: \calM \times \calM \to \calM$, and an element $\epsilon \in \calM$
which is neutral for $\oplus$, i.e., for all $m \in \calM$, we have $\epsilon
\oplus m = m \oplus \epsilon = m$. We will use a monoid as the structure in
which we perform accumulation. We can now define accumulation on a given list relation. 

\begin{definition}
  \label{def:aggregationb}
  For $k \in \mathbb{N}$,
  let
  $h: \calD^k \times
  \NN_{>0} \to \calM$ be a function
  called an \deft{\mbox{arity-$k$} accumulation map}, which maps pairs
  consisting of an $k$-tuple and
  a position to a value in the monoid~$\calM$.
  We call $\accum_{h, \oplus}$ an \deft{arity-$k$ accumulation operator};
  its result $\accum_{h, \oplus}(L)$
  on an arity-$k$
  list relation $L = (t_1, \ldots, t_n)$ is 
  $h(t_1, 1) \oplus \cdots \oplus h(t_n,
  n)$, and it is $\epsilon$
  if $L$ is empty.
  For complexity purposes, we \emph{always} require accumulation operators to be
  \emph{PTIME-evaluable}, i.e., we can evaluate the accumulation map and the
  monoid operator in time polynomial in their inputs, and we can compute $\accum_{h,
  \oplus}(L)$ in polynomial time on any input list relation~$L$.
\end{definition}

Intuitively, the accumulation operator
maps each occurrence of a tuple in the list  with $h$ to~$\calM$, where accumulation is performed
with~$\oplus$. (Remember that the input $L$ to the accumulation is a list
relation, so each tuple occurrence has a specific position.)
The map~$h$ may use its second argument to take into
account the absolute position of tuples in~$L$. In what
follows, we
omit the arity of accumulation
when clear from context.

We will often look at special cases for accumulation, especially when deriving
complexity results. Here are the restrictions that we will consider.

\begin{definition}
  \label{def:accrestr}
  We say that an accumulation operator is \emph{position-invariant} if its
  accumulation map ignores the
  second input, so that effectively its only input is the tuple itself

  We say that an accumulation operator is \emph{finite} if its monoid $(\calM,
  \oplus, \epsilon)$ is finite.

  For any monoid $(\calM, \oplus, \epsilon)$,
  we call $a \in \calM$ \deft{cancellable} if, for all~$b, c \in \calM$, we
  have that $a \oplus b = a \oplus c$ implies $b = c$, and $b \oplus a = c
  \oplus a$ implies $b = c$.
  We call $\calM$ a \deft{cancellative monoid}~\cite{howie1995fundamentals} if all
  its elements are cancellable.
  We say that an accumulation operator is
  \emph{cancellative} if its monoid is.
\end{definition}
  
  Note that, in particular, a group is always cancellative,
  but there are some cancellative monoids which are not groups, e.g., the monoid
  of concatenation. 

We can now define the
language \posRAagg that contains all queries of the form
$Q=\accum_{h,\oplus}(Q')$, where $\accum_{h,\oplus}$ is an accumulation
operator and $Q'$ is a \PosRA\ query. 
The \emph{possible results} of $Q$ on a
po-database~$D$, denoted $Q(D)$, is the set of results obtained by applying
accumulation to each possible world of $Q'(D)$, namely:

\begin{definition}
  \label{def:aggreg-po}
  For a po-relation $\OR$,
  we define
  \(\accum_{h, \oplus}(\OR)
  \defeq\{\accum_{h, \oplus}(L) \mid L \in
  \pw(\OR)\}\).
\end{definition}

Of course, accumulation has exactly one result whenever the accumulation operator $\accum_{h,
\oplus}$ does not depend on the order of input tuples:
this covers, e.g., the standard sum, min, max, etc.
Hence,
we focus
on accumulation operators which \emph{depend on the order of tuples}, e.g.,
the monoid $\calM$ of strings with
$\oplus$ being the concatenation operation. In this case, there may be more than one accumulation result.

\begin{example}
  \label{exa:aggreg}
  As a first example, let
  $\mathit{Ratings}(\mathit{user},\mathit{restaurant},\mathit{rating})$ be an
  \emph{unordered} po-relation describing
  the numerical ratings given by users to restaurants, where each user rated each restaurant at most
  once. Let $\mathit{Relevance}(\mathit{user})$ be a po-relation
  giving a partially-known ordering of users to indicate the relevance of
  their reviews. We wish to compute a \emph{total rating} for each restaurant
  which is given by the
  sum of its reviews weighted by a
  PTIME-computable weight
  function~$w$. Specifically, $w(i)$ gives a nonnegative weight to the rating of the $i$-th most relevant user.
  Consider 
  \(
  Q_1 \defeq
\accum_{h_1,+}(\sigma_\psi(\mathit{Relevance}\times_{\lex}\mathit{Ratings}))\)
  where we set \(h_1(t,n) \defeq t.\mathit{rating} \times w(n)\), and where
  $\psi$ is the tuple predicate:
  \(
    \mathit{restaurant}=\text{``Gagnaire''}\land\mathit{Ratings}.\mathit{user}=
    \mathit{Relevance}.\mathit{user}
\).
  The query $Q_1$ gives the total rating of \textquote{Gagnaire}, and each
possible world of $\mathit{Relevance}$ may lead to a different accumulation
  result. This accumulation operator is cancellative, but it is neither
  position-invariant nor finite.

As a second example, consider an unordered po-relation
$\mathit{HotelCity}(\mathit{hotel}, \mathit{city})$ indicating in which
city each hotel is located, and consider a po-relation
$\mathit{City}(\mathit{city})$ which is (partially) ranked by a
criterion such as interest level, proximity, etc. Now consider the
query
  \(
Q_2 \defeq \accum_{h_2,\mathrm{concat}}(
  \Pi_{\mathit{hotel}}(
Q_2'))\), with \(Q_2' \defeq \sigma_{\mathit{City}.\mathit{city} =
\mathit{HotelCity}.\mathit{city}}(
\mathit{City}\times_{\lex}\mathit{HotelCity})\) and
$h_2(t,n)\defeq t$. Here, the operator ``$\mathrm{concat}\!$'' denotes standard string concatenation.
$Q_2$ concatenates the hotel names according to the preference order on the city
where they are located, allowing any possible order between hotels of the same city and
between hotels in incomparable cities. This accumulation operator is
  cancellative and position-invariant, but it is not finite.
\end{example}

\section{Possibility and Certainty}\label{sec:posscertdef}
Evaluating a \PosRA or \PosRAacc query $Q$ on a po-database $D$ yields a {\em set of possible
results}: for \PosRAacc, it yields an explicit set of accumulation results, and
for \PosRA, it yields a po-relation that represents a set of possible worlds (list
relations).
The uncertainty on the result may come from uncertainty on the order of the
input relations (i.e., if they are po-relations with multiple possible worlds),
but it may also be caused by the query, e.g., the union of two non-empty totally
ordered relations is not totally ordered.
In some cases, however, there is only one
possible result to the query, i.e., a \emph{certain} answer.
In other cases, we may wish
to examine multiple \emph{possible} answers. We thus define the corresponding problems.

\begin{definition}[Possibility and Certainty]
  \label{def:posscert}
  Let $Q$ be a \PosRA query, $D$ be a po-database, and $L$ a list
  relation. The \emph{possibility problem} (\poss) asks 
  if $L \in \pw(Q(D))$, i.e.,
  if $L$ is a possible result of~$Q$ on~$D$.
  The \emph{certainty
  problem} (\cert) asks if $\pw(Q(D)) = \{L\}$, 
  i.e., if $L$ is the only possible result of~$Q$ on~$D$.

  Likewise,
  if $Q$ is a \PosRAacc query with an
 accumulation monoid $\calM$, for a result $v
 \in \calM$, the \poss problem asks whether $v \in Q(D)$, and \cert asks
 whether $Q(D) = \{v\}$.
\end{definition}

For \PosRAacc, our definition follows the usual notion of
possible and certain answers in data integration
\cite{DBLP:conf/pods/Lenzerini02} and incomplete information \cite{Libkin06}.
For \PosRA,
we ask for possibility or certainty of an \emph{entire} output list
relation of tuples {\em without identifiers}: indeed, as we explained above, the
identifiers are only internally generated and thus expected to be unknown to the
user. These problems correspond to 
\emph{instance possibility and certainty}~\cite{KO06}. 
We now justify that these notions are useful and discuss more ``local'' alternatives.

First, as we exemplify below,
the output of a query may be certain even for a complex
query and uncertain input. 
It is important to identify such cases and present the user with the certain
answer in full,
like order-by query results in current DBMSs.
Our \cert problem is useful for this task, because we can use it to decide if a
certain output exists: and if it is the case, then we can compute the certain
output in polynomial time, by choosing an arbitrary
linear extension and computing the corresponding possible world.
However, \cert is a challenging problem to solve, because of duplicate values 
(see the ``Technical difficulties'' paragraph below).

\begin{example}
	\label{ex:certexample}
  Consider the po-database $D$ of Figure~\ref{fig:examplerels} with relations
  $\mathit{Restaurant}$ and $\mathit{Hotel}_2$.
  To find recommended pairs of hotels and restaurants in the same
  district, we can write
  $Q \defeq
  \sigma_{\mathit{Restaurant}.\mathit{district}=\mathit{Hotel}_2.\mathit{district}}
  (\mathit{Restaurant} \times_{\dir} \mathit{Hotel}_2)$. Evaluating $Q(D)$ yields 
  the list relation $(\langle G,8,B,8\rangle,
  \langle \mathit{TA},5,M,5\rangle)$ as a unique possible world: it is a \emph{certain} result.

  We may also obtain a certain result in cases when the input relations are
  larger. Imagine for example that
  we join hotels and restaurants to find pairs of a hotel and a restaurant
  located in that hotel. The result can be certain if the relative ranking of the
  hotels and of their restaurants agree.
\end{example}

If there is no certain answer, 
we can instead try to decide whether some list relations 
are a possible answer. This can be useful, e.g., to check if a
list relation (obtained from another source) is consistent with a query result. For
example, we may wish to check if a website's 
ranking of 
hotel--restaurant pairs is
\emph{consistent} with the preferences expressed in its rankings for hotels and
restaurants, to detect when a pair is ranked higher than its components would
warrant: this can be done by checking if the ranking on the pairs is a possible
result of the query that unifies the hotel ranking and restaurant ranking.

When there is no overall certain answer, or when we want to check the
possibility of some aggregate property of the relation, we can use a \PosRAacc
query.
In particular, in addition to the applications of Example~\ref{exa:aggreg},
accumulation allows us to encode alternative notions of \poss and \cert for {\em
\PosRA} queries, and to express them as \poss and \cert for \PosRAacc. For example,
instead of possibility or certainty for a full relation, we can express
possibility or certainty of the \emph{position}\footnote{Remember that the {\em existence} of a tuple is not
order-dependent, so it is trivial to check in our setting.}
of particular tuples of interest.

One particular application of accumulation is to model
  {\emph{position-based selection}} queries. Consider for instance a \emph{top-$k$}
operator, defined on list relations, which retrieves a list relation of the
  first~$k$~tuples. Let us extend the top-$k$ operator to po-relations in the
  expected way: the set of top-$k$ results on a po-relation~$\OR$ is the set of
  top-$k$ results on the list relations of $\pw(\OR)$. We can
implement {\em top-$k$} as $\accum_{h_3,\concat}$
with $h_3(t, n)$ being $(t)$ for $n \leq k$ and $\epsilon$ otherwise,
and with $\text{concat}$ being list concatenation. We can similarly compute \emph{select-at-$k$}, i.e., return the
tuple at position $k$, via $\accum_{h_4,\concat}$
with $h_4(t, n)$ being $(t)$ for $n=k$ and $\epsilon$ otherwise. Both these
accumulation operators are cancellative because they use the concatenation
monoid, and they are finite if we assume that the domain of the output is fixed (e.g., ratings
in $\{1, \ldots, 10\}$), and if we also assume for top-$k$ that $k$ is fixed.

  Accumulation can also be used for a {\emph{tuple-level comparison}}.
  To check whether 
the first occurrence of a tuple $t_1$ precedes any occurrence of
$t_2$, we define $h_5$ for all $n\in\NN$ by $h_5(t_1,n) \defeq \top$, $h_5(t_2,n) \defeq \bot$ and
$h_5(t,n)\defeq\epsilon$ for $t \neq t_1,t_2$, and a monoid operator
$\oplus$ that returns its first argument:
  assuming that $t_1$ and $t_2$ are both present, the result is~$\top$ if
  the first occurrence of $t_1$ precedes any occurrence of $t_2$, and it
  is~$\bot$ otherwise. This accumulation operator is finite and
  position-invariant, but not cancellative.

We study the complexity of these variants in Section \ref{sec:fpt}. We now
give examples of their use.

\begin{example}
  \label{exa:accumul}
  Let $Q\defeq\Pi_{\mathit{district}}(
  \sigma_{\mathit{Restaurant}.\mathit{district}=\mathit{Hotel}.\mathit{district}}
  (\mathit{Restaurant} \times_{\dir} \mathit{Hotel}))$, that computes ordered
  recommendations of districts including both hotels and restaurants.
  The user can use accumulation to compute the best district to stay in with
  $Q'=\text{top-}1(Q)$.
  When $Q'$ has a certain answer,
  there is a dominating
  hotel--restaurant pair in this district which answers the user's need.
  If there is no certain answer, \poss allows the user to 
  determine 
  the \emph{possible} top-$1$ districts.
	
  We can also use \poss and \cert for \PosRAacc queries to restrict attention to
  \emph{tuples} of interest. If the user hesitates between districts $5$ and $6$, they
  can apply tuple-level comparison  to see
  whether the best pair of district~$5$ may be better (or is always better) than
  that of~$6$.
\end{example}

\subparagraph*{Technical difficulties.} The main challenge to solve \poss and
\cert for a \PosRA query $Q$ on an input po-database $D$ is that the tuple values of
the desired result~$L$ may
occur multiple times in the po-relation $Q(D)$, making it hard to match $L$ and
$Q(D)$.
In other words, even though we can compute the po-relation $Q(D)$ in
polynomial time (by
Proposition~\ref{prp:repsys}) and present it to the user, they still cannot
easily determine the possible and certain answers out of the po-relation.

\begin{figure}
\centering\begin{tikzpicture}[yscale=.8,xscale=1.5,every node/.style={align=center,outer sep=0,inner
  sep=2pt}]
  \node (n13) at (0, 0) {fr\\[-.45em]{\footnotesize a}};
  \node (n20) at (1, 0) {it\\[-.45em]{\footnotesize b}};
  \node (n37) at (0, 1) {fr\\[-.45em]{\footnotesize c}};
  \node (n42) at (1, 1) {it\\[-.45em]{\footnotesize d}};
  \node (n100) at (0, 2) {jp\\[-.35em]{\footnotesize e}};
  \node (n102) at (1, 2) {jp\\[-.35em]{\footnotesize f}};
  \draw[->] (n13) -- (n37);
  \draw[->] (n20) -- (n37);
  \draw[->] (n37) -- (n100);
  \draw[->] (n42) -- (n100);
  \draw[->] (n42) -- (n102);
\end{tikzpicture}
\caption{Po-relation in Example~\ref{exa:notposet}}
  \label{fig:notposet}
\end{figure}
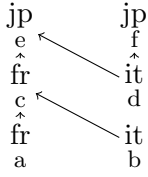

\begin{example}
  \label{exa:notposet}
  Consider a po-relation $\OR = (\ID, T, {<})$ with
  $\ID = \{\id_{\mathrm{a}}, \allowbreak\id_{\mathrm{b}},
  \allowbreak\id_{\mathrm{c}},
  \allowbreak\id_{\mathrm{d}}, \allowbreak\id_{\mathrm{e}},
  \allowbreak\id_{\mathrm{f}}\}$, with
  $T(\id_{\mathrm{a}}) \defeq \langle\text{Gagnaire}, \text{fr}\rangle$,
  $T(\id_{\mathrm{b}}) \defeq
  \langle\text{Italia}, \text{it}\rangle$,
  $T(\id_{\mathrm{c}}) \defeq \langle\text{TourArgent}, \text{fr}\rangle$,
  $T(\id_{\mathrm{d}}) \defeq
  \langle\text{Verdi}, \text{it}\rangle$, $T(\id_{\mathrm{e}}) \defeq
  \langle\text{Tsukizi}, \text{jp}\rangle$,
  $T(\id_{\mathrm{f}}) \defeq \langle\text{Sola}, \allowbreak \text{jp}\rangle$,
  and with $\id_{\mathrm{a}} < \id_{\mathrm{c}}$, $\id_{\mathrm{b}} <
  \id_{\mathrm{c}}$, $\id_{\mathrm{c}} <
  \id_{\mathrm{e}}$, $\id_{\mathrm{d}} < \id_{\mathrm{e}}$, and
  $\id_{\mathrm{d}} < \id_{\mathrm{f}}$.
  Intuitively, $\OR$ describes a preference relation over restaurants,
  with their name and the type of their cuisine.
  Consider the \PosRA query $Q \defeq \Pi(\OR)$ that projects~$\OR$ on type; we
  illustrate the result (with the original identifiers) in Figure~\ref{fig:notposet}.
  Let $L$ be the list
  relation $(\text{it}, \text{fr}, \text{jp}, \text{it}, \text{fr},
  \text{jp})$, and consider
  \poss for $Q$, $\OR$, and $L$.
 
  We have that $L \in \pw(Q(\OR))$, as shown by the linear extension
  $\id_{\mathrm{d}} <' \id_{\mathrm{a}} <' \id_{\mathrm{f}} <' \id_{\mathrm{b}} <'
  \id_{\mathrm{c}} <' \id_{\mathrm{e}}$
  of~$<$. However, this is hard to see, because 
  each of \text{fr}, \text{it}, \text{jp} appears more than once in the
  candidate list as well as in the po-relation; there are thus multiple
  ways to ``map'' the elements of the candidate list to those of the po-relation, and only some of these mappings lead to the existence of a corresponding linear extension.
  It is also challenging to check if $L$ is a
  certain answer: here, it is not, as there are other possible answers,
  such as $(\text{it},
  \text{fr}, \text{fr}, \text{it}, \text{jp}, \text{jp})$.
\end{example}

In the following sections we study the computational complexity of the \poss and
\cert problems, for multiple fragments of our language.

\section{General Complexity Results}\label{sec:posscert}

We have defined the \PosRA and \PosRAacc query languages, and defined and
motivated the problems \poss
and \cert. We now start the study of their complexity, which is the main technical
contribution of our paper. We will always study their \emph{data complexity}\footnote{In \emph{combined complexity}, with $Q$ part of
  the input, \poss and \cert are easily seen to be NP-hard even without
  order, by reducing from
  the evaluation of Boolean conjunctive queries (which is NP-hard in
  combined
complexity \cite{AHV-1995}).},
where the query $Q$ is fixed: in particular, for \PosRAacc, the accumulation map and
monoid, which we assumed to be PTIME-evaluable, is fixed as part of the query,
though it is allowed to be infinite. The input to \poss and \cert for the fixed
query $Q$ is the
po-database $D$ and the candidate result (a list relation for \PosRA, an
accumulation result for \PosRAacc). We summarize the complexity results of
Sections~\ref{sec:posscert}--\ref{sec:fpt} in Table~\ref{tab:complexity}.
\begin{table*}
  \footnotesize
  \caption{Summary of complexity results for possibility and
  certainty}
  \label{tab:complexity}
  {
    {
      \renewcommand{\tabcolsep}{2pt}
  \begin{tabularx}{\linewidth}{Xllll@{~~}l}
    \toprule
    &
    {\bfseries Query} &
    {\bfseries Restr.\ on accum.} &
    {\bfseries Input po-relations} &
    \multicolumn{2}{l}{\bfseries Complexity} \\
    \midrule
    \poss &
    \PosRA/\posRAagg &
    ---
      & arbitrary
      & NP-c. & (Thm.~\ref{thm:posscomp1}) \\
\cert & \posRAagg & --- & arbitrary
      & coNP-c. & (Thm.~\ref{thm:certcomp}) \\
    \cert & \PosRA & --- &
arbitrary
      & PTIME & (Thm.~\ref{thm:certaintyptimec}) \\
    \midrule
    \poss & \Plex & ---
      & width $\leq k$
      & PTIME & (Thm.~\ref{thm:aggregw}) \\
    \poss & \Pgen & ---
      & totally ordered
      & NP-c. & (Thm.~\ref{thm:posscompextend1}) \\
    \poss & \Pnoprod & ---
      & ia-width or width $\leq k$
      & PTIME & (Thm.~\ref{thm:aggregnoprod}) \\
    \poss & \Plex/\Pgen & ---
      & 1 total.\ ord., 1 unord.
      & NP-c. & (Thm.~\ref{thm:posscompextended}) \\
    \midrule
\cert & \posRAagg & cancellative &
arbitrary
      & PTIME & (Thm.~\ref{thm:certaintyptimec}) \\
    \poss & \PosRAacc & finite and pos.-invar.
    & totally ordered
      & NP-c. & (Thm.~\ref{thm:possfrihypoposs}) \\
    \cert & \PosRAacc & finite and pos.-invar.
    & totally ordered
      & coNP-c. & (Thm.~\ref{thm:possfrihypocert}) \\
    both & \Plexacc & finite
      & width $\leq k$
    & PTIME & (Thm.~\ref{thm:aggregwa}) \\
    both & \Pnoprodacc & finite and pos.-invar.
      & ia-width or width $\leq k$
      & PTIME & (Thm.~\ref{thm:aggregnoproda}) \\
    \poss & \Pnoprodacc & pos.-invar.
      & unordered
      & NP-c. & (Thm.~\ref{thm:possgri}) \\
    \bottomrule
  \end{tabularx}
}
}
\end{table*}

In this section, we state our main complexity results and prove the
corresponding upper
bounds. Lower bounds will be implied by more precise results that will be
established in Sections~\ref{sec:fpt2} and~\ref{sec:fpt}.

We start with \poss, which we show
to be NP-complete.

\begin{theorem}\label{thm:posscomp1} The \poss\ problem is in NP for any fixed \PosRA
  or \PosRAagg query. Further, there exists a \PosRA query and a \PosRAagg query for which the \poss problem is
  NP-complete.
\end{theorem}

\begin{proof}
	To show that \poss is in NP, evaluate the query without
	accumulation in PTIME using Proposition~\ref{prp:repsys}, yielding a
	po-relation~$\OR$.
	Now, guess a total order of~$\OR$, checking in PTIME
	that it is compatible with the comparability relations of~$\OR$.
	If there is no accumulation function, then check that it achieves the
	candidate result. Otherwise, evaluate the accumulation (in PTIME as the
	accumulation operator is PTIME-evaluable),
	and check that the correct result is obtained. This shows that \poss is
        in NP for \PosRA and \PosRAacc queries. 
        The NP-hardness will follow from stronger results that will be shown
        later: Theorem~\ref{thm:posscompextend1} for \PosRA and 
        Theorem~\ref{thm:possfrihypoposs} for \PosRAagg.
\end{proof}

A different route to prove the NP-hardness of \poss is to use existing
work~\cite{warmuth1984complexity} about the complexity of the so-called
\emph{shuffle problem}: given a string $w$ and a tuple of strings $s_1, \ldots,
s_n$ on the fixed alphabet $A=\{a, b\}$, decide whether there is an interleaving
of $s_1, \ldots, s_n$ which is equal to~$w$. It is easy to see that there is a
reduction from the shuffle problem to the \poss problem, by representing each
string $s_i$ as a totally ordered relation $L_i$ of tuples labeled~$a$ and~$b$
that code the string, letting $\OR$ be the po-relation defined as the union of
the $L_i$, and checking if the totally ordered relation that codes~$w$ is a
possible world of the identity \PosRA query on the po-relation~$\OR$. Hence, as
the shuffle problem is NP-hard~\cite{warmuth1984complexity}, we deduce that
\poss is NP-hard. However, this approach will not suffice to derive the
stronger NP-hardness results which we prove in the sequel.

We now show that \cert is coNP-complete for \PosRAacc.

\begin{theorem}\label{thm:certcomp}
	The \cert problem is in coNP for any fixed \PosRAagg query, and there is a
        \PosRAagg query for which it is coNP-complete.
\end{theorem}

\begin{proof}
  The co-NP upper bound is proved using precisely the same reasoning applied to
  the NP upper bound for \poss, except that we now guess an order that achieves
  a result {\em different} from the candidate result. The hardness result for
  \cert and \PosRAacc is presented (in a slightly stronger form) as
  Theorem~\ref{thm:possfrihypocert} in the sequel.
\end{proof}

For \PosRA queries, we will show
that \cert is in PTIME. This will follow from a stronger result that we
will prove in the sequel
  (Theorem~\ref{thm:certaintyptimec}): \cert is
    in PTIME for \PosRAacc queries that perform
      accumulation in a cancellative monoid.

\paragraph{Practical implications} We now discuss some implications of
the results highlighted in Table~\ref{tab:complexity} on the
implementation of the algebra on top of, say, a SQL database engine.
First, recall Proposition~\ref{prp:repsys}: computing the result of a
query, as a po-relation, is in PTIME in the size of the input
database, and can thus reasonably be implemented. Second, thanks to
Theorem~\ref{thm:certaintyptimec}, since \cert is in PTIME for \PosRA, it
should also be possible to implement certainty tests efficiently.
However, Theorem~\ref{thm:posscomp1} shows that possibility tests are
prohibitive to implement in all generality.

However, in a practical context, input relations are usually not arbitrary
po-relations: it makes sense to assume in many scenarios that input
relations are either totally ordered (say, because they are ordered by
their primary key, or by an explicit \texttt{ORDER BY} construct) or
unordered (because no specific ordering has been chosen). In this case,
we have two ways to ensure that the possibility problem is tractable: either we
only allow totally ordered po-relations as input and then
Theorem~\ref{thm:aggregw} (in Section~\ref{sec:fpt2})
shows that possibility is tractable if the only product
operator allowed is~$\times_\lex$;
or we allow both totally ordered and unordered po-relations as input, but
then only queries with no product are tractable for possibility tests
(which, arguably, considerably limits the expressive power).

When moving to \PosRAacc, the picture is similar, but we need additional
properties of the accumulation function to ensure that the possibility
and certainty problems
are tractable (depending on the cases, it should be cancellative, finite, or
position-invariant).

\medskip

We next identify further tractable cases. In the following section, we study \PosRA
queries: we focus on \poss, as we know that \cert is always in PTIME for such
queries. In Section~\ref{sec:fpt}, we turn to \PosRAacc.

\section{Tractable Cases for \possbold on \PosRAbold Queries}\label{sec:fpt2}
We have stated a general NP-hardness result for \poss with \PosRA queries.
We next show that tractability may be achieved if we both restrict the allowed
operators and bound some order-theoretic parameters of the input po-database,
e.g., its width. Recall that \Plex (respectively, \Pgen) denotes the
fragment of \PosRA that disallows $\times_{\dir}$ (respectively, $\times_{\lex}$).

\subsection{(Almost) Totally Ordered Inputs}
We start by the natural case
where we assume that the width of all input po-relations is bounded by a
constant. This assumption is a common practical case: it covers the case where
all input po-relations are \emph{totally ordered}, i.e., their order relation is a total
order, so they actually represent a list relation.
This applies
to situations where we integrate data from multiple sources that are
certain (totally ordered), and where uncertainty only arises
because of the integration query.
The assumption also covers the case of po-relations that are totally ordered
except for a few ``tied" data items at each level.
Recall that the query result can still have
exponentially many possible worlds under this assumption, e.g., when taking the
union of two totally ordered relations.
In a sense, the $\times_\dir$ operator is the one introducing the most
uncertainty and ``unorderedness'' in the result, so we
consider the fragment \Plex of \PosRA queries without
$\times_{\dir}$, and show the following result.

\begin{theorem}\label{thm:aggregw}
	For any fixed $k \in \NN$ and fixed \Plex query $Q$,
	the \poss problem for~$Q$ is in PTIME when
	all po-relations of the input
	po-database have width $\leq k$.
\end{theorem}

To show this result, letting $D$ be the input po-database, we can use
Proposition~\ref{prp:repsys} to evaluate $\OR \colonequals Q(D)$ in PTIME.
Recall that we have previously shown Lemma~\ref{lem:lexwidth} on \Plex, so we
know that the width of the po-relation $\OR$ is constant: it only depends on~$k$
and~$Q$, but not on~$D$. Hence,
to show Theorem~\ref{thm:aggregw}, it suffices to show the following.

\begin{lemma}\label{lem:aggregwinst}
  For any constant $k \in \mathbb{N}$,
  we can determine in PTIME,
  for any po-relation $\OR$ with width $\leq k$
  and list relation $L$,
  whether $L \in \pw(\OR)$.
\end{lemma}

Let us prove this lemma and conclude the proof of Theorem~\ref{thm:aggregw}.

\begin{proof}
     Let $\OR=(\ID,T,{<})$ be the po-relation of width $k' \leq k$, and let
     $P=(\ID,{<})$ be its underlying poset. We use 
     Dilworth's theorem~\cite{dilworth1950decomposition, fulkerson1955note}
     to compute in PTIME a chain partition $\ID = \Lambda_1 \sqcup \cdots \sqcup
     \Lambda_{k'}$ of~$P$.
     For $1 \leq i \leq k'$,
     we write $n_i \defeq \card{\Lambda_i}$, 
     we write $\Lambda_i[j]$ for $1 \leq j \leq n_i$ to denote the $j$-th
     element of~$\Lambda_i$, 
     and for $0 \leq j \leq n_i$, 
     we write
     $\Lambda_i^{\leq
     j}$ to denote the first $j$ elements of the chain $\Lambda_i$, formally,
     $\Lambda_i^{\leq j} \colonequals \{\Lambda_i[j'] \mid 1 \leq j' \leq j\}$.
     In particular, $\Lambda_i^{\leq 0} = \emptyset$ and $\Lambda_i^{n_i}
     = \Lambda_i$.
 
     We now consider all vectors~$\mathbf{m}$ of the form $(m_1, \ldots, m_{k'})$, with
     $0 \leq m_i \leq n_i$ for each $1 \leq i \leq k'$.
     There are polynomially many such vectors, more specifically at most
     $\card{\OR}^{k}$ of them (recall that $k$ is a constant).
     To each such vector $\mathbf{m}$ we
     associate the subset $s(\mathbf{m})$ of $P$ consisting of
     $\bigsqcup_{i=1}^{k'} \Lambda_i^{\leq m_i}$.
 
     We call such a vector $\mathbf{m}$ \deft{sane} if $s(\mathbf{m})$ is an
     order ideal. Note that this is not always the case: while
     $s(\mathbf{m})$ is always an order ideal of the subposet
     of the comparability relations within the chains, it may not be an order
     ideal of~$P$ overall because of
     the additional comparability relations across the chains.
     For each vector
     $\mathbf{m}$, we can check in PTIME whether it is sane: simply materialize 
     $s(\mathbf{m})$ and check that it is an ideal by considering each of the
     $\leq \card{P}^2$ comparability relations.
 
     By definition, for each sane vector $\mathbf{m}$, we know that $s(\mathbf{m})$
     is an ideal. We now observe that the converse is also true: for every
     ideal $S$ of $P$, there is a sane vector $\mathbf{m}$ such that
     $s(\mathbf{m}) = S$. To see why, consider any ideal $S$, and determine for
     each $1 \leq i \leq k'$ the last element of the chain $\Lambda_i$ which is
     in~$S$: let
     $m_i \colonequals 1 \leq i \leq n_i$ be the position of this element
     in~$\Lambda_i$, where we set $m_i \colonequals 0$ if $S$ contains no
     element of~$\Lambda_i$. We know that~$S$
     does not include any
     element of $\Lambda_i$ at a position later than~$m_i$, and because $\Lambda_i$ is a chain it must
     include all elements before~$m_i$; in other words, we have $S \cap \Lambda_i = \Lambda_i^{\leq m_i}$.
     As $(\Lambda_i)_{1\leq i\leq k'}$ is a chain partition of~$P$, this uniquely determines
     $S$. Thus we have indeed $S = s(\mathbf{m})$, and the fact that
     $s(\mathbf{m})$ is sane is witnessed by~$S$.

     We now use a dynamic algorithm to compute, for each sane vector
     $\mathbf{m}$, a Boolean denoted $t(\mathbf{m})$ which is true
     iff there is a topological sort of~$s(\mathbf{m})$ whose
     label is the prefix of the candidate possible world $L$ having length
     $\card{s(\mathbf{m})} = \sum_{i=1}^{k'} m_i$.
     We extend the function $t$ to
     arbitrary vectors by setting $t(\mathbf{m}) \colonequals 0$ whenever
     $\mathbf{m}$ is not sane.
     Specifically, the base case is
     that $t(0, \ldots, 0) \defeq \text{true}$, because the empty ideal
     trivially achieves the empty prefix. 
     To define the induction case, let us denote by 
     $e_i$ for $1 \leq i \leq k'$ the vector 
     consisting of $n-1$ zeros and a $1$ at position $i$. 
     Now, for each sane vector~$\mathbf{m}$,
     we have:
    \[
      t(\mathbf{m}) \defeq \bigvee_{\substack{1 \leq i \leq k'\\m_i > 0}} \left(
      \Bigg(T(\Lambda_i[m_i]) = L\Bigg[\sum_{i' = 1}^{k'} m_{i'}\Bigg]\Bigg)
      \land t(\mathbf{m} - e_i)
      \right)
    \]
    where $L$ is the candidate possible world and where ``$-$''
    denotes the component-wise difference on vectors. It is clear that
    $t(\mathbf{m})$ is correct by induction: the key argument is that, for any
    sane vector $\mathbf{m}$, any linear extension of~$s(\mathbf{m})$ 
    must finish by enumerating one of the maximal elements of~$s(\mathbf{m})$, that
    is, $\Lambda_i[m_i]$ for some $1 \leq i \leq
    k'$ such that $m_i > 0$: and then the linear extension
    achieves the prefix of~$L$ of length $\card{s(\mathbf{m})}$ iff the
    following two conditions are true: (i.)
    the label by~$T$ of the last element in the linear extension
    must be the label of element of~$L$ at
    position~$\card{s(\mathbf{m})}$; and (ii.) $\mathbf{m} - e_i$ must be a sane
    vector such that the restriction of the linear extension to $s(\mathbf{m}-
    e_i)$ achieves the prefix of~$L$ of length~$\card{s(\mathbf{m}-e_i)}$ which
    by induction was computed as~$t(\mathbf{m}-e_i)$.

    It is now clear that we can compute all~$t(\mathbf{m})$ in PTIME by a dynamic
    algorithm: we enumerate the vectors (of which there are polynomially
	many) in lexicographical order, and computing their image by $t$ in PTIME
	according to the equation above, from the base case $t(0, \ldots, 0) =
	\epsilon$ and from the previously computed values of $t$, recalling that
        $t(\mathbf{m'}) \colonequals 0$ whenever $\mathbf{m'}$ is not sane.
    Now, $t(n_1, \ldots, n_{k'})$ is true iff $\OR$ has a linear
    extension achieving $L$, so we have indeed solved the \poss problem
    for~$\OR$ and~$L$ in PTIME, concluding the proof.
\end{proof}
 
We have now shown Theorem~\ref{thm:aggregw} and established tractability
for~\poss with \Plex queries on po-databases of bounded width.
We will show in Theorem~\ref{thm:aggregwa} that this proof technique further
extends to queries with accumulation, under some assumptions over the
accumulation function.

We next show that our tractability result only holds for \Plex. 
Indeed, if we allow $\times_{\dir}$, then \poss is hard on totally ordered
po-relations, even if we disallow $\times_{\lex}$. This result implies the
general NP-hardness result on \poss that we stated earlier
(Theorem~\ref{thm:posscomp1}) for queries without accumulation.

\begin{theorem}\label{thm:posscompextend1}
  There is a \Pgen query for which the \poss problem is NP-complete
  even when input po-databases consist only of totally ordered
  po-relations.
\end{theorem}

\begin{proof}
  We reduce from the NP-hard UNARY-3-PARTITION
  problem~\cite{garey-johnson}: given $3m$ integers $E = (n_1, \ldots, n_{3m})$
  written in unary (not necessarily distinct) and a number $B$, decide if the integers can be partitioned
  in triples such that the sum of each triple is $B$. We reduce an instance
  $\mathcal{I} = (E, B)$ of UNARY-3-PARTITION to a \poss instance in PTIME.
  We fix $\calD \defeq \NN \sqcup \{\s, \n, \e\}$, with $\s$, $\n$ and $\e$
  standing for \emph{start}, \emph{inner},
  and \emph{end}.

  Let $D$ be the po-database which interprets the relation name $S$ by the 
  totally ordered po-relation $\ordern{3m-1}$,
  and the relation name $S'$ by the totally ordered po-relation 
  constructed from the instance~$\calI$ as
  follows: for $1 \leq i \leq 3m$, consider the concatenation of one tuple
  $\id^i_1$ with value $\s$, $n_i$~tuples $\id^i_j$ (with $2 \leq j \leq n_i+1$)
  with value $\n$, and one tuple $\id^i_{n_i+2}$ with value $\e$, and define
  the interpretation of~$S'$ 
  by concatenating the $3m$ sequences of length $n_i+2$.
  Consider the query $Q \defeq \Pi_2(S \times_{\gen} S')$, where $\Pi_2$
  projects to the attribute of the relation $S'$. See
  Figure~\ref{fig:gridpic} for an illustration with $E =
  (1, 1, 2)$ and $B = 4$.

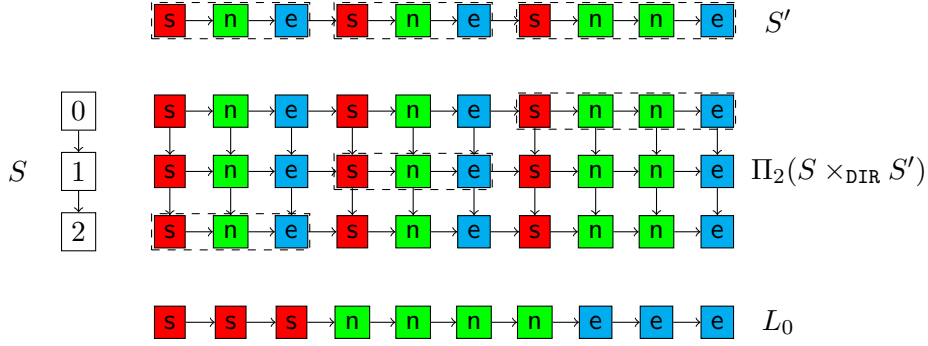
\begin{figure}
  \centering
  \begin{tikzpicture}[scale=.8]
    \node (S) at (-1.5, -1) {$S$};
    \node[draw] (S0) at (-.5, 0) {$0$};
    \node[draw] (S1) at (-.5, -1) {$1$};
    \node[draw] (S2) at (-.5, -2) {$2$};
    \draw[->] (S0) -- (S1);
    \draw[->] (S1) -- (S2);

    \node (S') at (11, 1.5) {$S'$};

    \node[draw,fill=red] (R1s) at (1, 1.5) {$\s$};
    \node[draw,fill=green] (R1n) at (2, 1.5) {$\n$};
    \node[draw,fill=cyan] (R1e) at (3, 1.5) {$\e$};
    \draw[->] (R1s) -- (R1n);
    \draw[->] (R1n) -- (R1e);

    \node[draw,fill=red] (R2s) at (4, 1.5) {$\s$};
    \node[draw,fill=green] (R2n) at (5, 1.5) {$\n$};
    \node[draw,fill=cyan] (R2e) at (6, 1.5) {$\e$};
    \draw[->] (R2s) -- (R2n);
    \draw[->] (R2n) -- (R2e);

    \node[draw,fill=red] (R3s) at (7, 1.5) {$\s$};
    \node[draw,fill=green] (R3n1) at (8, 1.5) {$\n$};
    \node[draw,fill=green] (R3n2) at (9, 1.5) {$\n$};
    \node[draw,fill=cyan] (R3e) at (10, 1.5) {$\e$};
    \draw[->] (R3s) -- (R3n1);
    \draw[->] (R3n1) -- (R3n2);
    \draw[->] (R3n2) -- (R3e);

    \draw[->] (R1e) -- (R2s);
    \draw[->] (R2e) -- (R3s);

    \draw[dashed] (0.7, 1.2) rectangle (3.3, 1.8);
    \draw[dashed] (3.7, 1.2) rectangle (6.3, 1.8);
    \draw[dashed] (6.7, 1.2) rectangle (10.3, 1.8);

    \node[draw,fill=red] (R1sb) at (1, 0) {$\s$};
    \node[draw,fill=green] (R1nb) at (2, 0) {$\n$};
    \node[draw,fill=cyan] (R1eb) at (3, 0) {$\e$};
    \draw[->] (R1sb) -- (R1nb);
    \draw[->] (R1nb) -- (R1eb);

    \node[draw,fill=red] (R2sb) at (4, 0) {$\s$};
    \node[draw,fill=green] (R2nb) at (5, 0) {$\n$};
    \node[draw,fill=cyan] (R2eb) at (6, 0) {$\e$};
    \draw[->] (R2sb) -- (R2nb);
    \draw[->] (R2nb) -- (R2eb);

    \node[draw,fill=red] (R3sb) at (7, 0) {$\s$};
    \node[draw,fill=green] (R3n1b) at (8, 0) {$\n$};
    \node[draw,fill=green] (R3n2b) at (9, 0) {$\n$};
    \node[draw,fill=cyan] (R3eb) at (10, 0) {$\e$};
    \draw[->] (R3sb) -- (R3n1b);
    \draw[->] (R3n1b) -- (R3n2b);
    \draw[->] (R3n2b) -- (R3eb);

    \draw[->] (R1eb) -- (R2sb);
    \draw[->] (R2eb) -- (R3sb);

    \draw[dashed] (0.7, -1.7) rectangle (3.3, -2.3);
    \draw[dashed] (3.7, -0.7) rectangle (6.3, -1.3);
    \draw[dashed] (6.7, .3) rectangle (10.3, -.3);

    \node[draw,fill=red] (R1sc) at (1, -1) {$\s$};
    \node[draw,fill=green] (R1nc) at (2, -1) {$\n$};
    \node[draw,fill=cyan] (R1ec) at (3, -1) {$\e$};
    \draw[->] (R1sc) -- (R1nc);
    \draw[->] (R1nc) -- (R1ec);

    \node[draw,fill=red] (R2sc) at (4, -1) {$\s$};
    \node[draw,fill=green] (R2nc) at (5, -1) {$\n$};
    \node[draw,fill=cyan] (R2ec) at (6, -1) {$\e$};
    \draw[->] (R2sc) -- (R2nc);
    \draw[->] (R2nc) -- (R2ec);

    \node[draw,fill=red] (R3sc) at (7, -1) {$\s$};
    \node[draw,fill=green] (R3n1c) at (8, -1) {$\n$};
    \node[draw,fill=green] (R3n2c) at (9, -1) {$\n$};
    \node[draw,fill=cyan] (R3ec) at (10, -1) {$\e$};
    \draw[->] (R3sc) -- (R3n1c);
    \draw[->] (R3n1c) -- (R3n2c);
    \draw[->] (R3n2c) -- (R3ec);

    \draw[->] (R1ec) -- (R2sc);
    \draw[->] (R2ec) -- (R3sc);

    \node[draw,fill=red] (R1sd) at (1, -2) {$\s$};
    \node[draw,fill=green] (R1nd) at (2, -2) {$\n$};
    \node[draw,fill=cyan] (R1ed) at (3, -2) {$\e$};
    \draw[->] (R1sd) -- (R1nd);
    \draw[->] (R1nd) -- (R1ed);

    \node[draw,fill=red] (R2sd) at (4, -2) {$\s$};
    \node[draw,fill=green] (R2nd) at (5, -2) {$\n$};
    \node[draw,fill=cyan] (R2ed) at (6, -2) {$\e$};
    \draw[->] (R2sd) -- (R2nd);
    \draw[->] (R2nd) -- (R2ed);

    \node[draw,fill=red] (R3sd) at (7, -2) {$\s$};
    \node[draw,fill=green] (R3n1d) at (8, -2) {$\n$};
    \node[draw,fill=green] (R3n2d) at (9, -2) {$\n$};
    \node[draw,fill=cyan] (R3ed) at (10, -2) {$\e$};
    \draw[->] (R3sd) -- (R3n1d);
    \draw[->] (R3n1d) -- (R3n2d);
    \draw[->] (R3n2d) -- (R3ed);

    \draw[->] (R1ed) -- (R2sd);
    \draw[->] (R2ed) -- (R3sd);

    \draw[->] (R1sb) -- (R1sc);
    \draw[->] (R1nb) -- (R1nc);
    \draw[->] (R1eb) -- (R1ec);
    \draw[->] (R2sb) -- (R2sc);
    \draw[->] (R2nb) -- (R2nc);
    \draw[->] (R2eb) -- (R2ec);
    \draw[->] (R3sb) -- (R3sc);
    \draw[->] (R3n1b) -- (R3n1c);
    \draw[->] (R3n2b) -- (R3n2c);
    \draw[->] (R3eb) -- (R3ec);

    \draw[->] (R1sc) -- (R1sd);
    \draw[->] (R1nc) -- (R1nd);
    \draw[->] (R1ec) -- (R1ed);
    \draw[->] (R2sc) -- (R2sd);
    \draw[->] (R2nc) -- (R2nd);
    \draw[->] (R2ec) -- (R2ed);
    \draw[->] (R3sc) -- (R3sd);
    \draw[->] (R3n1c) -- (R3n1d);
    \draw[->] (R3n2c) -- (R3n2d);
    \draw[->] (R3ec) -- (R3ed);
    \node (Q) at (12, -1) {$\Pi_2(S \times_{\dir} S')$};

    \node (L') at (11, -3.5) {$L_0$};

    \node[draw,fill=red] (R1sl) at (1, -3.5) {$\s$};
    \node[draw,fill=red] (R1nl) at (2, -3.5) {$\s$};
    \node[draw,fill=red] (R1el) at (3, -3.5) {$\s$};
    \draw[->] (R1sl) -- (R1nl);
    \draw[->] (R1nl) -- (R1el);

    \node[draw,fill=green] (R2sl) at (4, -3.5) {$\n$};
    \node[draw,fill=green] (R2nl) at (5, -3.5) {$\n$};
    \node[draw,fill=green] (R2el) at (6, -3.5) {$\n$};
    \draw[->] (R2sl) -- (R2nl);
    \draw[->] (R2nl) -- (R2el);

    \node[draw,fill=green] (R3sl) at (7, -3.5) {$\n$};
    \node[draw,fill=cyan] (R3n1l) at (8, -3.5) {$\e$};
    \node[draw,fill=cyan] (R3n2l) at (9, -3.5) {$\e$};
    \node[draw,fill=cyan] (R3el) at (10, -3.5) {$\e$};
    \draw[->] (R3sl) -- (R3n1l);
    \draw[->] (R3n1l) -- (R3n2l);
    \draw[->] (R3n2l) -- (R3el);

    \draw[->] (R1el) -- (R2sl);
    \draw[->] (R2el) -- (R3sl);
  \end{tikzpicture}
  \caption{Example for the proof of Theorem~\ref{thm:posscompextend1}.}
  \label{fig:gridpic}
\end{figure}

  We define the candidate possible world $L$ as the list relation 
  $L\coloneqq L_1 L'
  L_2$, with $L_1$, $L'$, and $L_2$ defined 
  as follows.

  \begin{itemize}
    \item $L_1$ is a list relation defined as the concatenation, for $1
  \leq i \leq 3m$, of $3m-i$ copies of the following 
  sublist: one tuple with value~$\s$, $n_i$
  tuples with value~$\n$, and one tuple with value~$\e$.
    \item $L_2$ is a list relation defined like~$L_1$, except that
      $3m-i$ is replaced by $i-1$.
    \item $L_0$ is the list relation consisting of three tuples with value $\s$, $B$ tuples with value
      $\n$, three tuples with value~$\e$. See Figure~\ref{fig:gridpic} for an
      illustration of~$L_0$.
    \item $L'$ is the list relation defined as the concatenation of
      $m$ copies of~$L_0$.
  \end{itemize}

  We now consider the \poss instance that asks whether
  $L$ is a possible world of the query $Q$ on the po-database~$D$.
  We claim that this \poss instance is positive
  iff the original UNARY-3-PARTITION instance $\calI$ is positive. As the
  reduction process described above is
  clearly PTIME, the only thing left to prove Theorem~\ref{thm:posscompextend1}
  is to show this claim, which we now do.
  
  Denote by $\OR'$ the po-relation obtained by evaluating~$Q(D)$,
  and note that all tuples of~$\OR'$ have value in $\{\s, \n, \e\}$.
  For $0 \leq k \leq \card{L_1}$, we write $L_1^{\leq k}$ for the prefix of $L_1$ of length $k$.
  We say that $L_1^{\leq k}$ is a \emph{whole prefix} if either $k
  = 0$ (that is, the empty prefix) or the $k$-th symbol of $L_1$ has value $\e$.
  We say that a linear extension $L''$ of $\OR'$ \emph{realizes} $L_1^{\leq k}$ if
  the sequence of its $k$-th first values is $L_1^{\leq k}$, and that it
  realizes $L_1$ if it realizes $L_1^{\leq \card{L_1}}$. When $L''$ realizes
  $L_1^{\leq k}$, we call the \emph{matched} elements the elements of~$\OR'$ that
  occur in the first $k$ positions of~$L''$, and say that the other elements are
  \emph{unmatched}. For $1 \leq i \leq 3m$, we call the \emph{$i$-th row} of~$\OR'$ the elements whose
  first component before projection was $i-1$: note that, for each~$i$, the
  po-relation $\OR'$
  imposes a total order on the $i$-th row. We define the \emph{row-$i$ matched
  elements} to refer to the elements on row-$i$ that are matched, and define
  analogously the \emph{row-$i$ unmatched elements}.

  We first observe that for any linear extension $L''$ realizing $L_1^{\leq k}$,
  for all $i$, writing the $i$-th row as $\id'_1 < \ldots <
  \id'_{\card{S'}}$, the unmatched elements must be all of the form $\id'_j$ for
  $k_i < j \leq \card{S'}$ for some $0 \leq k_i \leq \card{S'}$, i.e., they must be a prefix of the total order of the $i$-th row. Indeed, if they did not form a
  prefix, then some order constraint of~$\OR'$ would have been violated when
  enumerating $L''$. Further, by cardinality we clearly have $\sum_{i = 1}^{3m} k_i=k$. 

  Second, when a linear extension $L''$ of~$\OR'$ realizes $L_1^{\leq k}$, we say
  that we are in a \emph{whole situation} for~$k$ if for all $i$,
  either 
  the first row-$i$ unmatched
  element $\id'_{k_i+1}$ has value~$\s$
  or there are no row-$i$ unmatched elements (and we write $k_i \colonequals
  \card{S'}$).
  When we are in
  a whole situation for~$k$, the condition on~$k_i$
  means by definition that we must have $k_i = \sum_{j=1}^{l_i} (n_j+2)$ for
  some~$1 \leq l_i \leq 3m$; in this case, letting $S_i$ be the multiset of the $n_j$ for $1 \leq j \leq
  l_i$, we call $S_i$
  the 
  bag of \emph{row-$i$ consumed integers} at~$k$. The \emph{row-$i$ remaining
  integers} at~$k$
  are $E \setminus S_i$, where we see~$E$ as a multiset and define the
  difference operator on multisets
  by subtracting the multiplicities in~$S_i$ to the multiplicities
  in~$E$.

  We now prove the following claim: for any linear extension of~$\OR'$ realizing~$L_1$,
  we are in a whole situation for~$\card{L_1}$,
  and the multiset union $\biguplus_{1 \leq i
  \leq 3m} S_i$ of the row-$i$ consumed integers at~$k$
  is equal to the multiset obtained by repeating $3m-i$ times the integer $n_i$ of~$E$
  for all $1\leq i\leq 3m$.

  We prove the first part of the claim by showing it for all whole prefixes
  $L_1^{\leq k}$, by induction on $k$. It is certainly the case for $L_1^{\leq
  0}$ (the empty prefix). Now, assuming that it holds for prefixes of length up to
  $l$, to realize a whole prefix $L^{\leq l'}$ with $l' > l$, we must first realize a
  strictly shorter whole prefix $L^{\leq l''}$ with $l'' \leq l$ (take it to be of maximal
  length), so by induction hypothesis we are in a whole situation for~$l''$ when
  realizing $L^{\leq l''}$. Now to realize the whole prefix $L^{\leq l'}$ having
  realized the whole prefix $L^{\leq l''}$, by construction of~$L_1$, the
  sequence~$L''$ of additional values to realize is $\s$, a certain number of~$\n$'s, and
  $\e$. It is now clear that this must bring us from a whole situation to
  a whole situation: since there is only one $\s$ in $L''$, there is only one
  row such that an $\s$ value becomes matched; now, to match the additional
  $\n$'s and $\e$, only the elements of this particular row can be used, as any
  first unmatched element (if any) of all other rows is $\s$, and we must use
  the sequence of $\n$-labeled elements followed by the $\e$-labeled element of
  the row. Hence the first part of the claim is proved.

  To prove the second part of the claim, observe that whenever we go from a
  whole prefix to a whole prefix by additionally matching $\s$, $n_j$ times
  $\n$, and $\e$, then we add to $S_i$ the integer~$n_j$. So the claim holds by
  construction of $L_1$.

  A similar argument shows that for any linear extension $L''$ of $\OR'$ whose first
  $\card{L_1}$ tuples achieve $L_1$ and whose last $\card{L_2}$ tuples achieve $L_2$,
  for each $1 \leq i \leq 3m$,
  extending the definition of 
  the row-$i$ unmatched elements to refer to the elements that are matched
  neither to~$L_1$ nor to~$L_2$, these elements must form
  a contiguous sequence $\id'_j$ with $k_i < j <
  m_i$ for some $0 \leq k_i < m_i \leq \card{S'}+1$: here $k_i$ refers to the last
  element of row~$i$ matched to~$L_1$ (or~$0$ if none are), and $m_i$ to the first element of row~$i$
  matched to~$L_2$ (or $\card{S'}+1$ if none are).
  In addition, if we have $k_i<m_i-1$, then $\id'_{k_i}$ has value $\e$
  and $\id'_{m_i}$ has value $\s$, and the unmatched values (whose definition is
  extended in an
  analogous fashion) are a multiset corresponding exactly to the elements
  $n_1,\dots,n_{3m}$: indeed, each integer $n_i$ of~$E$ is matched $3m-i$ times
  within~$L_1$ and $i-1$ times in~$L_2$, so $3m-i+i-1 = 3m-1$ times overall,
  whereas it occurs $3m$ times in the grid.
  So the unmatched elements when having read $L_1$ (at
  the beginning) and $L_2$ (at the end) are formed of $3m$
  sequences, of length $n_i + 2$ for $1 \leq i \leq 3m$, of the
  form $\s$, $n_i$ times $\n$, and $\e$: each of the $3m$ sequences is totally
  ordered (as it occurs as consecutive elements in some row), and there is a
  certain order relation across the sequences depending on the rows where they
  are: the comparability relations exist across sequences that are on the same
  row, or that are in different rows but where comparability holds by definition
  of $\times_\gen$.

  Observe now that there is a way to achieve $L_1$ and~$L_2$ while ensuring that
  there are no order constraints across the sequences of unmatched elements,
  i.e., the only order constraints within the unmatched elements are those given
  by the total order on each sequence. To do so,
  we achieve $L_1$ by picking the following,
  in that order: for $1\leq j\leq 3m$, for $1\leq i\leq 3m-j$, pick the first
  $n_j+2$ unmatched tuples of row $i$. Similarly, to achieve $L_2$ at the end,
  we can pick the following, in \emph{reverse} order: for $3m\geq j\geq 1$, for
  $3m\geq i\geq 3m-j+1$, the last $n_j + 2$ unmatched tuples of row $i$. When we
  pick elements this way, the unmatched elements are $3m$ lists
  (one for each row, with that of row $i$ being $\s$, $n_i$ times $\n$
  and $\e$, for all $i$) and there are no order relations across sequences.
  We let $\OR$ be the sub-po-relation of $\OR'$ that consists of exactly these unmatched
  elements: it is illustrated in Figure~\ref{fig:gridpic} as the elements of the
  grid that are in the dashed rectangles. Formally, $\OR$ is the parallel
  composition of $3m$ totally ordered po-relations which we will call $\OR_i$
  for $1 \leq i \leq 3m$: the elements of $\OR_i$ consist of an element labeled~$\s$ followed by $n_i$ elements
  labeled~$\n$ and one element labeled~$\e$. 

  We now claim that for any list relation $L''$, the concatenation $L_1 L'' L_2$ is a
  possible world of~$\OR'$ if and only if $L''$ is a possible world of $\OR$.
  The ``if'' direction was proved with the construction above, and the ``only if''
  holds because $\OR$ is the \emph{least constrained}
  possible po-relation for the unmatched sequences: recall that the only
  comparability relations that it contains are those on the sequences of
  unmatched elements, which are known to be total orders.
  Hence, to prove our original claim, it only remains to show that the
  UNARY-3-PARTITION instance $\calI$ is positive iff $L'$ is a possible world
  of~$\OR$, which we now do.

  For the forward direction, we show that, if $\calI$ is a positive
  instance of \mbox{UNARY-3-PARTITION}, then there is a linear extension $<'$ of~$<$
  which witnesses that $L' \in \pw(\OR)$. Indeed, consider a
  3-partition $\mathbf{p} = (p^i_1, p^i_2, p^i_3)$ for $1 \leq i \leq m$, with $n_{p^i_1} +
  n_{p^i_2} + n_{p^i_3} = B$ for all $1 \leq i \leq m$, and each integer of $\{1, \ldots, 3m\}$ occurring exactly
  once in $\mathbf{p}$. We can realize $L'$ from~$\mathbf{p}$ by picking
  successively the following for $1
  \leq i \leq m$ to realize~$L_0$: the three $\s$-labeled elements of the
  po-relations~$\OR_{p^i_q}$ for $1 \leq q \leq 3$, then the $\n$-labeled
  elements of these same po-relations (this is $B$ tuples in total, because
  $\mathbf{p}$ is a solution to~$\calI$), and last the three $\e$-labeled
  elements of these po-relations.
  
  For the backward direction, we show that, if there is a linear extension $<'$ of~$<$ which
  witnesses that $L' \in \pw(\OR)$, then we can build a 3-partition
  $\mathbf{p} = (p^i_1, p^i_2, p^i_3)$ for $1 \leq i \leq m$ which satisfies the
  conditions above. To see why, we first observe that, for each $1 \leq i \leq
  m$, considering the $i$-th occurrence of the sublist $L_0$ in~$L'$, there
  must be three distinct values $p^i_1, p^i_2, p^i_3$, such that the elements
  which occur in~$<'$ at the positions of the
  value~$\n$ in this occurrence of~$L_0$
  are precisely the $\n$-labeled elements of the po-relations $\OR_{p^i_1}$, 
  $\OR_{p^i_2}$, and $\OR_{p^i_3}$. Indeed, we show this
  claim for increasing values of~$i$, from $i = 1$ to $i = m$. Just before
  we consider each occurrence of~$L_0$, and just after we have considered
  it, we will ensure the invariant that, for all $1 \leq i \leq
  3m$, either all elements of $\OR_i$ have been enumerated or none have: this
  invariant is clearly true initially because nothing is enumerated yet. 
  Now, considering the $i$-th
  occurrence of~$L_0$ for some $1 \leq i \leq m$, we define $p^i_1,
  p^i_2, p^i_3$, such that the elements $\s^3$ in this occurrence of~$L_0$
  are mapped to the $\s$-labeled elements of~$\OR_{p^i_1}$, $\OR_{p^i_2}$, and
  $\OR_{p^i_3}$: they must indeed
  be mapped to such elements because they are the only ones with
  value~$\s$. 
  Now, the $\n$-labeled elements of these three po-relations can all be
  enumerated (indeed, we have just enumerated the $\s$-labeled elements that
  precede them), and
  they are the only elements with value~$\n$ that can be enumerated,
  thanks to the invariant: the others either have already been enumerated
  or have a predecessor with value~$\s$ that has not been enumerated yet.
  Further, all
  elements of this form must be enumerated, because this is the only
  possible way for us to finish matching~$L_0$ and enumerate three elements
  with value~$\e$, namely, those of the three po-relations $\OR_{p^i_1}$,
  $\OR_{p^i_2}$, and $\OR_{p^i_3}$: this uses the invariant again to justify
  that they are the only elements with value~$\e$ that can be enumerated at this
  stage. We are now done with the $i$-th occurrence of~$L_0$, and clearly the
  invariant is satisfied on the result, because the elements that we have
  enumerated while matching this occurrence of~$L_0$ are all the elements of $\OR_{p^i_1}$,
  $\OR_{p^i_2}$, and $\OR_{p^i_3}$.
  
  Now that we have defined
  the 3-partition $\mathbf{p}$, it is clear by definition of a linear extension
  that each position $1 \leq i \leq 3m$, i.e., each number occurrence
  in~$E$, must occur exactly once in~$\mathbf{p}$. Further, as $<'$
  achieves~$L_0$, by considering each occurrence of~$L_0$, we know that, for $1
  \leq i \leq m$, we have $p^i_1 + p^i_2 + p^i_3 = B$. Hence, $\mathbf{p}$
  witnesses that $\calI$ is a positive instance to the UNARY-3-PARTITION problem.
  
  Hence, it is indeed the case that
  $\calI$ is a positive UNARY-3-PARTITION instance iff $L' \in \pw(\OR)$, which 
  is the case iff $L_1 L' L_2$ is a possible world of~$\OR'$, i.e., iff $L$ is a
  possible world of~$Q(D)$. This establishes the correctness of the reduction
  for \PosRA, showing that the \poss problem for \PosRA queries is NP-hard.
\end{proof}

\subsection{Disallowing Both Products}
We have shown the tractability of \poss without the $\times_\gen$
operator, when the input po-relations are assumed to have bounded width.
We now study the fragment \Pnoprod without both kinds of product,
and show that this \poss is tractable for this fragment even for
more general input po-relations. Specifically, we will allow input po-relations
that are almost totally ordered, i.e., have bounded \emph{width}; and we will also
allow input po-relations that are almost unordered, which we measure using a new
order-theoretic notion of \emph{ia-width}. The idea of ia-width is to decompose
the relation in classes of indistinguishable sets of incomparable elements.

\begin{definition}
  \label{def:iawidth}
  Given a poset $P = (\ID, <)$, a subset $A \subseteq \ID$ is an
  \emph{antichain} if there are no
  $x, y \in A$ such that $x < y$. It is an \emph{indistinguishable set} (or an
  \emph{interval}~\cite{fraisse1984intervalle}) if, for all
  $x, y \in A$ and $z \in \ID \backslash A$,
  we have $x < z$ iff $y < z$, and $z < x$ iff $z < y$.
  It is an \emph{indistinguishable antichain} if it is both an antichain and an
  indistinguishable set.

  An \emph{indistinguishable antichain partition} (ia-partition)
  of~$P$ is a
  partition of~$\ID$ into
  indistinguishable antichains.
  The 
  \emph{cardinality} of the partition is the number of antichains.
  The \emph{ia-width} of~$P$ is the
  cardinality of its smallest ia-partition.
  The \emph{ia-width} of a po-relation is that of its underlying poset,
  and the \emph{ia-width} of a po-database is the maximal ia-width of its po-relations.
\end{definition}

Hence, any po-relation~$\OR$ has ia-width at most~$\card{\OR}$, with the trivial
ia-partition consisting of singleton indistinguishable antichains, and unordered
po-relations have an ia-width of 1.
Po-relations may have low ia-width in practice if order is completely
unknown except for a 
few comparability pairs given by users, or when they consist of objects from a
constant number of types that are ordered based only on some order on the types.

We can now state our tractability result when disallowing both kinds of
products, and allowing both bounded-width and bounded-ia-width relations. For
instance, this result allows us to 
combine sources whose order is fully unknown or irrelevant, with
sources that are completely ordered (or almost totally ordered).

\begin{theorem}\label{thm:aggregnoprod}
  For any fixed $k \in \mathbb{N}$
  and fixed \Pnoprod query $Q$,
  the \poss problem for~$Q$ is in PTIME
  when each po-relation of the input po-database
  has either ia-width
  $\leq k$ or width~$\leq k$.
\end{theorem}

To prove this result, we start by making a simple observation.
  
\begin{lemma}
  \label{lem:rewritenoprod}
  Any \Pnoprod query $Q$ can be equivalently rewritten as a union of
  projections of selections of a constant number of input relations and constant
  relations.
\end{lemma}

\begin{proof}
  For the semantics that we have defined for
  operators, it it easy to show that selection commutes with union, selection
  commutes with projection, and projection commutes with union. Hence, we can
  perform the desired rewriting.
\end{proof}

We can thus rewrite the input query using this lemma. The idea is that we will
evaluate the query in PTIME using Proposition~\ref{prp:repsys}, argue that
the width bounds are preserved using Lemma~\ref{lem:lexwidth}, and compute a
chain partition of the relations using Dilworth's theorem. Let us first show
an analogue of Lemma~\ref{lem:lexwidth}
for the new notion of ia-width.

\begin{lemma}\label{lem:lexiawidthnoprod}
  Let $k \geq 2$ and $Q$ be a \Pnoprod query. 
  For any po-database~$D$ of ia-width~$\leq k$,
  the po-relation $Q(D)$ has ia-width $\leq \max(k,q) \times \card{Q}$, where 
$q$ denotes the largest value such that
  $\ordern{q}$ appears in~$Q$. 
\end{lemma}

\begin{proof}
  We first show by induction on~$Q$ that the ia-width of the query output can be
  bounded by a function of~$k$.
  We show the base cases.

  \begin{itemize}
    \item The input relations have ia-width at most~$k$.
    \item The constant relations have ia-width $\leq q$ with the trivial
      ia-partition consisting of singleton classes.
  \end{itemize}
We then show the induction step.

  \begin{itemize}
    \item Projection clearly does not change ia-width.
    \item Selection may only decrease the ia-width. Indeed, consider an
      ia-partition of the input po-relation, apply the selection to each class,
      and remove the classes that became empty. The number of classes has not
      increased, and it is clear that the result is still an ia-partition of the
      output po-relation.
    \item The union of two relations
  with ia-width $k_1$ and $k_2$ has ia-width at most $k_1 + k_2$. Indeed, we can
      obtain an ia-partition for the union as the union of ia-partitions for the
      input relations.
  \end{itemize}

  Second, we see that the bound $\max(k,q) \times \card{Q}$ on the ia-width
  of~$Q(D)$ is
  clearly correct, because the base cases have ia-width $\leq \max(k, q)$ and
  the worst operators are unions, which amount to summing the ia-width bounds on
  all inputs, of which there are $\leq \card{Q}$. So we have shown the desired
  bound.
\end{proof}

We next show that, like chain partitions for bounded-width po-relations, we can
efficiently compute an ia-partition for a bounded-ia-width po-relation.

\begin{proposition}\label{prp:ptime-ia-partition}
  The ia-width of any poset and a corresponding ia-partition can be computed
  in PTIME.
\end{proposition}

To show this result, we need two preliminary observations about
indistinguishable antichains.

\begin{lemma}
  \label{lem:iasubset}
  For any poset $(\ID, <)$ and indistinguishable antichain $A$, any $A'
  \subseteq A$ is an indistinguishable antichain.
\end{lemma}

\begin{proof}
  Clearly $A'$ is an antichain because $A$ is. We show that it is an
  indistinguishable set. Let $x, y \in A'$ and $z \in \ID \backslash A'$, and show
  that $x < z$ implies $y < z$ (the other three implications are symmetric). If
  $z \in \ID \backslash A$, then we conclude because $A$ is an indistinguishable set.
  If $z \in A \backslash A'$, then we conclude because, as $A$ is an antichain, $z$
  is incomparable both to $x$ and to $y$.
\end{proof}

\begin{lemma}
  \label{lem:indistinguishablea}
  For any poset $(\ID, <)$ and indistinguishable antichains $A_1, A_2 \subseteq
  \ID$ such that $A_1 \cap A_2 \neq \emptyset$, the union $A_1 \cup A_2$ is an
  indistinguishable antichain.
\end{lemma}

\begin{proof}
  We first show that $A_1 \cup A_2$ is an indistinguishable set. 
  Let $x, y \in A_1 \cup A_2$ and
  $z \in \ID \backslash (A_1 \cup A_2)$, assume that $x < z$ and show that $y <
  z$ (again the other three implications are symmetric).
  As $A_1$ and $A_2$ are indistinguishable sets, this is immediate unless $x
  \in A_1 \backslash A_2$ and $y \in A_2 \backslash A_1$, or vice-versa. We
  assume the first case as the second one is symmetric. Consider
  $w \in A_1 \cap A_2$. As $x < z$, we know that $w < z$ because $A_1$ is an
  indistinguishable set, so that $y < z$ because $A_2$ is an indistinguishable
  set, which proves the desired implication.

  Second, we show that $A_1 \cup A_2$ is an antichain.
  Proceed by contradiction, and let $x, y \in
  A_1 \cup A_2$ such that $x < y$. As $A_1$ and $A_2$ are antichains, we must
  have $x \in A_1 \backslash A_2$ and $y \in A_2 \backslash A_1$, or vice-versa.
  Assume the first case, the second case is symmetric. As $A_1$ is an
  indistinguishable set, letting $w \in A_1 \cap A_2$,
  as $x < y$ and $x \in A_1$, we have $w < y$. But $w \in A_2$ and
  $y \in A_2$, which is impossible because $A_2$ is an antichain. We have
  reached a contradiction, so we cannot have $x < y$. Hence, $A_1 \cup A_2$ is
  an antichain, which concludes the proof.
\end{proof}

We can now show Proposition~\ref{prp:ptime-ia-partition}.

\begin{proof}
  Start with the trivial partition in singletons (which is an
  ia-partition), and for every pair of items, see if their
  current classes can be merged (i.e., merge them, check in PTIME if it
  is an antichain, and if it is an indistinguishable set, and undo the merge if
  it is not). Repeat the process
  while it is possible to merge classes (i.e., at most linearly many times).
  This greedy process concludes in PTIME and yields an ia-partition
  $\mathbf{A}$. Let $n$ be its cardinality.

  Now assume that there is an ia-partition $\mathbf{A'}$ of cardinality $m < n$.
  There has to be a class $A'$ of~$\mathbf{A'}$
  which intersects two different classes $A_1 \neq A_2$ of
  the greedy ia-partition $\mathbf{A}$, otherwise $\mathbf{A'}$ would be a refinement
  of $\mathbf{A}$ so we would have $m \geq n$.
  Now, by Lemma~\ref{lem:indistinguishablea},
  $A \cup A_1$ and $A \cup A_2$, and hence
  $A \cup A_1 \cup A_2$, are
  indistinguishable antichains.
  By Lemma~\ref{lem:iasubset},
  this implies that $A_1 \cup A_2$ is an indistinguishable
  antichain. Now, when constructing the greedy ia-partition $\mathbf{A}$,
  the algorithm has
  considered one element of $A_1$ and one element of $A_2$, attempted to merge
  the classes $A_1$ and $A_2$,
  and, since it has not merged them in~$\mathbf{A}$, the union $A_1 \cup A_2$
  cannot be an indistinguishable antichain. We have reached a contradiction, so
  we cannot have $m < n$, which concludes the proof.
\end{proof}

We have shown the preservation of ia-width bounds through selection, projection,
and union (Lemma~\ref{lem:lexiawidthnoprod}), and shown how to compute an
ia-partition in PTIME (Proposition~\ref{prp:ptime-ia-partition}). Let us now return
to the proof of Theorem~\ref{thm:aggregnoprod}. We use
Lemma~\ref{lem:rewritenoprod} to rewrite the query to a union of projection of
selections. We evaluate the selections and projections in PTIME by
Proposition~\ref{prp:repsys}. As union is clearly associative and commutative,
we evaluate the union of relations of width $\leq k$, yielding $\OR$, and the union of those
of ia-width $\leq k$, yielding~$\OR'$. The first result $\OR$ has bounded width thanks to
Lemma~\ref{lem:lexwidth}, and we can compute a chain partition of it in PTIME using
Dilworth's theorem. The second result has bounded ia-width thanks to
Lemma~\ref{lem:lexiawidthnoprod}, and we can compute an ia-partition of it in
PTIME using Proposition~\ref{prp:ptime-ia-partition}. Hence, to show 
Theorem~\ref{thm:aggregnoprod}, it suffices to show the following strengthening
of Lemma~\ref{lem:aggregwinst}.

\begin{lemma}
  \label{lem:aggregwuawb}
  For any constant $k \in \mathbb{N}$,
  we can determine in PTIME,
  for any input po-relation $\OR$ with width $\leq k$,
  input po-relation $\OR'$ with ia-width $\leq k$,
  and list relation $L$,
  whether $L \in \pw(\OR \cupgen \OR')$.
\end{lemma}

\begin{proof}
  We first show the result when assuming that $\OR$ is empty, and will later return to the
  general case.
  Let $\mathbf{A} = (A_1, \ldots, A_k)$ be an ia-partition of width $k$ of $\OR'
  = (\ID, T, <)$, which can be computed in PTIME
  by Proposition~\ref{prp:ptime-ia-partition}. We assume that the
  length of the candidate possible world $L$ is $\card{\ID}$,
  as we can trivially reject otherwise.

  For any linear extension $<'$ of~$\OR'$, we define the
  \emph{finishing order} of $<'$ as
  the permutation $\pi$ of $\{1, \ldots, k\}$ obtained by
  considering, for each class $A_i$ of~$\mathbf{A}$, the largest position $1
  \leq n_i \leq \card{\ID}$
  in~$<'$ to which an element of $A_i$ is mapped, and sorting the
  class indexes in ascending order according to this largest position. We say we can realize $L$ with
  finishing order $\pi$ if there is a linear extension of~$\OR'$ that realizes
  $L$ and whose finishing order
  is~$\pi$. Hence, it suffices to check, for every possible permutation $\pi$ of
  $\{1, \ldots, k\}$,
  whether $L$ can be realized from $\OR'$ with finishing order $\pi$: this does not
  make the complexity worse because the number of finishing orders depends only on
  $k$ and not on $\OR'$, so it is constant. (Note that the order relations across
  classes may imply that some finishing orders are impossible to realize
  altogether.)

  We now claim that to determine whether $L$ can be realized with finishing
  order~$\pi$, the following greedy algorithm works. Read $L$ linearly. At any
  point, maintain the set of elements of $\OR'$ that have already been
  enumerated
  (distinguish the \emph{used} and \emph{unused} elements; initially all
  elements are unused), and distinguish the classes of~$\mathbf{A}$ in three
  kinds: the \emph{exhausted
  classes}, where all elements are used; the \emph{open classes},
  the ones where some elements are unused and all ancestor elements outside of
  the class are used;
  and the \emph{blocked
  classes}, where some ancestor element outside of the class is not used.
  Initially, the
  open classes are those which are roots in the poset obtained from the
  underlying poset of~$\OR'$ by taking the quotient by the equivalence relation induced by
  $\mathbf{A}$; and the other classes are blocked.

  When reading a value $t$ from $L$, consider all open classes. If none of these
  classes have an unused element with value $t$, reject, i.e., conclude that we cannot realize $L$
  as a possible world of $\OR'$ with finishing order $\pi$. Otherwise, take the
  open class that comes first in the finishing order, and use
  an arbitrary suitable element from it. Update the class to be \emph{exhausted} if it
  is: in this case, check that the class was the next one in the finishing
  order~$\pi$ (and reject otherwise), and 
  update from \emph{blocked} to \emph{open} the classes that
  must be. Once $L$ has been completely read, accept: as $\card{L} =
  \card{\ID}$,
  all elements are now used.

  It is clear by construction that, if this greedy algorithm accepts, then there
  is a linear extension of~$\OR'$ that realizes $L$ with finishing order~$\pi$;
  indeed, when the algorithm succeeds, then it has clearly respected the finishing
  order $\pi$, and whenever an identifier $\id$ of~$\OR'$ is marked as
  \emph{used} by the algorithm, then $\id$ has
  the right value relative to the element of~$L$ that has just been read, and
  $\id$
  is in an open class so no order relations of~$\OR'$ are violated by enumerating
  $\id$ at this point of the linear extension. The interesting direction is
  the converse: show that if $L$
  can be realized by a linear extension $<'$ of~$\OR'$ with finishing order $\pi$,
  then the algorithm accepts when considering $\pi$.
  To do so, we must show that if there is such a linear extension,
  then there is such a linear extension where identifiers are enumerated as in
  the greedy algorithm, i.e., we always choose an identifier with the right
  value and in the open class with the smallest finishing time: we call this a
  \emph{minimal} identifier. (Note that we do not need to worry about which
  identifier is chosen: once we have decided on the value of the identifier and
  on its class, it does not matter which element we choose, because all
  elements in the class are unordered and have the same order relations to
  elements outside the class thanks to indistinguishability.)
  If we can prove this, then it justifies the existence of a linear extension
  that the greedy algorithm will construct, which we call a \emph{greedy linear
  extension}.

  Hence, let us see why it is always possible to enumerate minimal identifiers. 
  Consider a linear extension $<'$ and take the smallest position in~$L$ where
  $<'$ chooses an identifier $\id$ which is non-minimal. We know that $\id$ must
  still have the correct value, i.e., $T(\id)$ is determined, and by the definition
  of a linear extension, we know that $\id$ must be in an open class. Hence, we
  know that the class $A$ of~$\id$ is non-minimal, i.e., there is another open
  class $A'$ containing an unused element with value $T(\id)$, and $A'$ is
  before~$A$ in the finishing order~$\pi$. Let us take for $A'$ the first open
  class with such an unused element in the finishing order~$\pi$, and let $\id'$
  be a minimal element, i.e., an element of~$A'$ with $T(\id') = T(\id)$. Let us
  now construct a different linear extension $<''$ by swapping $\id$ and $\id'$,
  i.e., enumerating $\id'$ instead of~$\id$, and enumerating $\id$ in~$<''$ at
  the point where $<'$ enumerates $\id'$. It is clear that the sequence of
  values (images by~$T$) of the identifiers in~$<''$ is still the same as
  in~$<'$. Hence, if we can show that $<''$ additionally satisfies the order
  constraints of~$\OR'$, then we will have justified the existence of a linear
  extension that enumerates minimal identifiers until a later position; so,
  reapplying the rewriting argument, we will deduce the existence of a greedy
  linear extension. So it only remains to show that $<''$ satisfies the order
  constraints of~$\OR'$.

  Let us assume by way of contradiction that $<''$
  violates an order constraint of~$\OR'$. There are two possible kinds of
  violation. The first kind is if $<'$ enumerates an element~$\id''$
  between~$\id$ and~$\id'$ for which $\id < \id''$, so that having $\id'' <''
  \id$ in~$<''$ is a violation. The second kind is if $<'$ enumerates an
  element~$\id''$ between $\id$ and $\id'$ for which $\id'' < \id'$, so that
  having $\id'' <'' \id'$  in~$<''$ is a violation. The second kind of violation
  cannot happen because $\id'$ is in an open class when $<'$
  considers~$\id$, i.e., we have ensured that $\id'$ can be enumerated instead
  of~$\id$. Hence, we focus on violations of the first kind. Consider $\id''$
  such that $\id <' \id'' <' \id'$ and let us show that
  $\id \not<
  \id''$. Letting $A''$ be the class of~$\id''$, we assume that $A'' \neq A$, as
  otherwise there is nothing to show because the classes are antichains.
  Now, we know from~$<'$ that $\id' \not
  <' \id''$, and that the class~$A'$ of~$\id'$ is not exhausted when $<'$
  enumerates~$\id''$. As~$<'$ respects the finishing order~$\pi$, and $A'$ comes
  before~$A$ in~$\pi$, we know that $A$ is not exhausted either when~$<'$
  enumerates~$\id''$. Letting $\id_A$ be an element of~$A$ which is still unused
  when~$<'$ enumerates~$\id''$, we know that $\id_A \not< \id''$. So,
  as $\id'' \notin A$, by indistinguishability, we have $\id \not<
  \id''$. This is what we wanted to show, so $\id''$ cannot witness a
  violation of the first kind.  Hence $<''$  does not violate the order
  constraints of~$\OR'$, and repeating this rewriting argument shows that there
  is a greedy linear extension that the greedy algorithm will find,
  contradicting our assumption. This establishes our result in the case where we
  only have the bounded-ia-width po-relation $\OR'$.

  We now return to the general case where the bounded-width po-relation $\OR$ is not
  empty.
  In this case, we will again enumerate all possible finishing
  orders for the classes of $\OR'$, of which there are constantly many, and apply
  an algorithm for each finishing order $\pi$, with the algorithm succeeding iff it
  succeeds for some finishing order.

  We first observe that if there is a way to achieve $L$ as a possible world of
  $\OR \cup \OR'$ for a finishing order $\pi$, then there is one where the subsequence of
  the tuples that are matched to~$\OR'$ are matched following the greedy strategy as
  we presented before. This is simply because $L$ must then be an
  interleaving of a possible world of $\OR$ and a possible world of $\OR'$, and a
  match for the possible world of $\OR'$ can be found as a greedy match, by what
  was shown above. So it suffices to assume
  that the tuples matched to~$\OR'$ are matched following the greedy algorithm
  that we previously described.

  Second, we observe the following: for any prefix $L'$ of $L$ and order ideal
  $\OR''$ of $\OR$, if we realize $L'$ by matching exactly the tuples of
  $\OR''$ in $\OR$,
  and by matching the other tuples to $\OR'$ following the greedy algorithm, then the
  matched tuples in~$\OR'$ are entirely determined (up to replacing tuples in a
  class by other tuples with the same value). This is because, while there may
  be multiple ways to match parts of $L'$ to $\OR''$ in a way that leaves a
  different sequence of tuples to be matched to $\OR'$, all these ways make us
  match the same bag of tuples to $\OR'$; now the state of $\OR'$ after matching a bag
  of tuples following the greedy algorithm (for a fixed finishing order) is the
  same, no matter the order in which these tuples are matched, assuming that the
  match does not fail.

  This justifies that we can solve the problem with a dynamic algorithm again.
  The state contains the position $\mathbf{m}$ in each chain of $\OR$, and a position $i$
  in the candidate possible world. As in the proof of
  Lemma~\ref{lem:aggregwinst}, we filter the configurations so that they are sane
  with respect to the order constraints between the chains of $\OR$. For each
  state, we will store a Boolean value indicating whether the prefix of length
  $i$ of $L$ can be realized by $\OR \cupgen \OR'$ such that the tuples
  of~$\OR$ that
  are matched is the order ideal $s(\mathbf{m})$
  described by $\mathbf{m}$, and such that the
  other tuples of the prefix are matched to $\OR'$ following the greedy
  algorithm with
  finishing order $\pi$. By
  our second remark above, when the Boolean is true, the state of $\OR'$ is
  uniquely determined, and we also store it as part of the state (it is
  polynomial) so that we do not have to recompute it each time.

  From each state we can make progress by consuming the next tuple from the
  candidate possible world, increasing the length of the prefix, and reaching
  one of the following states: either match the tuple to a chain of $\OR$,
  in which case we make progress in one chain and the consumed tuples in
  $\OR'$
  remain the same; or make progress in $\OR'$, in which case
  we look at the previous state of $\OR'$ that was stored and consume a tuple from
  $\OR'$ following the greedy algorithm:
  more specifically, we find an unused tuple with
  the right label which is in the open class that appears first in the finishing
  order, 
  if the class is now exhausted we verify that it was supposed to be the next
  one according to the finishing order, and
  we update the open, exhausted and blocked status of the
  classes.

  Applying the dynamic algorithm allows us to conclude whether $L$ can be
  realized by matching all tuples of $\OR$, and matching tuples in $\OR'$ following
  the greedy algorithm with finishing order $\pi$ (and checking cardinality
  suffices to ensure that we have matched all tuples of $\OR'$). If the answer of
  the dynamic algorithm is YES, then it is clear that, following
  the path from the initial to the final state found by the dynamic algorithm,
  we can realize $L$. Conversely, if $L$ can be realized, then by our
  preliminary remark it can be realized in a way that matches tuples in
  $\OR'$
  following the greedy algorithm for some finishing order. Now, for that
  finishing order, the path of the dynamic algorithm that matches tuples
  to $\OR$
  or to $\OR'$ following that match will answer YES.
\end{proof}

Disallowing product is 
severe, but we can still integrate sources by taking the
\emph{union} of their tuples, selecting subsets, and modifying tuple
values with projection. In fact, allowing product makes
\poss intractable when allowing both unordered and totally ordered inputs.

\begin{theorem}\label{thm:posscompextended}
  There is a \Plex query and a \Pgen query for which the \poss problem is NP-complete
  even when the input po-database is restricted to consist only of one totally
  ordered and one unordered po-relation.
\end{theorem}

\begin{proof}
  The proof is by adapting the proof of Theorem~\ref{thm:posscompextend1}.
  The argument is exactly the
  same, except that we take relation $S$ to be \emph{unordered} rather than
  totally ordered. Intuitively, in Figure~\ref{fig:gridpic}, this means that we
  drop the vertical edges in the grid. The proof adapts, because it only used
  the fact that $\id'_j < \id'_k$ for $j < k$ within a row-$i$; we never used
  the comparability relations across rows.
\end{proof}

\section{Tractable Cases for Accumulation Queries}\label{sec:fpt}
We next study \poss and
\cert \emph{in presence of accumulation}. Recall that in the general case, \poss
is NP-hard and \cert is coNP-hard, so we study tractable cases in this section.

\subsection{Cancellative Accumulation}
We first study the case where
accumulation is performed in a \emph{cancellative} monoid (recall
Definition~\ref{def:accrestr}).
This large class of accumulation functions includes the top-$k$
operator (defined above Example~\ref{exa:accumul}) and both operators in
Example~\ref{exa:aggreg}.
We design an efficient algorithm for certainty in this case.

\begin{theorem}\label{thm:certaintyptimec}
    \cert is in PTIME for any fixed \PosRAagg{} query
    that performs accumulation in a cancellative monoid.
\end{theorem}

To prove this result, we define a notion of \emph{possible ranks} for pairs of
incomparable elements, and define a \emph{safe swaps} property, intuitively
designed to ensure that we have only one possible world.

\begin{definition}
    \label{def:pr2}
    Let $P = (\ID, <)$ be a poset.
    For $x \in \ID$, we call
    $A_x
    \colonequals \{y \in \ID \mid y < x\}$ the \emph{ancestors} of~$x$
    and call $D_x \colonequals \{y \in \ID
    \mid x < y\}$ the \emph{descendants} of~$x$.

    Now, given two \emph{incomparable} elements $x$ and $y$ in $\ID$,
    we define
    the \deft{possible ranks} $\pr_P(x, y)$ as the interval $[a+1, \card{\ID}
    - d]$, where
    $a
    \colonequals \card{A_{x} \cup A_{y}}$ and $d \colonequals \card{D_{x}
    \cup D_{y}}$.

    Let $(\calM, \oplus, \epsilon)$ be a monoid and let $h :
    \calD \times \mathbb{N} \to \calM$ be an accumulation map.
    Let $\OR$ be a po-relation with underlying poset~$P$. We say that
    $\OR$ has the \emph{safe swaps} property with respect
    to~$\oplus$ and~$h$ if the following holds: for any pair $x \neq y$ of
    incomparable identifiers of $\OR$, for any pair $p, p+1$ 
    in $\pr_P(x, y)$, we have
    \[
      h(T(x), p) \oplus h(T(y), p+1) = h(T(y), p) \oplus h(T(x),
      p+1).
    \]
\end{definition}

We first show the following soundness result for possible ranks.

\begin{lemma}
  \label{lem:achieverank}
  For any poset $P = (\ID, <)$ and incomparable elements $x, y \in \ID$,
  for any $p \neq q \in \pr_P(x,
  y)$, we can compute in PTIME a linear extension~$<'$ of~$P$ in which element $x$ is
  enumerated at position $p$, and element $y$ is enumerated at position
  $q$.
\end{lemma}

\begin{proof}
  We write $a
    \colonequals \card{A_{x} \cup A_{y}}$ and $d \colonequals \card{D_{x}
    \cup D_{y}}$.
  We will build the desired linear extension $<'$ by enumerating all elements
  of~$A_x \cup A_y$ in any order at the beginning, and enumerating all elements
  of~$D_x \cup D_y$ at
  the end: this can be done without enumerating either $x$
  or $y$ because $x$ and $y$ are incomparable.

  Let $p' \colonequals p - a$, and $q' \colonequals q - a$; it follows from the definition of
  $\pr_P(x, y)$ that $1 \leq p', q' \leq \card{\ID} - d - a$, and clearly $p'
  \neq q'$.

  Now, all elements that are not enumerated by~$<'$ are either $x$, $y$, or
  incomparable to both $x$ and $y$. Consider any linear extension $<''$ of these
  unenumerated elements except
  $x$ and $y$; it has length $\card{\ID} - d - a - 2$.
  Now, as $p' \neq q'$, if $p' < q'$, then we can enumerate $p' - 1$ of these
  elements, enumerate $x$, enumerate $q' - p' - 1$ of these elements,
  enumerate $y$, and enumerate the remaining elements, following~$<''$.
  We proceed similarly, reversing the roles of $x$ and $y$, if $q'
  < p'$. We have constructed $<'$ in PTIME and it clearly has the required
  properties.
\end{proof}

We can then show that the safe swaps criterion is tractable to verify.

\begin{lemma}
  \label{lem:safeptime}
  For any fixed (PTIME-evaluable) accumulation operator $\accum_{h, \oplus}$ we
  can determine in PTIME, given a po-relation $\OR$, whether $\OR$ has safe
  swaps with respect to~$\oplus$ and~$h$.
\end{lemma}

\begin{proof}
  Consider each pair $(\id_1, \id_2)$ of elements of $\OR$ and check in PTIME whether they are incomparable. 
  If this is the case, compute in PTIME
  $\pr_\OR(\id_1, \id_2)$ and for each pair $p$, $p+1$ of consecutive
  integers, compute 
  $h(T(\id_1), p) \oplus h(T(\id_2), p+1)$ and $h(T(\id_2), p) \oplus
  h(T(\id_1), p+1)$ in PTIME (this uses PTIME-evaluability of the accumulation
  operator), and check whether they are equal.
\end{proof}

We last show the following lemma, from which we will easily be able to prove
Theorem~\ref{thm:certaintyptimec}.

\begin{lemma}
  \label{lem:auxcancel}
  For any (PTIME-evaluable) accumulation operator $\accum_{h,\oplus}$ on a
  cancellative monoid $(\calM, \oplus, \epsilon)$, for any po-relation~$\OR$, we
  have $\card{\accum_{h,\oplus}(\OR)} = 1$ iff $\OR$ has safe swaps with respect
  to $\oplus$ and $h$.
\end{lemma}

\begin{proof}
  For the forward direction, assume that $\OR$ does \emph{not} have the safe swaps
  property. Hence, there exist two incomparable identifiers $\id_1$ and $\id_2$ in
  $\OR$ and a pair of consecutive integers $p, p+1$ in $\pr_\OR(\id_1, \id_2)$ such
  that:
  \begin{equation}\label{eq:diseq}h(T(\id_1), p) \oplus h(T(\id_2), p+1) \neq h(T(\id_2), p) \oplus
  h(T(\id_1), p+1)\end{equation}
  We use Lemma~\ref{lem:achieverank}
  to compute
  two possible worlds $L$ and $L'$ of $\OR$, where $\id_1$
  and $\id_2$ occur respectively at positions $p$ and $p+1$ in~$L$, and at
  positions $p+1$ and $p$ respectively in~$L'$: from the proof of
  Lemma~\ref{lem:achieverank} it is clear that we can ensure that $L$ and $L'$
  are otherwise identical.
  As accumulation is associative, we know that $\accum_{h,\oplus}(\OR) = v
  \oplus h(T(\id_1), p) \oplus h(T(\id_2), p+1) \oplus v'$,
  where $v$ is the result of accumulation on the tuples in $L$ before $\id_1$,
  and $v'$ is the result of accumulation on the tuples in $L$ after $\id_2$.
  Likewise, $\accum_{h,\oplus}(\OR) = v \oplus h(T(\id_2), p) \oplus h(T(\id_1),
  p+1) \oplus v'$. We then use cancellativity of~$\calM$ to deduce that these
  two values are different thanks to Equation~\eqref{eq:diseq}. Hence,
  $L$ and $L'$ are possible worlds of~$\OR$ that yield different accumulation
  results, so we conclude that $\card{\accum_{h,\oplus}(\OR)} > 1$.

  \medskip

  For the backward direction, assume that $\OR$ has the safe swaps property.
  Assume by way of contradiction that there are two possible worlds $L_1, L_2
  \in \pw(\OR)$ such that 
  $w_1 \colonequals \accum_{h,\oplus}(L_1)$ and
  $w_2 \colonequals \accum_{h,\oplus}(L_2)$ are different.
  Take $L_1$ and $L_2$ to have the longest possible common prefix,
  i.e., the first position $i$ such that $L_1$ and $L_2$ enumerate a different
  identifier at position~$i$ is as large as possible.
  Let $0 \leq i_0 < \card{\OR}$ be the length of the
  common prefix. Let $\OR'$ be the result of removing from $\OR$ 
  the identifiers enumerated in the common prefix of $L_1$ and $L_2$,
  and let $L_1'$ and $L_2'$ be $L_1$ and $L_2$ without their common prefix.
  Let $\id_1 \neq \id_2$ be the first identifiers
  enumerated by $L_1'$ and
  $L_2'$; it is immediate that $\id_1$ and $\id_2$ are roots of the underlying
  poset of~$\OR'$, that is, no
  element of $\OR'$ is less than them. Further, it is clear
  that the result $w_1'$ of performing accumulation over $L_2'$ (but offsetting all ranks
  by $i_0$), and the result $w_2'$ of 
  performing accumulation over $L_1'$ (also offsetting all ranks by $i_0$),
  are different. Indeed, by the contrapositive
  of cancellativity, combining $w_1'$ and $w_2'$ with the accumulation result of the common
  prefix leads to the different accumulation results $w_1$ and $w_2$.

  Our goal is to
  construct a possible world $L_3' \in \pw(\OR')$ which starts by enumerating
  $\id_1$ but ensures that the result of accumulation on $L_3'$ (again
  offsetting all ranks by~$i_0$) is $w_2'$.
  If we can build such a possible world $L_3'$, then combining it
  with the common prefix will give a possible world $L_3$ of $\OR$ such that the
  result of accumulation on $L_3$ is $w_2 \neq w_1$, yet $L_1$ and $L_3$
  have a common prefix of length $>i_0$, contradicting minimality. Hence, it
  suffices to show how to construct such a possible world $L_3'$.

  As $\id_1$ is an identifier of~$\OR'$, there must be a position where $L_2'$
  enumerates~$\id_1$, and all identifiers before~$\id_1$ in~$L_2'$ cannot be
  descendants of~$\id_1$: as $\id_1$ is a root of~$\OR'$, these identifiers must
  be incomparable to~$\id_1$.
  Write the sequence of these identifiers in~$L_2'$ as
  $L_2'' = \id'_1, \ldots, \id'_m$, and let $L_2'''$ be the
  sequence following $\id_1$, so that $L_2'$ is the concatenation of $L_2''$,
  $\id_1$, and $L_2'''$. We now consider the following sequence of
  list relations, which are clearly possible worlds of $\OR'$, where we
  intuitively move~$\id_1$ to the beginning of the list via successive swaps:
  \begin{flalign*}
    \id'_1 \ldots \id'_{m-2} ~ \id'_{m-1} ~ \id'_m &~ \underline{\id_1} ~ L_2''',&\\
    \id'_1 \ldots \id'_{m-2} ~ \id'_{m-1} &~ \underline{\id_1} ~ \id'_m ~ L_2''',&\\
    \id'_1 \ldots \id'_{m-2} &~\underline{\id_1} ~ \id'_{m-1} ~ \id'_m ~ L_2''',&\\
    &~~\, \vdots &\\
    \id'_1 ~ \id'_{2} &~\underline{\id_1} ~ \id'_{3} \ldots \id'_{m-2} ~\id'_{m-1} ~\id'_m ~
    L_2''',&\\
    \id'_1 &~ \underline{\id_1} ~ \id'_{2} ~\id'_3 \ldots \id'_{m-2} ~\id'_{m-1} ~\id'_m ~
    L_2''',&\\
    &~ \underline{\id_1} ~ \id'_{1} ~\id'_2 ~\id'_3 \ldots \id'_{m-2} ~\id'_{m-1} ~\id'_m ~
    L_2'''.&\\
  \end{flalign*}

  We can see that any consecutive pair in this list achieves the same
  accumulation result. To do so, consider any pair of consecutive
  lists in this sequence, and observe that the two lists only differ at two successive
  identifiers, i.e., the first list contains $\id'_j \id_1$ and the second
  contains $\id_1 \id'_j$ for some $1 \leq j \leq m$. 
  Thus, it suffices to show that the accumulation result for $\id'_j \id_1$ and
  $\id_1 \id'_j$ is the same, and this is exactly what
  the safe swaps property for $\id_1$ and $\id'_j$ says, as it is easily checked that
  $j, j+1 \in \pr_{\OR'}(\id'_j, \id_1)$, so that $j+i_0, j+i_0+1 \in
  \pr_\OR(\id'_j, \id_1)$. 
  Now, the first list relation above is $L_2'$, and the last list relation above
  starts by~$\id_1$, so we have built our desired $L_3'$. This
  establishes the second direction of the proof and concludes.
\end{proof}

We are now ready to prove Theorem~\ref{thm:certaintyptimec}.

\begin{proof}
Given the instance $(D, v)$ of the \cert problem for the query $Q$ with
  accumulation operator $\accum_{h, \oplus}$, we use
  Proposition~\ref{prp:repsys} to build
  $\OR\colonequals Q(D)$ in PTIME. We then use Lemma~\ref{lem:safeptime}  to
test in PTIME whether $\OR$ has safe swaps with
  respect to~$\oplus$ and~$h$. If it does not, then, by
  Lemma~\ref{lem:auxcancel},
$v$ cannot be certain, so $(D, v)$ is not a positive instance of \cert. If
  it does, then, by Lemma~\ref{lem:auxcancel}, $Q(D)$ has only one
  possible world, so we can compute an arbitrary linear extension of~$\OR$,
  obtain one possible world $L \in \pw(\OR)$, check whether
  $\accum_{h,\oplus}(L) = v$, and decide \cert accordingly.
\end{proof}

We have shown Theorem~\ref{thm:certaintyptimec} on \PosRAacc queries. Note that
this result clearly implies that \cert is also tractable for \PosRA queries, as
we claimed in Section~\ref{sec:posscert}: indeed, we can translate any \PosRA
query to a \PosRAacc query that uses a dummy accumulation operator in the
concatenation monoid, and hence the \cert problem for \PosRA
queries reduces to the \cert problem for \PosRAacc queries in this fixed
cancellative monoid. The same
reasoning applied to Theorem~\ref{thm:posscompextend1} implies that the \poss
problem for \PosRAacc is NP-hard even on cancellative monoids, in contrast with
Theorem~\ref{thm:certaintyptimec}.

\subsection{Finite and Position-Invariant Accumulation}
We have shown that \cert (but not \poss) is tractable on cancellative
accumulation operators. It is then natural to wonder whether a similar result
holds when assuming that accumulation is finite and position-invariant (recall
Definition~\ref{def:accrestr}). We will now show that these restrictions do
\emph{not} suffice to make \poss and \cert tractable. However, we will show in
Section~\ref{sec:revisiting} that they can ensure tractability when we combine
them with assumptions on the input po-relations.

We start by showing that \poss is intractable.
  \begin{theorem}\label{thm:possfrihypoposs}
  There is a \PosRAagg{} query with a finite and position-invariant accumulation
    operator
  for which \poss is NP-hard 
        even assuming that the input po-database contains only totally ordered po-relations.
  \end{theorem}

    To prove this result, we define the following finite domains:

\begin{itemize}
  \item $\calD_- \defeq \{\s_-, \n_-, \e_-\}$ (the element names used here
    intuitively correspond to the names used in
     the proof of Theorem~\ref{thm:posscompextend1});
  \item $\calD_+ \defeq \{\s_+, \n_+, \e_+\}$;
  \item $\calD_{\pm} \defeq \calD_- \sqcup \calD_+ \sqcup \{\l, \r\}$ (the additional
    elements stand for ``left'' and ``right'').
\end{itemize}
We define the following regular expression on $\calD_{\pm}^*$, and call
\deft{balanced} a word that satisfies it:
\[
  e \defeq \l \left(\s_- \s_+ | \n_- \n_+ | \e_- \e_+\right)^* \r
\]

We now define the following problem.

\begin{definition}
  The \deft{balanced checking problem} for a \PosRA query $Q$ asks, given a po-database $D$ of
  po-relations over $\calD_{\pm}$, whether there is $L \in \pw(Q(D))$ such that $L$
  is balanced, i.e., it has arity~$1$, its domain is~$\calD_{\pm}$, and $L$
  satisfies $e$ when seen as a 
  word over $\calD_{\pm}$.
\end{definition}

We also introduce the following regular
expression: $e' \defeq \l\, \calD_{\pm}^*\, \r$, which we will use later to guarantee that
there are only two possible worlds. We now show that the balanced checking
problem is intractable.

\begin{lemma}
  \label{lem:balhard}
  There exists a \PosRA query $Q_\b$ over po-databases with domain in~$\calD_{\pm}$
  such that the balanced checking problem for $Q_\b$ is
  NP-hard, even when all input po-relations are totally ordered.
  Further, $Q_\b$ is such that, for any input po-database~$D$, all possible worlds of~$Q_\b(D)$ satisfy $e'$.
\end{lemma}

To prove this lemma, recall the definition of $\cupcat$
(Definition~\ref{def:concat}), and recall from Lemma~\ref{lem:lexconcat} that
$\cupcat$ can be expressed by a \PosRA query.
We construct the query $Q'_\b(R, T) \defeq
\singleton{\l}\cupcat((R \cupgen T)\cupcat \singleton{\r})$,
i.e., 
the union of $R$ and $T$, preceded by $\l$ and followed by~$\r$. 

For any word $w \in \calD_+^*$,
we denote by~$L^+_w$ the unary list
relation defined by mapping each
letter of $w$ to the corresponding letter in $\calD_+$, 
we define $L^-_w$ analogously for~$\calD_-$, and we 
write $\OR^-_w$ for the totally ordered po-relation with $\pw(\OR^-_w) =
\{L^-_w\}$. We now claim that the balanced checking problem for~$Q'_\b$ can be
rephrased in terms of the possibility problem.

\begin{lemma}
  \label{lem:balred}
  For any $w \in \calD_+^*$ and unary po-relation $\OR$ over $\calD_+$, we have
  $L^+_w \in
  \pw(\OR)$ iff the po-database $D$ mapping $R$ to~$\OR^-_w$ and~$T$ to~$\OR$ is a positive instance to the
  balanced checking problem for $Q'_\b$.
\end{lemma}

\begin{proof}
  For the forward direction, assume that $w$ is indeed a possible world~$L$
  of~$\OR$ and let us construct a balanced possible world~$L'$ of $Q'_\b(D)$.
  $L'$ starts with $\l$. Then, $L'$ alternatively enumerates
  one tuple from $\OR^-_w$ (in their total order) and one from $\OR$ (taken in the
  order of the linear extension that yields $L$).
  Finally, $L'$ ends with $\r$. It is clear that $L'$ is balanced.

  For the backward direction, observe that a
  balanced possible world of $Q'_\b(D)$ must start by~$\l$, finish by~$\r$, and, between the two,
  it must alternatively enumerate tuples from~$\OR^-_w$ in their total order and
  tuples from one of the possible worlds
  $L \in \pw(\OR)$: it is clear that $L$ then achieves $w$.
\end{proof}

We now use Lemma~\ref{lem:balred} to prove Lemma~\ref{lem:balhard}.

\begin{proof}
  By Theorem~\ref{thm:posscompextend1}, there is a query~$Q_0$ in
  \PosRA{} such
  that the \poss problem for~$Q_0$ is
  NP-hard, even for totally ordered input relations. What is more, by inspecting the construction in the proof of
  Theorem~\ref{thm:posscompextend1}, we can observe that the output arity
  of~$Q_0$ is~1, and that the input relations can be assumed to have
  domain~$\calD_+$: indeed,
  the input po-relation $S$
  defined as $\ordern{3m-1}$ uses labels that are irrelevant (they are projected
  away), and the input po-relation $S'$ uses only labels from $\{\s, \n, \e\}$,
  so we can rename them to $\{\s_+, \n_+, \e_+\}$.
  We now define the \PosRA query
  $Q_\b$: its input relations are those of~$Q_0$ plus a fresh relation name~$R$,
  and it maps any po-relation $\OR'$ for~$R$ and input po-database~$D$ for~$Q_0$
  to~$Q'_\b(\OR', Q_0(D))$. By definition of $Q'_\b$, our query $Q_\b$ clearly 
  satisfies the additional condition that all possible worlds satisfy
  $e'$.
  
  We reduce the \poss problem for $Q_0$ to the balanced checking problem for
  $Q_\b$ in PTIME. More specifically, we claim that $(D, L)$ is a positive instance to
  \poss for $Q_0$ iff $D'$ is a positive instance to the balanced checking
  problem for~$Q_\b$, where $D'$ is obtained from~$D$ by adding the totally ordered
  relation $\OR^-_w$ to interpret the fresh name~$R$, with $w$ the word
  on~$\calD_+$ achieved
  by~$L$. But this
  is exactly what Lemma~\ref{lem:balred} shows, for $\OR \colonequals Q_0(D)$. This concludes the reduction, so
  we have shown that the balanced checking problem for $Q_\b$ is NP-hard, even
  assuming that the input po-database (here, $D'$) contains only totally ordered
  po-relations.
\end{proof}

To prove our hardness result for \poss (Theorem~\ref{thm:possfrihypoposs}),
we will now reduce the balanced
checking problem to \poss, using an accumulation operator to do
the job. We will further ensure that there are at most two possible
results, which will be useful for \cert later.
To do this, we need to introduce some new concepts.

We define a deterministic complete finite automaton $A$ as follows, where all
omitted transitions go to a sink state $q_\bot$ not shown in the picture.
It is clear that~$A$ recognizes the
language of the regular expression~$e$.

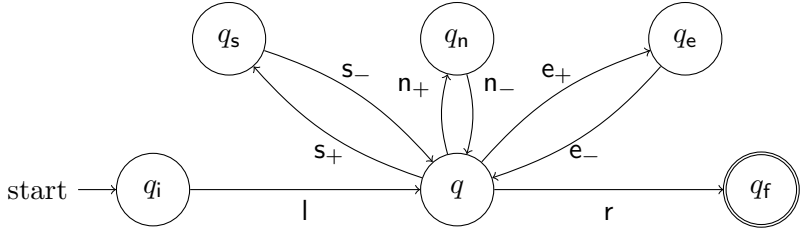
\medskip
\begin{tikzpicture}[->]
  \node[initial,state] (qi) at (-4, 0) {$q_{\ii}$};
  \node[state] (q) at (0, 0) {$q$};
  \node[accepting,state] (qf) at (4, 0) {$q_{\f}$};
  \node[state] (qs) at (-3, 2) {$q_\s$};
  \node[state] (qn) at (0, 2) {$q_\n$};
  \node[state] (qe) at (3, 2) {$q_\e$};
  \path (qi) edge node [below] {$\l$} (q)
        (q)  edge node [below] {$\vphantom{\l}\r$} (qf)
             edge [bend left=15] node [below] {$\s_+$} (qs) 
             edge [bend left=15] node [left,yshift=10] {$\n_+$} (qn) 
             edge [bend left=15] node [above] {$\e_+$} (qe) 
        (qs) edge [bend left=15] node [above] {$\s_-$} (q) 
        (qn) edge [bend left=15] node [right,yshift=10] {$\n_-$} (q) 
        (qe) edge [bend left=15] node [below] {$\e_-$} (q) 
        ;
\end{tikzpicture}
\medskip

We let $S$ be the state space of~$A$, and use it to define the \deft{transition monoid} of~$A$, which is a finite
monoid (so we are indeed performing finite accumulation). 
Let $\calF_S$ be the finite set of
total functions from~$S$ to $S$, and consider the monoid defined on
$\calF_S$ with the identity function~$\mathrm{id}$ as the neutral element, and
with function composition $\circ$ as the (associative) binary operation. We
define inductively a mapping $h$ from $\calD_{\pm}^*$ to
$\calF_S$ as follows, which can be understood as a homomorphism from
the free monoid $\calD_{\pm}^*$ to the transition monoid of~$A$:

\begin{itemize}
  \item For $\epsilon$ the empty word, $h(\epsilon)$ is the identity
    function $\mathrm{id}$.
  \item For $a \in \calD_{\pm}$, $h(a)$ is the transition table for symbol
    $a$ for the automaton $A$, i.e., the function that maps each state $q \in S$
    to the one state $q'$ such that there is an $a$-labeled transition from $q$
    to $q'$; the fact that $A$ is deterministic and complete is what ensures
    that this is well-defined.
  \item For $w \in \calD_{\pm}^*$ and $w \neq \epsilon$, writing $w = a w'$ with
    $a \in \calD_{\pm}$, we define $h(w) \defeq h(w') \circ h(a)$.
\end{itemize}

It is easy to show inductively that, for any $w \in \calD_{\pm}^*$, and
for any $q \in S$, the state $(h(w))(q)$ is the one that we reach in $A$
when reading the word $w$ from the state~$q$. We will identify two special
elements of $\calF_S$:

\begin{itemize}
  \item $f_0$, the function mapping every state of $S$ to the sink state
    $q_{\bot}$;
  \item $f_1$, the function mapping the initial state $q_{\ii}$ to the final
    state $q_{\f}$, and mapping every other state in $S \backslash \{q_{\ii}\}$
    to $q_{\bot}$.
\end{itemize}

Recall the definition of the regular expression $e'$ earlier. We
claim the following property on the automaton $A$.

\begin{lemma}
  \label{lem:transmon}
  For any word $w \in \calD_{\pm}^*$ that matches $e'$, we have $h(w) =
  f_1$ if $w$ is balanced (i.e., satisfies $e$) and $h(w) = f_0$ otherwise.
\end{lemma}

\begin{proof}
  By the definition of~$A$, for any state $q \neq q_{\ii}$, we have $(h(\l))(q) =
  q_{\bot}$,
  so that, as $q_{\bot}$ is a sink state, we have
  $(h(w))(q) = q_{\bot}$ for any $w$ that satisfies $e'$. Further, by
  definition of~$A$, for any state $q$, we have $(h(\r))(q) \in \{q_{\bot},
  q_{\f}\}$, so that, for any state $q$ and $w$ that satisfies $e'$, we have
  $(h(w))(q) \in \{q_{\bot}, q_{\f}\}$. This implies that, for any word $w$
  that satisfies $e'$, we have $h(w) \in \{f_0, f_1\}$.

  Now, as we know that $A$ recognizes the language of $e$, we have the desired
  property, because, for any $w$ satisfying $e'$, $h(w)(q_{\ii})$ is $q_{\f}$
  or not depending on whether $w$ satisfies $e$ or not, so $h(w)$ is $f_1$ or
  $f_0$ depending on whether $w$ satisfies $e$ or not.
\end{proof}

This ensures that we have only two possible accumulation results, and that they
accurately test whether the input word is balanced. We can now prove our
hardness result for \poss, Theorem~\ref{thm:possfrihypoposs}.

\begin{proof}
Consider the query $Q_\b$ whose existence is guaranteed by
Lemma~\ref{lem:balhard}, and remember that  all its possible worlds on any input
  po-database must satisfy~$e'$.
Construct now the query $Q_{\a} \defeq \accum_{h,\circ} (Q_\b)$, using the
  mapping~$h$ that we defined above, seen as a position-invariant accumulation map. We
conclude the proof by showing that \poss
is NP-hard for $Q_{\a}$, even when the input po-database consists only of totally ordered
po-relations. To see that this is the case, we reduce the balanced checking
  problem for $Q_{\b}$ to \poss for $Q_{\a}$ with the trivial reduction: we claim
  that for any po-database $D$, there is a balanced possible world in~$Q_{\b}(D)$ iff $f_1 \in Q_{\a}(D)$,
  which is proved by Lemma~\ref{lem:transmon}. Hence, $Q_{\b}(D)$ is balanced
  iff $(D, f_1)$ is a positive instance of \poss for~$Q_{\a}$. This concludes
  the reduction, and establishes our hardness result.
\end{proof}

We last show an analogue of Theorem~\ref{thm:possfrihypoposs} for \cert as well.

  \begin{theorem}\label{thm:possfrihypocert}
  There is a \PosRAagg{} query with a finite and position-invariant accumulation
    operator
  for which \cert is coNP-hard 
        even assuming that the input po-database contains only totally ordered po-relations.
  \end{theorem}

\begin{proof}
  Consider the query $Q_{\a}$ from Theorem~\ref{thm:possfrihypoposs}.
  We show a
  PTIME reduction from the NP-hard problem of \poss for~$Q_{\a}$ (for totally ordered
  input po-databases) to the negation of
  the \cert problem for $Q_{\a}$ (for input po-databases of the same kind).

  Consider an instance of \poss for $Q_{\a}$ consisting of an input po-database $D$
  and candidate result $v \in \calM$.
  Recall that the
  query $Q_{\a}$ uses accumulation, so it is of the
  form $\accum_{h,\oplus}(Q')$.
  Evaluate $\OR\colonequals Q'(D)$ in
  PTIME by Proposition~\ref{prp:repsys}, and compute in PTIME an
  arbitrary possible world $L' \in \pw(\OR)$ by picking an arbitrary linear
  extension of~$\OR$.
  Let $v'=\accum_{h,\oplus}(L')$. If $v=v'$ then $(D, v)$ is a positive
  instance for \poss for $Q_{\a}$. Otherwise, we have $v\neq v'$. Now, solve
  the \cert problem for $Q_{\a}$ on the input $(D, v')$.
    If the answer is YES, then $(D, v)$ is a negative instance for \poss
    for $Q_{\a}$. Otherwise, there must exist a possible world $L''$
    in $\pw(\OR)$ with $v''=\accum_{h,\oplus}(L'')$ and $v''\neq
    v'$. However, $\card{\pw(Q_{\a}(D))} \leq 2$
    and thus, as $v \neq v'$ and $v'
      \neq v''$,
        we must have $v = v''$. So $(D, v)$ is a
        positive
          instance for \poss for $Q_{\a}$.
  This finishes the reduction and shows that \cert for $Q_{\a}$ is coNP-hard.
\end{proof}

\subsection{Revisiting Section~\ref{sec:fpt2}}
\label{sec:revisiting}

We now know that finiteness and position-invariance do not suffice to ensure the
tractability of \poss and \cert. In this section, we will show that they can
nevertheless be used to obtain tractability when combined with assumptions on
the input po-database, as we did in Section~\ref{sec:fpt2}. Specifically, in the
rest of this section, we will always assume that accumulation is finite,
and we will sometimes assume that it is
position-invariant.
We call \Plexacc and \Pnoprodacc
the
extension of \Plex and \Pnoprod
with accumulation.

We can first generalize our width-based tractability result on \Plex
(Theorem~\ref{thm:aggregw}) to \Plexacc queries with
\emph{finite} accumulation.
\begin{theorem}
  \label{thm:aggregwa}
  For any \Plexacc query with a \emph{finite}
  accumulation operator,
  \poss and \cert are in PTIME on po-databases of bounded width.
\end{theorem}

To show this, as in Section~\ref{sec:fpt2}, we can use
Proposition~\ref{prp:repsys} and Lemma~\ref{lem:lexwidth} to argue that it
suffices to show the following analogue of Lemma~\ref{lem:aggregwinst}. Note
that we compute exactly the (finite) set of all possible accumulation results,
so this allows us to answer both \poss and \cert.

\begin{lemma}\label{lem:aggregwinstacc}
  For any constant $k \in \mathbb{N}$,
  and finite accumulation operator $\accum_{h, \oplus}$, we can compute in PTIME,
  for any input po-relation $\OR$ with width $\leq k$,
  the set $\accum_{h, \oplus}(\OR)$.
\end{lemma}

\begin{proof}
	We extend the proof of Lemma~\ref{lem:aggregwinst} and reuse its
        notation.
        For every sane vector $\mathbf{m}$,
        we now write $t(\mathbf{m}) \defeq \accum_{h,
		\oplus}(T(s(\mathbf{m})))$, where 
                $T(s(\mathbf{m}))$ denotes the sub-po-relation of~$\OR$ with the
                tuples of the order ideal $s(\mathbf{m})$. In other words,
                $t(\mathbf{m})$ is the set of possible accumulation
                results for the sub-po-relation on the order ideal
                $s(\mathbf{m})$: as the accumulation monoid is fixed on finite,
                the set has constant size.
        It is immediate that $t(0, \ldots, 0) = \{\epsilon\}$, i.e., the only
        possible result is the neutral
	element of the accumulation monoid, and that
	$t(n_1, \ldots, n_{k'}) = \accum_{h, \oplus}(\OR)$ is our desired
        answer. Recall that $e_i$ denotes the vector
	consisting of $n-1$ zeros and a $1$ at position $i$, for $1 \leq i
	\leq k'$, and that ``$-$'' denotes the component-wise difference of
        vectors.
	We now observe that, for any sane vector $\mathbf{m}$, we have
        \begin{equation}\label{eq:dynamicdef}t(\mathbf{m}) = \bigcup_{\substack{1 \leq i \leq k'\\m_i > 0}} \left\{
	v \oplus h\left(T(\Lambda_i[m_i]), \sum_{i'} m_{i'}\right)
        \:\middle|\: v \in t(\mathbf{m} - e_i)\right\},\end{equation}
        where we set $t(\mathbf{m}) \colonequals \emptyset$ whenever $\mathbf{m}$ is
        not sane.
        The correctness of Equation~\eqref{eq:dynamicdef} is shown as in the proof of
        Lemma~\ref{lem:aggregwinst}:
	any linear extension of $s(\mathbf{m})$ must end with one of the maximal
	elements of $s(\mathbf{m})$, which must be one of the $\Lambda_i[m_i]$ for $1 \leq
	i \leq m$ such that $m_i > 0$, and the preceding elements must be a linear
	extension of the ideal where this element was removed (which must be an
        ideal, i.e., $\mathbf{m} - e_i$ must be sane). Conversely, any sequence
	constructed in this fashion is indeed a linear extension. Thus, the possible
	accumulation results are computed according to this characterization of the
	linear extensions. We store with each possible accumulation result a
	witnessing totally ordered relation from which it can be computed in PTIME,
	namely, the linear extension prefix considered in the previous reasoning,
	so that we can use the PTIME-evaluability of the underlying monoid to ensure
	that all computations of accumulation results can be performed in PTIME.

        As in the proof of Lemma~\ref{lem:aggregwinst},
        Equation~\eqref{eq:dynamicdef} allows us
        to compute $t(n_1, \ldots, n_{k'})$ in PTIME by a
	dynamic algorithm, which is the set $\accum_{h, \oplus}(\OR)$ that we
        wished to compute. This concludes the proof.
\end{proof}

Second, we can adapt the tractability result for queries without product
(Theorem~\ref{thm:aggregnoprod}) when accumulation is \emph{finite} and
\emph{position-invariant}.
\begin{theorem}
  \label{thm:aggregnoproda}
  For any \Pnoprodacc query with a \emph{finite} and \emph{position-invariant}
  accumulation operator,
  \poss and \cert are in PTIME on po-databases whose relations
  have either bounded width or bounded ia-width.
\end{theorem}

To do so, again, it suffices to show the following analogue of
Lemma~\ref{lem:aggregwuawb} for finite and position-invariant accumulation.

\begin{lemma}
  \label{lem:aggregwuawbagg}
  For any constant $k \in \mathbb{N}$,
  and finite and position-invariant accumulation operator $\accum_{h, \oplus}$,
  we can compute in PTIME,
  for any input po-relation $\OR$ with width $\leq k$
  and input po-relation $\OR'$ with ia-width $\leq k$,
  the set $\accum_{h, \oplus}(\OR \cupgen \OR')$.
\end{lemma}

\begin{proof}
  We use Dilworth's theorem to compute in PTIME a chain partition of
  $\OR$, and we use Proposition~\ref{prp:ptime-ia-partition} to compute in PTIME
  an ia-partition $A_1 \sqcup \cdots \sqcup A_n$ of minimal cardinality of
  $\OR'$, with $n \leq k$.

  We then apply a dynamic algorithm whose state
  consists of the following:
  \begin{itemize}
  \item for each chain in the partition of $\OR$, the position in
    the chain;
  \item for each class $A$ of the ia-partition of $\OR'$, for each element $m$ of the
    monoid, the number of identifiers $\id$ of~$A$ such that $h(T(\id), 1)
      = m$ that have already been used.
  \end{itemize}

  There are polynomially many possible states; for the second bullet point, this
  uses the fact that the monoid is finite, so its size is constant because it is
  fixed as part of the query. Also note that we use the rank-invariance of~$h$
  in the second bullet point.

  The possible accumulation results for each of the possible states can then be
  computed by a dynamic algorithm. At each state, we can decide to make progress
  either in a chain of $\OR$ (ensuring that the element that we enumerate has
  the right image by~$h$, and that the new vector of positions of the chains is
  still sane, i.e., yields an order ideal of~$\OR$) or in a class of $\OR'$
  (ensuring that this class is open, i.e., it has no ancestors in~$\OR'$ that
  were not enumerated yet, and that it contains an element which has the right
  image by~$h$). This algorithm is correct because there is a
  bijection between the ideals of $\OR \cupgen \OR'$ and the pairs of ideals of
  $\OR$ and of ideals of~$\OR'$. Now, the dynamic algorithm considers all ideals
  of~$\OR$ as in the proof of Lemma~\ref{lem:aggregwinstacc}, and it clearly
  considers all possible ideals of~$\OR'$ except that we identify ideals that
  only differ by elements in the same class which are mapped to the same value
  by~$h$ (but this choice does not matter because the class is an antichain and
  these elements are indistinguishable outside the class).

  As in the proof of Lemma~\ref{lem:aggregwinstacc}, we can ensure that all
  accumulation operations are in PTIME, using PTIME-evaluability of the
  accumulation operator, up to the technicality of storing at each state,
  for each of the possible accumulation results, a witnessing totally ordered
  relation from which to compute it in PTIME.
\end{proof}

We note that the finiteness assumption is important, as the previous result
does not hold otherwise. Specifically, there is an accumulation
operator that is \emph{position-invariant} but not \emph{finite},
for which \poss is NP-hard even on unordered
po-relations and with a trivial query.
\begin{theorem}\label{thm:possgri}
  There is a position-invariant accumulation operator $\accum_{h, \oplus}$
  such that \poss is NP-hard for the \Pnoprodacc query $Q \defeq \accum_{h,
  \oplus}(R)$, even on input po-databases where $R$ is interpreted as an unordered relation.
\end{theorem}

\begin{proof}
  We consider the NP-hard partition problem: given a multiset $S$ of
  integers, decide whether it can be partitioned into two sets
  $S_1$ and $S_2$ that have the same sum.
  Let us reduce an instance of the partition problem with
  this restriction to an instance of the \poss problem, in PTIME.

  Let $\calM$ be the monoid generated by
  the functions $f : x \mapsto -x$ and $g_a : x \mapsto x + a$ for $a \in \ZZ$
  under the function composition operation. We have $g_a \circ g_b = g_{a+b}$
  for all $a, b \in \NN$, $f \circ f = \id$, and $f \circ g_a = g_{-a} \circ f$,
  so we actually have $\calD = \{g_a \mid a \in \ZZ\} \sqcup \{f \circ g_a \mid
  a \in \ZZ\}$. Further, $\calM$ is actually a group, as we can define
  $(g_a)^{-1} = g_{-a}$ and $(f \circ g_a)^{-1} = f \circ g_a$ for all $a \in
  \ZZ$.

  We fix $\calD = \NN \sqcup \{-1\}$.
  We define the position-invariant accumulation map $h$ as mapping $-1$ to $f$ and $a
  \in \NN$ to $g_a$.
  We encode the partition problem instance~$S$ in PTIME
  to an unordered po-relation $\OR_S$ with a single attribute, that contains
  one tuple with value $s$ for each $s \in S$, plus one tuple with value~$-1$.
  Consider the \poss
  instance for the query $\accum_{h,+}(\OR)$, on the po-database $D$ where the
  relation name $R$ is interpreted as the po-relation $\OR_S$, and for the
  candidate result $v \colonequals f \in \calM$.

  We claim that this \poss instance is positive iff the partition problem has a
  solution. Indeed, if $S$ has a partition, 
  let $s = \sum_{i \in S_1} i = \sum_{i \in S_2} i$.
  Consider the total order on
  $\OR_S$ which enumerates the tuples corresponding to the elements of~$S_1$, then
  the tuple $-1$, then the tuples corresponding to the elements of~$S_2$.
  The result of accumulation is then $g_s \circ f \circ g_s$, which is $f$.

  Conversely, assume that the \poss problem has a solution. Consider a witness
  total order of~$\OR_S$; it must a (possibly empty) sequence of tuples
  corresponding to a subset $S_1$ of~$S$, then the tuple $-1$, then a (possibly
  empty) sequence corresponding to $S_2 \subseteq S$.
  Let $s_1$ and $s_2$ respectively be the sums
  of these subsets of~$S$. The result of accumulation is then $g_{s_1} \circ f
  \circ g_{s_2}$, which simplifies to $g_{s_1 -
  s_2} \circ f$. Hence, we have $s_1 = s_2$, so that $S_1$ and
  $S_2$ are a partition witnessing that $S$ is a positive instance of the
  partition problem.

  As the reduction is in PTIME, this concludes the proof.
\end{proof}

Finally, as explained above Example~\ref{exa:accumul}, we can use accumulation 
capture \emph{position-based selection}
($\text{top-}k$, $\text{select-at-}k$) and \emph{tuple-level comparison}
(whether the first occurrence of a tuple precedes all occurrences of another
tuple) for \PosRA queries. Using a direct construction for these problems, we
can show that they are tractable.

\begin{proposition}\label{prop:otherdef}
        For any \PosRA query $Q$, the following problems are in PTIME.
        \begin{itemize}
          \item \textbf{select-at-$\bm{k}$:} Given a po-database $D$, tuple value $t$, and
            position $k \in \NN$, determine whether it is \mbox{possible/certain} that $Q(D)$ has
            value $t$ at position~$k$;
          \item \textbf{top-$\bm{k}$:} For any \emph{fixed} $k \in \NN$, given a po-database
            $D$ and list relation $L$ of length~$k$, determine whether it is
            possible/certain that the top-$k$ values in $Q(D)$ are exactly~$L$;
          \item \textbf{tuple-level comparison:} Given a po-database $D$ and two tuple
            values $t_1$ and $t_2$,
            determine whether it is possible/certain that the first occurrence of~$t_1$
            precedes all occurrences of~$t_2$.
        \end{itemize}
\end{proposition}

\begin{proof}
  To solve each problem, we first compute the po-relation $\OR \defeq Q(D)$ in
  PTIME by Proposition~\ref{prp:repsys}.
  We then address each problem in turn.

  First, we show tractability for \textbf{select-at-$\bm{k}$}.
    Considering the po-relation $\OR = (\ID, T, <)$, we can compute in PTIME,
    for every element $\id \in \ID$, its \emph{earliest
		index} $\ii^-(\id)$, which is the number of ancestors of~$\id$ by $<$ plus one, and its \emph{latest
		index} $\ii^+(\id)$, which is the number of elements of $\OR$ minus the number of
	descendants of~$\id$. It is easily seen that for any element $\id \in \ID$, there is a
        linear extension of~$\OR$ where $\id$ appears at position $\ii^-(\id)$ (by
        enumerating first exactly the ancestors of~$\id$), or at position
        $\ii^+(\id)$ (by enumerating first everything except the descendants
        of~$\id$), or in fact at any position of $[\ii^-(\id), \ii^+(\id)]$, the
        \emph{interval} of~$\id$ (this is by enumerating first the ancestors
        of~$\id$, and then as many elements as needed that are incomparable
        to~$\id$, along a linear extension of these elements).
	Hence, select-at-$k$ possibility for tuple $t$ and position $k$ can be
        decided by checking, for each $\id \in \ID$ such that $T(\id) = t$,
        whether $k \in [\ii^-(\id), \ii^+(\id)]$, and answering YES iff we can find
        such an~$\id$. For select-at-$k$ certainty, we answer NO iff we can find
        an $\id \in \ID$ such that $k \in [\ii^-(\id), \ii^+(\id)]$ but we have
        $T(\id) \neq t$.

  Second, we show tractability for \textbf{top-$\bm{k}$}.
    Considering the po-relation $\OR = (\ID, T, <)$,
  we consider each sequence of~$k$ elements of~$\OR$, of which there are at most
      $\card{\ID}^k$, i.e., polynomially many, as $k$ is fixed. To solve possibility for top-$k$, we
  consider each such sequence $\id_1, \ldots, \id_k$ such that $(T(\id_1),
  \ldots, T(\id_k))$ is equal to the candidate list relation $L$, and we check if this sequence is indeed a prefix of a linear extension
  of~$\OR$, i.e., whether, for each $i \in \{1, \ldots, k\}$, for any $\id \in
  \ID$ such that $\id < \id_i$, if $\id_i \in \{\id_1, \ldots, \id_{i-1}\}$,
  which we can do in PTIME. We answer YES iff we can find such a sequence.

  For certainty, we consider each sequence $\id_1, \ldots, \id_k$ such that we
      have $(T(\id_1), \ldots, T(\id_k)) \neq L$, and we check whether it is a prefix of
  a linear extension in the same way: we answer NO iff we can find such a
  sequence.

  Third, we show tractability for \textbf{tuple-level comparison}.
  We are given the two tuple values $t_1$ and $t_2$, and we assume that both are
  in the image of~$T$, as the tuple-level comparison problem is vacuous
  otherwise.

  For possibility, given the two tuple values $t_1$ and $t_2$, we consider each
  $\id \in \ID$ such that $T(\id) = t_1$, and for each of them, we construct
  $\OR_\id \defeq (\ID, T, {<_\id})$ where ${<_\id}$ is the transitive closure
  of ${<} \cup \{(\id, \id') \mid \id' \in \ID, T(\id') = t_2\}$. We answer YES iff one of the
  $\OR_\id$ is indeed a po-relation, i.e., if $<_\id$ as defined does not
  contain a cycle. This is correct, because it is possible that the first
  occurrence of~$t_1$ precedes all occurrences of~$t_2$ iff there is some
  identifier $\id$ with tuple value $t_1$ that precedes all identifiers with
  tuple value $t_2$, i.e., iff one of the $\OR_\id$ has a linear extension.

  For certainty, given $t_1$ and $t_2$, we answer the negation of possibility
  for $t_2$ and $t_1$. This is correct because certainty is false iff there is a
  linear extension of~$\OR$ where the first occurrence of~$t_1$ does not precede
  all occurrences of~$t_2$, i.e., iff there is a linear extension where the
  first occurrence of $t_2$ is not after an occurrence of~$t_1$, i.e., iff some
  linear extension is such that the first occurrence of~$t_2$ precedes all
  occurrences of~$t_1$, i.e., iff possibility is true for $t_2$
      and~$t_1$.
\end{proof}

\section{Extensions}\label{sec:extensions}
We consider two extensions to our model: group-by and
duplicate elimination. 

\subsection{Group-By}

First, we extend accumulation with a \emph{group-by}
operator, inspired by SQL.
\begin{definition}
    Let $(\calM, \oplus, \epsilon)$ be a monoid and $h :
    \calD^k\times\mathbb N_{>0} \to \calM$ be
    an accumulation map, and
    let 
    $\mathbf{A} = A_1,...,A_n$ be a
    sequence of attributes: we call $\accumgby_{h, \oplus, \mathbf{A}}$ an \emph{accumulation
        operator with group-by}.
    Letting $L$~be a list relation with compatible schema,
    we define $\accumgby_{h, \oplus, \mathbf{A}}(L)$ as
    an \emph{unordered} relation
    that has, for each tuple value $t
    \in \Pi_{\mathbf{A}}(L)$, one tuple $\langle t, v_t \rangle$, where $v_t$ is $\accum_{h,
        \oplus}(\sigma_{A_1 = t.A_1\land\dots\land A_n=t.A_n}(L))$ with $\Pi$ and $\sigma$ on
    the list relation~$L$ having the expected semantics. The result on a
    po-relation $\OR$ is the set of unordered relations $\{\accumgby_{h, \oplus, \mathbf{A}}(L) \mid L \in
    \pw(\OR)\}$.
\end{definition}

In other words, the operator ``groups by'' the values of $A_1,...,A_n$, and
performs accumulation within each group, forgetting the order
across groups. As for standard accumulation, we only allow 
group-by as an outermost operation,
calling \PosRAaccgby the language of \PosRA queries followed by one accumulation
operator with group-by. Note that the set of possible results is generally not 
a po-relation,
    because the underlying bag relation is not certain.

We next study the complexity of \poss and \cert for \PosRAaccgby queries. Of course,
whenever \poss and \cert are hard for some \PosRAacc query $Q$ on some kind of
input po-relations, then there is a corresponding \PosRAaccgby query
for which hardness also holds (with empty $\mathbf{A}$). The main point of this section is to show that
the converse is not true: the addition of group-by increases
complexity. Specifically, we show that the \poss problem for \PosRAaccgby is
hard even on totally ordered po-relations and without the $\times_\dir$
operator.
This result contrasts with the tractability of \poss for \Plex queries
(Theorem~\ref{thm:aggregw})
 and for \Plexacc queries with finite
accumulation (Theorem~\ref{thm:aggregwa})
on totally ordered po-relations.

\begin{theorem}\label{thm:hardpossgby}
    There is a \PosRAaccgby query $Q$ 
    with finite and position-invariant
    accumulation, not using $\times_\dir$, such that \poss
    for $Q$ is NP-hard even on totally ordered po-relations.
\end{theorem}

\begin{proof}
  Let $Q$ be the query $\accumgby_{\oplus, h, \{1\}}(Q')$, where we define
  \[
    Q' \defeq \Pi_{3,4}(\sigma_{.1=.2}(R \times_\lex
    (S_1 \cup S_2 \cup S_3))).
  \]
  In the accumulation operator, the accumulation map $h$ maps each tuple $t$ to its
  second component. Further, we define the finite monoid $\calM$ to be the
  \emph{syntactic monoid} \cite{pin1997syntactic} of the language defined by the regular expression $\s (\l_+\l_- | \l_-
  \l_+)^* \e$, where $\s$ (for ``start''), $\l_-$ and $\l_+$, and $\e$ (for
  ``end'') are fresh values from $\calD$: this monoid 
  ensures that, 
  for any non-empty word $w$ over the alphabet $\{\s, \l_-, \l_+, \e\}$ that
  starts with $\s$ and ends with $\e$, the word $w$
  evaluates to $\epsilon$ in $\calM$ iff $w$ matches this regular expression.

  We reduce from the NP-hard 3-SAT problem: we are given a conjunction of
  clauses $C_1, \ldots, C_n$, with each clause being a disjunction of three
  literals, namely, a variable or negated variable among $x_1, \ldots, x_m$, and
  we ask whether there is a valuation of the variables such that the clause is
  true. We fix an instance of this problem. We assume without loss of generality
  that the instance has been preprocessed to ensure that no clause contained two
  occurrences of the same variable, i.e., we remove duplicate literals in
  clauses, and we remove any clause that
  contains two occurrences of the same variable with different polarities (as
  the clause is then vacuous). We further assume that the instance has been
  preprocessed to ensure that each clause contains exactly 3 variables: we do so
  by introducing three fresh variables $d_1$, $d_2$, and $d_3$,
  by adding all possible clauses $\pm d_1 \lor \pm d_2 \lor \pm d_3$ on these
  variables except $\neg d_1 \lor \neg d_2 \lor \neg d_3$ (i.e., seven clauses), and by padding
  the other clauses to three literals by adding distinct disjuncts chosen from
  the $\neg d_i$. It is clear that
  this does not change the semantics of the instance: any satisfying assignment
  of the original instance yields a satisfying assignment of the rewritten instance by setting $d_1$,
  $d_2$, and $d_3$ to true, and conversely any satisfying assignment to the
  rewritten instance must set $d_1$, $d_2$, and $d_3$ to
  true (any other assignment will violate the clause
  where each $d_i$ has the polarity which is the opposite of its value in
  the assignment), so the padding literals are never used to make a clause true.

  We define the relation $R$ to be $\ordern{m+3}$. The totally ordered relations
  $S_1$, $S_2$, and $S_3$ consist of $3m+2n$ tuple values defined as follows.

  \begin{itemize}
    \item First, for the tuples with positions from $1$ to $m$ (the ``opening gadget''):
      \begin{itemize}
        \item The first component is $1$ for all tuples in $S_1$ and $0$ for
          all tuples in $S_2$ and $S_3$ (so they do not join with $R$);
        \item The second component is $i$ for the $i$-th tuple in $S_1$ (and
          irrelevant for tuples in $S_2$ and $S_3$);
        \item The third component is $\s$ for all these tuples.
      \end{itemize}
      The intuition for the opening gadget is that it ensures that accumulation in
      each of the $m$ groups will start with the start value $\s$, used to
      disambiguate the possible monoid values and ensure that there is exactly
      one correct value.
    \item For the tuples with positions from $m+1$ to $2m$ (the ``variable
      choice'' gadget):
      \begin{itemize}
        \item The first component is $2$ for all tuples in~$S_1$ and~$S_2$ and
          $0$ for all tuples in~$S_3$ (so they do not join with $R$);
        \item The second component is $i$ for the $(m+i)$-th tuple in $S_1$ and in
          $S_2$ (and irrelevant for~$S_3$);
        \item The third component is $\l_-$ for all tuples in~$S_1$ and $\l_+$
          for all tuples in~$S_2$ (and irrelevant for~$S_3$).
      \end{itemize}
      The intuition for the variable choice gadget is that, for each group, we
      have two incomparable elements, one labeled $\l_-$ and one labeled $\l_+$.
      Hence, any linear extension must choose to enumerate one after the other,
      committing to a valuation of the variables in the 3-SAT instance;
      to achieve the candidate possible world, the linear extension will then
      have to continue enumerating the elements of this group in the correct
      order.
    \item For the tuples with positions from $2m+1$ to $2m + 2n$ (the ``clause
      check'' gadget), for each $1 \leq j \leq n$, letting $j' \defeq 2n + j +
      1$, we describe tuples $j'$ and $j'+1$ in $S_1$, $S_2$, $S_3$:
      \begin{itemize}
        \item The first component is $j+2$;
        \item The second component carries values in
          $\{a, b, c\}$, where we write clause $C_j$ as $\pm x_a \vee \pm x_b
          \vee \pm x_c$. Specifically, the tuple $j'+1$ in relations $S_1$,
          $S_2$, and $S_3$ have values $a$, $b$, and $c$ respectively; and the tuple $j'$ in
          relations $S_1$, $S_2$, and $S_3$ have values $c$, $a$, $b$
          respectively.
        \item The third
          component carries values in $\{\l_-, \l_+\}$. In relation $S_1$, we give value
          $\l_+$ to tuple $j'+1$ and value $\l_-$ to tuple~$j'$ if the first
          variable of~$C_j$ is positive, and we do the reverse if it is
          negative. We do the same in relations $S_2$ and $S_3$ depending on the
          polarity of the second and third variables of~$C_j$, respectively.
      \end{itemize}
      The intuition for the clause check gadget is that, for each $1 \leq j \leq
      n$, the tuples at levels $j'$ and $j'+1$ check that clause $C_j$ is
      satisfied by the valuation chosen in the variable choice gadget.
      Specifically, if we consider the order constraints on the two elements
      from the same group (i.e., second component) which are implied by the
      order chosen for this variable in the variable choice gadget, the
      construction ensures that these order constraints plus the comparability
      relations of the chains imply a cycle (that is, an impossibility) iff the
      clause is violated by the chosen valuation.
    \item For the tuples with positions from $2m + 2n + 1$ to $3m + 2n$ (the
      ``closing gadget''), the definition is like the opening gadget but
      replacing $\e$ by $\s$, namely:
      \begin{itemize}
        \item The first component is $n+3$ for all tuples in $S_1$ and $0$ for
          all tuples in $S_2$ and $S_3$ (which again do not join with $R$);
        \item The second component is $i$ for the $i$-th tuple in $S_1$;
        \item The third component is $\e$ for all these tuples.
      \end{itemize}
      The intuition for the closing gadget is that it ensures that accumulation
      in each group ends with value $\e$.
  \end{itemize}

  We define the candidate possible world to consist of a list relation of $n$
  tuples; the $i$-th tuple carries value $i$ as its first component (group
  identifier) and the
  acceptation value from the monoid $\calM$ as its second component (accumulation value).
  The reduction that we described is
  clearly in PTIME, so all that remains is to show correctness of the reduction.
  
  To do so, we first describe the result of evaluating $\OR \defeq Q'(R, S_1, S_2, S_3)$ on
  the relations described above. Intuitively, it is just like $\Pi_{2, 3}
  (\sigma_{.2 \neq
  \text{``0''}}(S_1 \cup S_2 \cup S_3))))$, but with the following additional
  comparability relations: all tuples in all chains whose first component
  carried a value $i$ are less than all tuples in all chains whose first
  component carried a value $j > i$. In other words, we add comparability
  relations across chains as we move from one ``first component'' value to the
  next. The point of this is that it forces us to enumerate the tuples of the
  chains in a way that ``synchronizes'' across all chains whenever we change the
  first component value. Observe that, in keeping with
  Lemma~\ref{lem:lexwidth}, the width of $\OR$ has a constant bound,
  namely,~$3$.

  Let us now show the correctness of the reduction. For the forward direction,
  consider a valuation $\nu$ that satisfies the 3-SAT instance. Construct the
  linear extension of~$\OR$ as follows.

  \begin{itemize}
    \item For the opening gadget, enumerate all tuples of $S_1$ in the prescribed
      order. Hence, the current accumulation result in all $m$ groups is $\s$.
    \item For the variable choice gadget, for all $i$, enumerate the $i$-th
      tuples of $S_1$ and $S_2$ of the gadget in an order depending on
      $\nu(x_i)$: if $\nu(x_i)$ is $1$, enumerate first the tuple of~$S_1$ and
      then the tuple of~$S_2$, and do the converse if $\nu(x_i) = 0$. Hence, for
      all $1 \leq i \leq m$, the
      current accumulation result in group~$i$ is $\s \l_- \l_+$ if $\nu(x_i)$
      is $1$ and $\s \l_+ \l_-$ otherwise.
    \item For the clause check gadget, we consider each clause in order, for $1
      \leq j \leq n$, maintaining the property that, for each group $1 \leq i
      \leq n$, the current accumulation result in group $i$ is of the form $\s
      (\l_- \l_+)^*$ if $\nu(x_i) = 1$ and $\s (\l_+ \l_-)^*$ otherwise.
      
      Fix
      a clause $C_j$, let $j' \defeq 2n + j + 1$ as before, and study the tuples
      $j'$ and $j'+1$ of~$S_1, S_2, S_3$. As $C_j$ is satisfied under $\nu$, let
      $x_d$ be the witnessing literal (with $d \in \{a, b, c\}$), and let $d'$
      be the index (in $\{1, 2, 3\}$) of variable~$d$. Assume that $x_d$
      occurs positively; the argument is symmetric if it occurs negatively. By
      definition, $\nu(x_d) = 1$, and by construction tuple $j'$ in relation
      $S_{1 + (d' + 1 \text{ mod } 3)}$ carries value $\l_-$ and it is in
      group~$d$. Hence, we can enumerate it and group $d$ now carries a value of
      the form $\s (\l_- \l_+)^* \l_-$. Now, letting $x_e$ be the $1 + (d' + 1
      \text{ mod } 3)$-th variable of $\{x_a, x_b, x_c\}$, the two elements of
      group $e$ (tuple $j'+1$ of
      $S_{1 + (d' + 1 \text{ mod } 3)}$ and tuple $j'$ of $S_{1 + (d' + 1 \text{
        mod } 3)}$) both had all their predecessors enumerated; so we can
      enumerate them in the order that we prefer to satisfy the condition on the
      accumulation values; then we enumerate likewise the two elements in the
      remaining group in the order that we prefer, and last we enumerate the
      second element of group~$d$; so we have satisfied the invariants.
    \item Last, for the closing gadget, we enumerate all tuples of~$S_1$ and we have
      indeed obtained the desired accumulation result.
  \end{itemize}

  This concludes the proof of the forward direction.

  \medskip

  For the backward direction, consider any linear extension of~$\OR$. Thanks to
  the order constraints of $\OR$, the linear extension must enumerate tuples in
  the following order.

  \begin{itemize}
    \item First, all tuples of the opening gadget.
    \item Then, all tuples of the variable choice gadget. We use this to define
      a valuation $\nu$: for each variable $x_i$, we set $\nu(x_i) = 1$ if the
      tuple of $S_1$ in group $i$ was enumerated before the one in group~$S_2$,
      and we set $\nu(x_i) = 0$ otherwise.
    \item Then, for each $1 \leq j \leq n$, in order, tuples $2n + j + 1$ of
      $S_1$, $S_2$, $S_3$.

      Observe that, for each value of~$j$, just before we enumerate these tuples,
      it must be the case that the current accumulation value for every variable $x_i$ is
      of the form $\s (\l_- \l_+)^*$ if $\nu(x_i) = 1$, and $\s (\l_+ \l_-)^*$
      otherwise. Indeed, fixing $1 \leq i \leq n$, assume the case where $\nu(x_i) = 1$
      (the case where $\nu(x_i) = 0$ is symmetric). In this case, the accumulation
      state for $x_i$ after the variable choice gadget was $\s \l_- \l_+$, and
      each pair of levels in the clause check gadget made us enumerate either
      $\epsilon$ (variable $x_i$ did not occur in the clause) or one of
      $\l_-\l_+$ or $\l_+\l_-$ (variable $x_i$ occurred in the clause); as the
      3-SAT instance was preprocessed to ensure that each variable occurred only
      at most once in each clause, this case enumeration is exhaustive. Hence,
      the only way to obtain the correct accumulation result is to always
      enumerate $\l_- \l_+$, as if we ever do the contrary the accumulation
      result can never satisfy the regular expression that it should satisfy.

    \item Last, all tuples of the closing gadget.
  \end{itemize}

  What we have to show is that the valuation $\nu$ thus defined indeed satisfies
  the formula of the 3-SAT instance. Indeed, fix $1 \leq j \leq n$ and consider
  clause $C_j$. Let $S_i$ be the first relation where the linear extension
  enumerated a tuple for the clause check gadget of~$C_j$, and let $x_d$ be its
  variable (where $d$ is its group index). If $\nu(x_d) = 1$, then the
  observation above implies that the label of the enumerated element must be
  $\l_-$, as otherwise the accumulation result cannot be correct. Hence, by
  construction, it means that variable $x_d$ must occur positively in $C_j$, so
  $x_d$ witnesses that $\nu$ satisfies $C_j$. If $\nu(x_d) = 0$, the reasoning is
  symmetric. This concludes the proof in the backwards direction, so we have
  established correctness of the reduction, which concludes the proof.
\end{proof}

By contrast, it is not hard to see that the \cert problem for \PosRAaccgby
reduces to \cert for the same query without group-by, so it is no harder than
the latter problem, and all \cert tractability results from Section \ref{sec:fpt} extend.

\begin{theorem}\label{thm:easycertgby}
    Theorems~\ref{thm:certaintyptimec}, \ref{thm:aggregwa}, and \ref{thm:aggregnoproda} 
    extend to the \PosRAaccgby problem
    when imposing the same restrictions on query operators, accumulation, and input
    po-relations. Specifically:
  \begin{itemize}
    \item     \cert is in PTIME for any fixed \PosRAaccgby{} query
          that performs accumulation in a cancellative monoid.

        \item
          For any \PosRAaccgby{} query not using the $\times_\dir$ operator and with a \emph{finite}
  accumulation operator,
  \poss and \cert are in PTIME on po-databases of bounded width.

        \item
          For any \PosRAaccgby{} query not using any product operator and with a \emph{finite} and \emph{position-invariant}
  accumulation operator,
  \poss and \cert are in PTIME on po-databases whose relations
  have either bounded width or bounded ia-width.
  \end{itemize}

\end{theorem}

To prove this, we show the following auxiliary result.

  \begin{lemma}
    \label{lem:certreduction}
    For any \PosRAaccgby query $Q \defeq \accumgby_{h, \oplus, P} (Q')$ and family
    $\calD$ of po-databases, the \cert problem for $Q$ on input po-databases
    from $\calD$
    reduces in PTIME to the \cert problem for $\accum_{h, \oplus} (R)$ (where
    $R$ is a relation name), on the family $\calD'$ of po-databases mapping
    the name $R$ to a subset of a po-relation of $\{Q'(D) \mid D \in \calD\}$.
  \end{lemma}

\begin{proof}
  To prove that, consider an instance of \cert for $Q$, defined by an input
  po-database $D$ of $\calD$
  and candidate possible world $L$. We first evaluate $\OR' \defeq Q'(D)$ in
  PTIME. Now, for each tuple value $t$ in $\Pi_P(\OR')$, let $\OR_t$ be the
  restriction of~$\OR'$ to the elements matching this value; note that the
  po-database mapping $R$ to $\OR_t$ is indeed in the family $\calD'$. We solve \cert for
  $\accum_{h, \oplus}(R)$ on each $R\mapsto\OR_t$ in PTIME with the candidate possible world
  obtained from $L$ by extracting the accumulation value for that group, and
  answer YES to the original \cert instance iff all these invocations answer
  YES. As this process is clearly in PTIME, it just remains to show correctness of the
  reduction.

  For the forward direction, assume that each of the invocations answers YES,
  but the initial instance to \cert was negative. Consider two linear extensions
  of $\OR'$ that achieve different accumulation results and witness that the
  initial instance was negative, and consider a group $t$ where these
  accumulation results for these two linear extensions differ. Considering the
  restriction of these linear extensions to that group, we obtain the two
  different accumulation values for that group, so that the \cert invocation for
  $\OR_t$ should not have answered YES.

  For the backward direction, assume that the invocation for tuple~$t$ does not answer YES, then
  considering two witnessing linear extensions for that invocation, and
  extending them two linear extensions of~$\OR'$ by enumerating other tuples in
  an indifferent way, we obtain two different accumulation results for~$Q$ which
  differ in their result for~$t$. This concludes the proof.
\end{proof}

This allows us to show Theorem~\ref{thm:easycertgby}.
  
  \begin{proof}
    We consider all tractability results
of Section~\ref{sec:fpt} in turn, and show that they extend to \PosRAaccgby
queries, under the same restrictions on operators, accumulation, and input
po-relations.

First, we consider the tractability of \cert for accumulation in a cancellative
    monoid (Theorem~\ref{thm:certaintyptimec}). As this result holds for any
    input po-database, tractability 
    for \PosRAaccgby follows directly from Lemma~\ref{lem:certreduction}.

Second, we consider the tractability of \cert for \Plexacc queries with a finite
accumulation operator on po-databases of bounded width
    (Theorem~\ref{thm:aggregwa}).
    The result extends because, for any
family $\calD$ of po-databases whose po-relations have width at most $k$ for some
$k \in \Nat$, we know by 
Lemma~\ref{lem:lexwidth} that the result $Q'(D)$ for $D \in \calD$ also
has width depending only on~$Q'$ and on~$k$, and we know that restricting to a
subset of $Q'(D)$ (namely, each group) does not increase the width (this is like
the case of selection in the proof of Lemma~\ref{lem:lexwidth}). Hence, the
    family $\calD'$ also has bounded width, and we can concludes using
    Lemma~\ref{lem:certreduction}.

Third, we consider 
the tractability of \cert for \Pnoprodacc queries with a finite and
    position-invariant accumulation operator on po-databases whose relations
    have either bounded width or bounded ia-width
    (Theorem~\ref{thm:aggregnoproda}).
    The result extends because, by
Lemma~\ref{lem:rewritenoprod} and subsequent observations, the result
$Q'(D)$ for $D \in \calD$ is a union of a po-relation of bounded width and of a
po-relation with bounded ia-width. Restricting to a subset (i.e., a group), this
    property is preserved (as in the case of selection in the proof of
    Lemma~\ref{lem:lexwidth} and of Lemma~\ref{lem:lexiawidthnoprod}),
    which allows us to
    conclude using Lemma~\ref{lem:certreduction}.
\end{proof}

\subsection{Duplicate Elimination}

We last study the problem of consolidating tuples with
\emph{duplicate values}. To this end, we define a new operator, $\dupelim$, and introduce a
semantics for it. The main problem is that tuples with the same
values may be ordered differently relative to other
tuples. To mitigate this, we introduce the notion
of \emph{id-sets}.

\begin{definition}
    \label{def:idset}
    Given a totally ordered po-relation $(\ID, T, <)$, a subset $\ID'$ of~$\ID$
    is an \emph{indistinguishable duplicate set} (or
    \deft{id-set}) if for every $\id_1, \id_2 \in \ID'$,
    we have $T(\id_1) = T(\id_2)$,
    and, for every $\id \in \ID \backslash \ID'$,
    we have $\id < \id_1$ iff $\id < \id_2$, and $\id_1 < \id$ iff $\id_2 < \id$.
\end{definition}

\begin{example}
    \label{exa:dup1}
    Consider the totally ordered relation
    $\OR_1 \defeq \Pi_{\mathit{hotelname}}(\mathit{Hotel})$, with
    $\mathit{Hotel}$ as in Figure~\ref{fig:examplerels}. The two
    ``Mercure''
    tuples are not an id-set: they disagree on their
    ordering with ``Balzac''.
    Consider now the totally ordered relation
    $\OR_2 \defeq \Pi_{\mathit{hotelname}}(\mathit{Hotel}_2)$: the two
    ``Mercure'' tuples are an id-set.
    Note that a singleton is always an id-set.
\end{example}
We define a semantics for $\dupelim$ on a totally ordered po-relation $\OR =
(\ID, T, <)$
via id-sets.
First, check that for every
tuple value $t$ in the image of~$T$, the set $\{\id \in \ID \mid T(\id) = t\}$ is an id-set
in~$\OR$. If this
holds, then we call $\OR$ \emph{safe}, and set $\dupelim(\OR)$ to be the singleton
$\{L\}$ of the only possible world of
the restriction of~$\OR$ obtained by picking one representative
element per id-set (clearly $L$ does not depend on the chosen
representatives).
Otherwise, we call $\OR$ \emph{unsafe} and say that
duplicate consolidation has \emph{failed}; we then set $\dupelim(\OR)$ to
be an empty set of possible worlds. Intuitively, duplicate
consolidation tries to reconcile (or ``synchronize'') order
constraints for tuples with the same values, and fails when it
cannot be done.

\begin{example}
    In Example~\ref{exa:dup1}, we have $\dupelim(\OR_1)=\emptyset$
    but $\dupelim(\OR_2) = (\textup{Balzac}, \textup{Mercure})$.
\end{example}
We then extend 
$\dupelim$ to po-relations by
considering all possible
results of duplicate elimination on the possible worlds,
ignoring the unsafe possible worlds. If no possible worlds are safe, then
we \emph{completely fail}.

\begin{definition}
    For any list relation $L$, we let $\OR_L$ be a po-relation such that
    $\pw(\OR_L) = \{L\}$.
    For $\OR$ a po-relation, let $\dupelim(\OR) \defeq
    \bigcup_{L \in pw(\OR)}\dupelim(\OR_L)$. We say that $\dupelim(\OR)$
    \deft{completely fails} if we have $\dupelim(\OR) = \emptyset$, i.e.,
    $\dupelim(\OR_L) = \emptyset$ for every $L\in
    pw(\OR)$.
\end{definition}

\begin{example}
    Consider the totally ordered po-relation
    $\mathit{Restaurant}$ from
    Figure~\ref{fig:examplerels}, and a totally ordered po-relation
    $\mathit{Restaurant}_2$ whose only possible
    world is $(\langlem
    \textup{Tsukizi} \ranglem,$ $\langlem \textup{Gagnaire} \ranglem)$.
    Let 
    $Q \defeq
    \dupelim(\Pi_{\mathit{restname}}(\mathit{Restaurant}) \cupgen
    \mathit{Restaurant}_2)$.
    Intuitively, $Q$ combines restaurant rankings,
    using duplicate consolidation to collapse two occurrences of the
    same 
    name to a single tuple.
    The only possible world of $Q$ is
    (\textup{Tsukizi}, \textup{Gagnaire}, \textup{TourArgent}), since
    duplicate elimination fails in the other possible worlds: 
    indeed, this is the only possible way to combine the rankings.
\end{example}
We next show that the result of $\dupelim$ can still be represented as a
po-relation,
up to complete failure (which may be efficiently identified).

 We first define the notion of \emph{quotient} of a
    po-relation by \emph{value equality}.

    \begin{definition}
        For a po-relation $\OR=(\ID,T,{<})$, we define the \emph{value-equality
            quotient of~$\OR$} as
        the directed graph $\mathrm{G}_\OR=(\ID', E)$, where
        \begin{itemize}
            \item $\ID'$ is the quotient of
            $\ID$ by the equivalence relation $\id_1\sim\id_2\Leftrightarrow
            T(id_1)=T(id_2)$, i.e., it is a set of equivalence classes that are subsets of~$\ID$;
          \item The edge set $E$ is defined by setting $(\id'_1, \id'_2) \in E$
            for $\id_1', \id_2' \in \ID'$ iff
            $\id_1' \neq
            \id_2'$ and there are $\id_1\in\id_1'$ and $\id_2\in\id_2'$ such
            that $\id_1<\id_2$.
        \end{itemize}
    \end{definition}

    We claim that cycles in the value-equality quotient of~$\OR$ precisely
    characterize complete failure of $\dupelim$.

    \begin{proposition}\label{prp:pocycle}
        For any po-relation $\OR$, $\dupelim(\OR)$ completely fails iff
        $\mathrm{G}_\OR$ has a cycle.
    \end{proposition}

    \begin{proof}
      Fix an input po-relation $\OR = (\ID, T, <)$.
        We first show that the existence of a cycle implies complete failure of
        $\dupelim$. Let $\id'_1,\dots,\id'_n,\id'_1$ be a simple cycle of
        $\mathrm{G}_\OR$.
        For all $1\leq i\leq n$, there exist $\id_{1i},\id_{2i}\in\id'_1$ such
        that $\id_{2i}<\id_{1(i+1)}$
        (with the convention $\id_{1(n+1)}=\id_{11}$)
        and the $T(\id_{2i})$ are pairwise distinct.

        Let $L$ be a possible world of $\OR$ and let us show that $\dupelim$
        fails on any po-relation $\OR_L$ that \emph{represents}~$L$, i.e.,
        $\OR_L = (\ID_L, T_L, {<_L})$ is totally ordered and $\pw(\OR_L) = \{L\}$.
        Assume by contradiction that
        for all $1\leq i\leq n$, $\id'_i$ forms an id-set of $\OR_L$. Let us show
        by induction on~$j$ that for all $1\leq j\leq n$,
        $\id_{21}\leq_{L}\id_{2j}$, where $\leq_L$ denotes the non-strict order
        defined from $<_L$ in the expected fashion. The base case is trivial. Assume this holds
        for $j$ and let us show it for $j+1$. Since $\id_{2j}<\id_{1(j+1)}$, we
        have $\id_{21}\leq\id_{2j}<_L\id_{1(j+1)}$. Now, if $\id_{2(j+1)}<_{L}\id_{21}$,
        then $\id_{2(j+1)}<_{L}\id_{21}<_L\id_{1(j+1)}$ with
        $T(\id_{2(j+1)})=T(\id_{1(j+1)})\neq T(\id_{21})$, so this contradicts the fact that
        $\id'_{j+1}$ is an id-set. Hence, as $L$ is a total order, we must have
        $\id_{21} \leq_L \id_{2(j+1)}$, which
        proves the induction case. Now the claim proved by induction implies that
        $\id_{21}\leq_{L}\id_{2n}$, and we had $\id_{2n}<\id_{11}$ in~$\OR$ and
        therefore $\id_{2n}<_L\id_{11}$, so this contradicts the fact that
        $\id'_1$ is an id-set. Thus, $\dupelim$ fails in $\OR_L$. We have thus shown that
        $\dupelim$ fails in every possible world of~$\OR$, so that it completely fails.

        \medskip

        Conversely, let us assume that $\mathrm{G}_\OR$ is acyclic. Consider a
        topological sort of $\mathrm{G}_\OR$ as $\id'_1,\dots,\id'_n$. For $1\leq
        j\leq n$, let $L_j$ be a linear extension of the poset
        $(\id'_j,\restr{<}{\id'_j})$. Let $L$ be the concatenation of $L_1,\dots L_n$.
        We claim $L$ is a linear extension of $\OR$ such that $\dupelim$ does not
        fail in~$\OR_L = (\ID_L, T_L, {<_L})$; this latter fact is clear by construction of $L$, so we must only show
        that $L$ obeys the comparability relations of~$\OR$. Now, let
        $\id_1<\id_2$ in $\OR$. Either for some $1\leq j\leq n$ we have
        $\id_1, \id_2\in\id'_j$, and then the tuple for $\id_1$ precedes the one
        for $\id_2$ in~$L_j$ by construction, so
        we have $t_1<_L t_2$; or they are in different classes $\id'_{j_1}$ and
        $\id'_{j_2}$
        and this is
        reflected in $\mathrm{G}_\OR$, which means that $j_1<j_2$ and
        $\id_1<_L \id_2$. Hence, $L$ is a linear extension, which concludes the proof.
    \end{proof}
    We can now state and prove the result.

\begin{theorem}\label{thm:duelim-por}
    For any po-relation $\OR$, we can test in PTIME if
    $\dupelim(\OR)$ completely fails; if it does not, then
    we can compute in PTIME a po-relation $\OR'$ such that
    $pw(\OR')=\dupelim(\OR)$.
\end{theorem}

\begin{proof}
    We first observe that $\mathrm{G}_\OR$ can be constructed in PTIME,
    and that testing that $\mathrm{G}_\OR$ is acyclic is also done in
    PTIME. Thus, using Proposition~\ref{prp:pocycle}, we can determine
    in PTIME whether $\dupelim(\OR)$ fails.

    If $\dupelim(\OR)$ does not fail, then we let $\mathrm{G}_\OR=(\ID',E)$ and construct the
    relation $\OR'$ that will stand for $\dupelim(\OR)$ as
    $(\ID',T',<')$, where  $T'(\id')$ is the unique $T'(\id)$ for
    $\id\in\id'$ and $<'$ is the transitive closure of $E$, which is
    antisymmetric because $\mathrm{G}_\OR$ is acyclic. Observe that
    the underlying bag relation of $\OR'$ has one identifier for each distinct
    tuple value in $\OR$, but has no duplicates.

    Now, it is easy to check that $\pw(\OR')=\dupelim(\OR)$. Indeed, any
    possible world $L$ of~$\OR'$ can be achieved in $\dupelim(\OR)$ by
    considering, as in the proof of Proposition~\ref{prp:pocycle}, some
    possible world of $\OR$ obtained following the topological sort of
    $\mathrm{G}_\OR$ defined by $L$. This implies that $\pw(\OR')
    \subseteq \dupelim(\OR)$.

    Conversely, for any possible world $L$ of $\OR$, $\dupelim(\OR_L)$ (for
    $\OR_L$ a po-relation that represents $L$) fails
    unless, for each tuple value, the occurrences of that tuple value in
    $\OR_L$ is an id-set. Now, in such an $L$, as the occurrences of each
    value are contiguous and the order relations reflected in
    $\mathrm{G}_\OR$ must be respected, $L$ is defined by a topological
    sort of $\mathrm{G}_\OR$ (and some topological sort of each id-set
    within each set of duplicates), so that $\dupelim(\OR_L)$ can also be
    obtained as the corresponding linear extension of $\OR'$. Hence, we
    have $\dupelim(\OR) \subseteq \pw(\OR')$, proving their equality and
    concluding the proof.
\end{proof}

Last, we observe that $\dupelim$ can indeed be used to undo some of the effects
of bag semantics.

\begin{proposition}\label{prp:dupunion}
    For any po-relation $\OR$, we have $\dupelim(\OR \cup \OR) = \dupelim(\OR)$:
    in particular, one completely fails iff the other does.
\end{proposition}

\begin{proof}
    Let $G_\OR$ be the value-equality quotient of $\OR$ and $G'_\OR$ be the
    value-equality quotient of $\OR \cup \OR$. It is easy to see that these two
    graphs are identical: any edge of~$G_\OR$ witnesses the existence of the same
    edge in~$G'_\OR$, and conversely any edge in $G'_\OR$ must correspond to a
    comparability relation between two tuples of one of the copies of~$\OR$ (and
    also in the other copy), so
    that it also witnesses the existence of the same edge in~$\OR$. Hence,
    by Proposition~\ref{prp:pocycle}, 
    one
    duplicate elimination operation completely fails iff the other does.
    Further, by Theorem~\ref{thm:duelim-por}, we have
    indeed the equality that we claimed.
\end{proof}

We can also show that most of our previous tractability results 
Sections~\ref{sec:posscert}--\ref{sec:fpt}
still apply when
the duplicate elimination operator is added. We first clarify the semantics of query evaluation when complete failure
    occurs: given a query~$Q$ in \PosRA{} extended with $\dupelim$, and given a po-database $D$,
    if complete failure occurs at any
    occurrence of the $\dupelim$ operator when evaluating $Q(D)$, then we set
    $\pw(Q(D)) \defeq \emptyset$,
    pursuant to our choice of defining query evaluation on po-relations as
    yielding all possible results on all possible worlds.
    If $Q$ is
    a \posRAagg{} query extended with $\dupelim$, we likewise say that its
    possible accumulation results are $\emptyset$.

    This implies that for any \PosRA query $Q$ extended with $\dupelim$,
    for any input po-database $D$, and
    for any candidate possible world $v$, the \poss and \cert problems for $Q$ are
    vacuously false on instance $(D, v)$
    if complete failure occurs at any stage when evaluating $Q(D)$. The same holds
    for \PosRAacc queries.

\begin{theorem}\label{thm:easypossdupelim}
  Theorems~\ref{thm:aggregw}, \ref{thm:certaintyptimec}, \ref{thm:aggregwa} and
  Proposition~\ref{prop:otherdef}
  extend to \PosRA and \PosRAagg where we allow $\dupelim$ (but impose
  the same restrictions on query operators, accumulation, and input
  po-relations). Specifically:
  \begin{itemize}
    \item For any fixed $k\in\NN$ and fixed \Plex query~$Q$ which may
      additionally use $\dupelim$, the \poss problem for~$Q$ in in PTIME on
      po-databases of bounded width.
    \item For any \PosRAagg query~$Q$ which may additionally use $\dupelim$ and
      where accumulation is performed in a cancellative monoid, the \cert
      problem for~$Q$ is in PTIME.
    \item For any \Plexacc query~$Q$ which may additionally use $\dupelim$ and
      where the accumulation operator is finite, the \poss and \cert problems
      are in PTIME on po-databases of bounded width.
    \item For any \PosRA query which may additionally use the $\dupelim$
      operator, the problems \textbf{select-at-k}, \textbf{top-k}, and
      \textbf{tuple-level comparison} are in PTIME.
  \end{itemize}
\end{theorem}

To prove this result, observe that these four results are proved by first
    evaluating the query result in PTIME using Proposition~\ref{prp:repsys}. So we
    can still evaluate the query in PTIME, using in addition
    Theorem~\ref{thm:duelim-por}. Either complete failure occurs at some point in
    the evaluation, and we can
    immediately solve \poss and \cert by our initial remark above, or no complete failure occurs and we obtain
    in PTIME a po-relation $\OR$ on which to solve \poss and \cert. Hence, in what
    follows, we can assume that no complete failure occurs at any stage.

    It is then immediate that
    Theorem~\ref{thm:certaintyptimec} and Proposition~\ref{prop:otherdef} still
    apply, because they did not make any assumptions on the po-relation~$\OR$ on
    which they applied. As for Theorems~\ref{thm:aggregw}
    and~\ref{thm:aggregwa}, the only assumption that they made on~$\OR$ is that
    its width was constant. Hence, we can conclude the proof of Theorem~\ref{thm:easypossdupelim}
    from the following width preservation result.

    \begin{lemma}
        \label{prp:wdupe}
        For any constant $k \in \NN$ and po-relation $\OR$ of width $\leq k$, if
        $\dupelim(\OR)$ does not completely fail, then it has width $\leq k$.
    \end{lemma}

    \begin{proof}
        It suffices to show that to every antichain $A$ of
        $\dupelim(\OR)$, there is an
        antichain $A'$ of the same cardinality in $\OR$. Construct $A'$ by picking a
        member of each of the classes of~$A$. Assume by contradiction that $A'$ is
        not an antichain, hence, there are two tuples $t_1 < t_2$ in~$A'$, and
        consider the corresponding classes $\id_1$ and $\id_2$ in $A$. By our
        characterization of the possible worlds of $\dupelim(\OR)$
        in the proof of Theorem~\ref{thm:duelim-por}
        as obtained from the topological sorts
        of the value-equality quotient $\mathrm{G}_\OR$ of $\OR$, as $t_1 < t_2$ implies
        that $(\id_1, \id_2)$ is an edge of $\mathrm{G}_\OR$, we conclude that we have $\id_1
        < \id_2$ in $A$, contradicting the fact that it is an antichain.
    \end{proof}

We have just shown in Theorem~\ref{thm:easypossdupelim} that our tractability
results still apply when we allow the duplicate elimination operator.
Furthermore, if in a set-semantics spirit we {\em require} that the
query output has no duplicates, \poss and \cert are always tractable (as this
avoids the technical difficulty of Example~\ref{exa:notposet}).

\begin{theorem}\label{thm:posscertnodupes}
    For any \PosRA query $Q$, \poss and \cert for $\dupelim(Q)$
    are in PTIME.
\end{theorem}

\begin{proof}
    Let $D$ be an input po-relation, and $L$ be the candidate possible world
    (a list relation).
    We compute the po-relation $\OR'$ such that $\pw(\OR')= Q(D)$ in PTIME using
    Proposition~\ref{prp:repsys}
    and the po-relation $\OR \defeq \dupelim(\OR')$ in PTIME using
    Theorem~\ref{thm:duelim-por}. If duplicate elimination fails, then we vacuously
    reject for \poss and \cert.
    Otherwise, by the definition of $\dupelim$, the resulting
    po-relation~$\OR$ is such that each tuple
    value is realized exactly once. Note that we can
    reject immediately if $L$ contains multiple occurrences of the same tuple, or
    does not have the same underlying set of tuples as~$\OR$; so we assume that $L$ has the
    same underlying set of tuples as~$\OR$ and no duplicate tuples.

    The \cert problem is in PTIME on $\OR$ by
    Theorem~\ref{thm:certaintyptimec}, so
    we need only study the case of \poss, namely, decide whether $L \in
    \pw(\OR)$.
    Let $\OR_L$ be a po-relation that represents~$L$.
    As $\OR_L$ and $\OR$ have no duplicate tuples, there is only one way to
    match each identifier of $\OR_L$ to an identifier of $\OR$.
    Build $\OR''$ from $\OR$ by adding, for each pair $\id_i <_L \id_{i+1}$ of consecutive
    tuples of~$\OR_L$, the order constraint $\id_i'' {<''} \id''_{i+1}$ on the
    corresponding identifiers in $\OR''$. We claim that $L \in \pw(\OR)$ iff the
    resulting $\OR''$
    is a po-relation, i.e., its transitive closure is still
    antisymmetric, which can be tested in PTIME by computing the strongly connected
    components of $\OR''$ and checking that they are all trivial.

    To see why this
    works, observe that, if the result $\OR''$ is a po-relation, it is a total
    order, and so it describes a way to achieve $L$ as a linear extension of
    $\OR$ because it does not contradict any of the comparability relations
    of $\OR$.
    Conversely, if $L \in \pw(\OR)$, assuming to the contrary the existence of a
    cycle in $\OR''$, we observe that such a cycle must consist of order relations
    of $\OR$ and $\OR_L$, and the order relations of $\OR$ are reflected in
    $\OR_L$ as it is a
    linear extension of $\OR$, so we deduce the existence of a cycle in $\OR_L$,
    which is impossible by construction. Hence, we have reached a contradiction,
    and we deduce the desired result.
\end{proof}

\subparagraph*{Discussion.}  The introduced group-by and duplicate elimination operators have some  shortcomings: the result of group-by is in general not representable by po-relations, and duplicate elimination may fail. These are both consequences of our design choices, where we capture only uncertainty on order (but not on tuple values) and design each operator so that its result corresponds to the result of applying it to each individual world of the input (see further discussion in Section \ref{sec:compare}). Avoiding these shortcomings is left for future work.

\section{Comparison With Other Formalisms}\label{sec:compare}

We next compare our formalism to previously proposed formalisms:
query languages over bags (with no order); a query language for
partially ordered multisets; and other related work. 
To our knowledge, however, none of these works studied the possibility or
certainty problems for partially ordered data, so that our technical results do not follow from them.

\subparagraph*{Standard bag semantics.}
A natural desideratum for our semantics on
(partially) ordered relations is that it should be a faithful extension
of the bag semantics for relational algebra. We first consider the $\BALG^{1}$ language on bags
\cite{GrumbachMiloBags} (the ``flat fragment'' of their language $\BALG$ on
nested relations).
We denote by $\BALG_{+}^{1}$ the fragment of~$\BALG^{1}$
that includes the standard extension of positive relational algebra
operations to bags: additive union, cross product, selection,
and projection.
We observe that, indeed, our semantics
faithfully extends $\BALG_{+}^{1}$:
\emph{query
evaluation commutes with ``forgetting'' the order}. Formally, for a po-relation $\OR$, we
denote by $\bag(\OR)$ its underlying bag relation, and define likewise 
$\bag(D)$ for a po-database~$D$ as the database of the underlying bag relations.
For the
following comparison, we identify both $\times_{\dir}$ and
$\times_{\lex}$ with the $\times$ of~\cite{GrumbachMiloBags} 
(as both our product operations yield the same bag as output, for
    any input), 
and we identify our
union with the additive union of~\cite{GrumbachMiloBags}. The following then trivially holds.
\begin{proposition}\label{prop:faithful}
    For any \PosRA query $Q$ and a po-relation $D$, $\bag(Q(D))=Q(\bag(D))$, where
$Q(D)$ is defined according to our semantics and $Q(\bag(D))$ is defined by
$\BALG_{+}^{1}$.
\end{proposition}

\begin{proof}
    There is an exact correspondence in terms of the output bags between
    additive union and our union; between cross product and $\times_{\dir}$ and
    $\times_{\lex}$; between our selection and that of $\BALG_{+}^{1}$, and similarly
    for projection (as noted before the statement of
    Proposition~\ref{prop:faithful} in the main text, a technical subtlety is that the projection
    of $\BALG$ can only project on a single attribute, but one can encode
    ``standard'' projection on multiple attributes).
     The proposition follows by induction on the query structure.
\end{proof}

The full $\BALG^{1}$ language includes additional operators such as bag
intersection and subtraction, which are non-monotone and as such may not be
expressed in our language: it is also unclear how they could be extended to our
setting (see further discussion in ``Algebra on pomsets'' below).
On the other hand, $\BALG^{1}$ does not include
aggregation, and so \posRAagg and $\BALG^{1}$ are incomparable in terms of
expressive power.

A better yardstick to compare against for accumulation could be the work of~\cite{libkin1997query}:
they show that their basic language $\BQL$ is equivalent
to $\BALG$, and then further extend the language with aggregate operators, to
define a language called $\NRLaggr$ on nested relations.
On flat relations,
$\NRLaggr$
captures
functions that cannot be captured in our
language: in particular the average function \textsf{AVG} is non-associative and thus cannot be captured by
our accumulation function (which anyway focuses on order-dependent functions, as
\poss/\cert are trivial otherwise).
On the other hand, $\NRLaggr$ cannot test 
parity 
(Corollary~5.7 in \cite{libkin1997query}) whereas this is easily captured by our accumulation operator.
We conclude that $\NRLaggr$ and \posRAagg are incomparable in terms of captured transformations on bags, even when restricted to flat relations.

\subparagraph*{Algebra on pomsets.} We now compare our work to algebras defined
on \emph{pomsets}~\cite{grumbach1995algebra,grumbach1999algebra}, which also 
attempt to bridge partial order theory and data management (although, again,
these works do not
study possibility and certainty). 
\emph{Pomsets} are labeled posets quotiented by
isomorphism (i.e., renaming of identifiers), like po-relations.
Beyond similarities in the language design, a
major
conceptual difference between our formalism and that
of~\cite{grumbach1995algebra,grumbach1999algebra} is that their work focuses on
processing {\em connected components} of the partial order graph,
and their operators are tailored for that semantics. As a
consequence, their semantics is {\em not} a faithful extension of bag
semantics, i.e., their language would not satisfy the counterpart of
Proposition~\ref{prop:faithful} (see,
for instance, the semantics of duplicate elimination in \cite{grumbach1995algebra}). 
By contrast, we manipulate po-relations that stand for sets of possible list
relations, and our operators are designed accordingly, unlike those of
\cite{grumbach1995algebra}, where transformations take into account the structure
(connected components) of the entire poset graph.
Because of this choice, \cite{grumbach1995algebra} introduces
non-monotone operators that we cannot express, and 
can design a duplicate elimination operator that cannot fail. Indeed, the possible failure of our duplicate elimination operator is a direct consequence of its semantics of operating on each possible world, possibly leading to contradictions.

If we consequently disallow duplicate elimination in both languages for the sake
of comparison, then the resulting fragment $\PomAlgEps$ of the language
of \cite{grumbach1995algebra} can yield only series-parallel outputs
(Proposition~4.1 of~\cite{grumbach1995algebra}), unlike \PosRA
queries whose output order may be arbitrary. To formalize this, we need the
notion of a \emph{realizer}~\cite{schroder2003ordered} of a poset
$P = (V, <)$: this is a set of total orders
$(V, {<_1}), \ldots, (V, {<_n})$ such that, for every $x, y \in V$, we have $x
< y$ iff $x <_i y$ for all~$i$. We can use realizers to express arbitrary
po-relations using the $\times_\dir$-product, as is shown by rephrasing in our
context an existing result on partial orders (Theorem~9.6 of~\cite{hiraguchi1955dimension},
see also \cite{ore1962theory}).

\begin{lemma}
  \label{lem:dimprodoneway}
  Let $n \in \NN$, and let
  $(P, <_P)$ be a poset that has a realizer $(L_1, \ldots, L_n)$ of size~$n$.
  Then $P$ is isomorphic to a subset $\OR'$ of $\OR = \ordern{l} \times_{\gen}
  \cdots \times_{\gen} \ordern{l}$, with $n$ factors in the product, for some integer $l \in \mathbb{N}$ (the
  order on $\OR'$ being the restriction on that of $\OR$).
\end{lemma}

\begin{proof}
  We define $\OR$ by taking $l \defeq \card{P}$,
  and we
  identify each
  element $x$ of $P$ to $f(x) \defeq (n_1^x, \ldots, n_n^x)$, where
$n^x_i$ is
  the position where $x$ occurs in $L_i$. Now, for any $x, y \in P$, we have $x <_P y$ iff $n_i^x < n_i^y$ for
  all $1 \leq i \leq n$ (that is, $x <_{L_i} y$), hence iff
  $f(x) <_\OR f(y)$: this is because there are no two elements $x \neq y$
  and $1 \leq i \leq n$
  such that the $i$-th components of $f(x)$ and of $f(y)$ are the same.
  Hence, taking $\OR'$ to be the image of $f$ (which is
  injective), $\OR'$ is indeed isomorphic to $P$.
\end{proof}

This implies that \PosRA queries can yield arbitrary po-relations as output.

\begin{proposition}\label{prp:gen}
  For any po-relation $\OR$, there is a \PosRA query $Q$ with no inputs
  such that $Q() = \OR$.
\end{proposition}

\begin{proof}
  We first show that for any
  poset $(P, <)$, there exists a \Pgen query $Q$
  such that the tuples of $\OR' \defeq Q()$ all have unique
  values and the underlying poset of $\OR'$ is $(P, <)$.
  Indeed, we can take $d$ to be the \emph{order dimension} of
  $P$, which is necessarily finite~\cite{schroder2003ordered}, and then, by
  definition, $P$ has a realizer of size~$d$. By
  Lemma~\ref{lem:dimprodoneway}, there is an integer
  $l \in \mathbb{N}$ such that $\OR'' \defeq \ordern{l} \times_\gen \cdots
  \times_\gen \ordern{l}$ (with $n$ factors in the product)
  has a subset $S$ isomorphic to $(P, <)$. Hence, letting $\psi$ be a tuple predicate
  such that $\sigma_\psi(\OR'') = S$ (which can clearly be constructed by
  enumerating the elements of $S$), the query $Q' \defeq \sigma_\psi(\OR'')$ proves the
  claim, with $\OR''$
  expressed as above.

  Now, to prove the desired result from this claim, build $Q$ from $Q'$ by
  taking its join (i.e., $\times_\lex$-product, selection,
projection) with a union of singleton
  constant expressions that map each unique tuple value of $Q'()$ to the desired
  value of the corresponding tuple in the desired po-relation~$\OR$. This
  concludes the proof.
\end{proof}

We conclude that $\PomAlgEps$ does not subsume \PosRA.

\label{sec:related}
\subparagraph*{Incompleteness in databases.}
Our work is inspired by the field of incomplete information management, which
has been studied for
various models~\cite{Incompletexml,Libkin06}, in particular
relational databases~\cite{IL84}. This field inspires our design
of po-relations
and our study of
possibility and certainty~\cite{KO06,lipski1979semantic}. However, uncertainty in
these settings typically focuses on \emph{whether} tuples exist
or on what their \emph{values} are (e.g., with nulls~\cite{codd1979extending},
including the novel approach of~\cite{libkin2014incompleteness,libkin2015icdt};
with c-tables~\cite{IL84}, probabilistic
databases~\cite{probdbbook} or fuzzy numerical values as
in~\cite{soliman2009ranking}). To our knowledge, though, our work is the first to study possible and
certain answers in the context of 
{\em
order}-incomplete data. Combining order incompleteness with standard
tuple-level uncertainty is left as a challenge for future work.
Note that some works~\cite{BunemanJO91,libkin1998semantics,libkin2015icdt} use
partial orders on \emph{relations} to compare the
informativeness of representations. This is
unrelated to our partial orders on \emph{tuples}.

\subparagraph*{Ordered domains.} Another line of work has studied
relational data management where the \emph{domain elements} are
(partially) ordered \cite{Immermanptime,Ng,van1997complexity,BonnerM98,CoburnW93}. This is in particular the case in works aiming to query sequences or support iteration (see e.g. \cite{BonnerM98,CoburnW93}, which however do not consider uncertainty and consequently neither partial order). However, our goal, setting and perspective
are different: we see order on tuples as part of the relations, and
as being constructed by applying our operators; these works see
order as being given \emph{outside} of the query, hence do not study the propagation of uncertainty through queries.
Also, queries in such works can often directly access the order
relation ~\cite{van1997complexity,BenS2009,BonnerM98}. Some works also study uncertainty on totally ordered
\emph{numerical}
domains~\cite{soliman2009ranking,soliman2010supporting}, while
we look at general order relations.

\subparagraph*{Temporal databases.}
\emph{Temporal
databases}~\cite{chomicki2005time,snodgrass2000developing} consider
order on facts, but it is usually induced by timestamps, hence
total. A notable exception is~\cite{fan2012determining} which
considers that some facts may be \emph{more current} than others,
with constraints leading to a partial order. In particular, they
study the complexity of retrieving query answers that are certainly
current, for a rich query class. In contrast, we can
{\em manipulate} the order via queries, and we can also ask
about aspects beyond currency, as shown throughout the paper
(e.g., via accumulation).

\subparagraph*{Using preference information.} Order theory has been also used to handle \emph{preference information} in
database systems~\cite{jacob2014system,arvanitis2014preferences,
kiessling2002foundations,alexe2014preference, stefanidis2011survey},
with some operators being the same as ours, and for \emph{rank
aggregation}~\cite{Fagin,jacob2014system,dwork2001rank}, i.e., retrieving top-$k$ query answers given multiple
rankings. However, such works typically try to \emph{resolve}
uncertainty by reconciling many conflicting representations (e.g.,
via knowledge on the individual scores given by different sources
and a function to aggregate them \cite{Fagin}, or a preference
function~\cite{alexe2014preference}). In contrast, we focus on maintaining a faithful model of
\emph{all} possible worlds without reconciling them, studying possible and certain answers in this respect.

\subparagraph*{Computational social choice.}
The notion of preferences has been studied in the domain of computational social
choice to determine the possible outcomes of an election given partial preference
information expressed by voters~\cite{konczak2005voting}. In this setting, the notions of \emph{possible
winners} and \emph{necessary winners} have been introduced to summarize the
possible outcomes, and they have been connected to the notion of possible and
necessary answers of database queries~\cite{kimelfeld2018computational}. The
complexity of these problems has been studied, with a dichotomy result
that classifies its complexity depending on the aggregation used~\cite{konczak2005voting,xia2011determining,betzler2010towards,baumeister2012taking}.
However, the expressiveness of this computational social choice framework is
incomparable to that of our
framework.
Specifically, their framework
only studies possible and necessary answers in terms of achieving a maximal
score computed as a sum of numerical values following some positional scoring
rule, whereas our framework can perform accumulation in arbitrary monoids and on
top of positive relational algebra queries. Conversely, there is no apparent way in our
framework to encode accumulation following positional scoring rules, as we would need
to apply accumulation to sum the candidate scores, and then look at the top
answers according to a different order on the result.

\section{Conclusion}
\label{sec:conclusion}
This paper introduced an algebra for order-incomplete data. We have studied the complexity of possible and certain
answers for this algebra, have shown the problems to be generally intractable, and identified
several tractable cases. 

A prime motivation for our work is to provide a semantics for a fragment of SQL
(namely, SPJU+aggregates) in presence of partially ordered data. We see our work as a first step in this respect, and our choice of operators for the algebra is by no means the only possible one. In future work we plan to study the incorporation of additional operators, including in particular constructors of the (partial) order based on 
the tuple values. We will also investigate how to combine order-uncertainty with
uncertainty on values, and study additional semantics for $\dupelim$ (to avoid
the pitfalls of the proposed semantics which we discussed above). 

In connection with the choice of operators, a natural question is  
whether one may achieve a completeness result. We have shown
(Proposition~\ref{prp:gen}) that our language is complete in
terms of ``individual outputs'', i.e., that PosRA can be used to construct any
po-relation using only the built-in constant relations and operators. A more challenging goal is to design a language that is complete in terms of transformations, i.e. that may capture all functions over po-relations in some class. This is another intriguing topic for further investigation.

Last, many open questions remain about the complexity of \poss, e.g., we do not
know whether \poss is tractable when the accumulation monoid is a finite group.
Ideally, we would want to establish a dichotomy
result for the complexity of \poss, and a
complete syntactic characterization of cases where \poss is tractable: this is
investigated further in a follow-up work involving the first author
\cite{amarilli2018topological}.

\subparagraph*{Acknowledgments.} We are grateful to Marzio De
Biasi, to P\'alv\"olgyi D\"om\"ot\"or, and to Mikhail Rudoy, from
\url{cstheory.stackexchange.com}, for helpful suggestions. We are also
grateful to the anonymous reviewers for their feedback that helped
improve this paper.
This research was partially supported by the Israeli Science Foundation
(grant 1636/13), the Blavatnik~ICRC, and Intel.

\bibliographystyle{abbrv}
\bibliography{main}

\end{document}